\newif\ifarxiv\arxivtrue%
\documentclass[runningheads]{llncs}
\usepackage[paperheight=235mm,paperwidth=155mm,textwidth=12.2cm,textheight=19.3cm]{geometry}
\usepackage{silence}
\usepackage{etoolbox}
\apptocmd{\sloppy}{\hbadness 10000\relax}{}{} 

\newcommand\nempty{\textrm{\tiny{\sf ne}}}

\usepackage{booktabs}   
\usepackage{subcaption} 

\makeatletter
\RequirePackage[bookmarks,unicode,colorlinks=true]{hyperref}%
   \def\@citecolor{blue}%
   \def\@urlcolor{blue}%
   \def\@linkcolor{blue}%

\def\orcidID#1{\smash{\href{http://orcid.org/#1}{\protect\raisebox{-1.25pt}{\protect\includegraphics{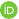}}}}}
\makeatother

\usepackage{local}
\usepackage[all]{xy}

\usepackage{bbding}

\begin{document}

\title{Concurrent NetKAT}         
\subtitle{Modeling and analyzing stateful, concurrent networks}
\titlerunning{Concurrent NetKAT}


\author{
    Jana Wagemaker\inst{1}~\orcidID{0000-0002-8616-3905}~\Envelope \and 
    Nate Foster\inst{2}~\orcidID{0000-0002-6557-684X} \and 
    Tobias Kapp\'{e}\inst{3}~\orcidID{0000-0002-6068-880X} \and\\ 
    Dexter Kozen\inst{2}~\orcidID{0000-0002-8007-4725} \and 
    Jurriaan~Rot\inst{1} \and
    Alexandra Silva \inst{2}~\orcidID{0000-0001-5014-9784} 
}

\authorrunning{J. Wagemaker et al.}
%
\institute{
    Radboud University, Nijmegen, The Netherlands \\
    \email{Jana.Wagemaker@ru.nl}
    \and
    Cornell University, Ithaca, New York, USA
    \and
    ILLC, University of Amsterdam, The Netherlands
}

\maketitle

\begin{abstract}
We introduce Concurrent \netkat (\cnetkat), an extension of \netkat with operators for specifying and reasoning about concurrency in scenarios where multiple packets interact through state. We provide a model of the language based on partially-ordered multisets (pomsets), which are a well-established mathematical structure for defining the denotational semantics of concurrent languages. We provide a sound and complete axiomatization of this model, and we illustrate the use of \cnetkat through examples. More generally, \cnetkat can be understood as an algebraic framework for reasoning about programs with both local state (in packets) and global state (in a global store). 
\end{abstract}


\section{Introduction}
Kleene algebra (\KA) is a well-studied formalism~\cite{kozen94,krob-1990,salomaa-1966,conway-1971} for analyzing and verifying imperative programs. Over the past few decades, various extensions of \KA have been proposed for modeling increasingly sophisticated scenarios.
For example, Kleene algebra with tests (\kat)~\cite{kozen-1996} models conditional control flow while \netkat~\cite{netkat,netkat2} models behaviors in packet-switched networks.

A key limitation of \netkat, however, is that the language is stateless and sequential. It cannot model programs composed in parallel, and it offers no way to reason algebraically about the effects induced by multiple concurrent packets. Meanwhile, the software-defined networking (SDN) paradigm has evolved to include richer functionality based on stateful processing including data aggregation and dynamic routing. In languages like P4~\cite{p4}, issues of concurrency arise because the semantics depends on the order that packets are processed.

Given this context, it is natural to wonder we can add concurrency to \netkat while retaining the elegance of the underlying framework. In this paper, we answer this question in the affirmative, by developing \cnetkat. However, to do this, we must overcome several challenges.
A first hurdle is that networks exhibit many different forms of concurrent behavior.
The most obvious source of concurrency arises when multiple packets are processed by different devices. In these situations, certain packets may cause changes in forwarding behavior by modifying global state variables on switches.
However, there is also concurrency within individual devices: a high-speed switching chip often has
multiple pipelines, each with multiple stages of match-action tables and stateful registers. The tables can be programmed to act concurrently on (parts of) a single packet, and the pipelines also act concurrently on multiple packets.

Another hurdle is that it is not entirely clear how to simultaneously extend \KA with networking features and concurrency. Orthogonal to the development of \netkat, the issue of adding concurrency to \KA has been researched extensively, starting with concurrent Kleene algebra
(\CKA)~\cite{hoare-moeller-struth-wehrman-2009,laurence-struth-2014,laurence2017completeness,cka}. However, the combination of concurrency from \CKA
and tests from \kat is not straightforward---see, e.g.~\cite{jipsen-moshier-2016,kao,fossacs2020}---which motivated the development of partially-observable concurrent Kleene algebra (\POCKA)~\cite{pocka}.
In \POCKA, a single thread only has \emph{partial} view of the state.
Hence, when evaluating control guards, a thread makes \emph{observations} about the machine state, rather than definitive tests.
This allows for fine-grained reasoning about concurrent programs with variables, conditionals, loops, and imperative statements that manipulate a shared global memory.

In this work, we use \POCKA as a basis for designing a
language with state and concurrent threads, which we combine with a multi-packet
extension of \netkat. The resulting language, Concurrent NetKAT (\cnetkat), models the
behavior of packets in a network that communicate through a shared global state,
and addresses the fundamental and non-trivial question of how to combine concurrency and the interaction between local and global state within \KA. 
%
%

Overall, the contributions of the paper are as follows:
\begin{enumerate}
\item We present the design of the \cnetkat language (\Cref{sec:cnetkat}). The semantics combines the language models of \netkat and \POCKA, incorporating pomsets that record the evolution of the global state (as in \POCKA) as well as sets of (output) packets (as in \netkat).

\item We develop a sound and complete axiomatization of \cnetkat(\Cref{sec:completeness}).

\item We illustrate the applicability of \cnetkat for modeling and analyzing concurrent network behaviors through case studies and examples (\cref{sec:overview} and \cref{sec:example}).
\end{enumerate}

The next section contains an overview of the challenges in the design of extending \netkat with multiple packets, global state, and concurrency, as well as a glimpse of how to use the language in a practical example.


\section{Overview}\label{sec:overview}
\cnetkat models the behavior of two basic entities: the packets being routed through the network, and a global store, which may be accessed by the network as it processes the packets.
These elements give rise to two kinds of basic programs.
On the one hand, \emph{basic packet programs}---imported from \netkat~\cite{netkat}---include tests ($\match{f_i}{n}$) and modifications ($\modify{f_i}{n}$) of packet fields $f_1, \dots, f_N$.
Examples of fields are $\sw$, denoting the switch of the packet in the network, and $\type$, denoting the type of a packet.
In general, we expect  packets to have fields for a collection of standard attributes; unused fields may be populated with a dummy value.

On the other hand, \emph{basic state programs} include observations\footnote{Intuitively, these are tests on the state that can be understood as observing the part of the global state containing the variable, hence the terminology.} ($\match{v_i}{n}$), modifications ($\modify{v_i}{n}$) and a copy operation ($\modify{v_i}{v_j}$) on state variables $v_1, \dots, v_M$.
It will always be clear from context whether an action concerns a state or field variable.
\cnetkat also includes a primitive program $a$ for any set of packets $a$, which is useful for specifying the set of packets currently being processed.

\begin{remark}
We could augment the set of primitives with features such as general expressions in assignments. However, to keep things simple, we will only consider these primitives, which are already rich enough to describe non-trivial behaviors.
\end{remark}
%

\cnetkat programs are composed using sequential composition (`$\mathbin{;}$'), iteration (`$*$'), and non-deterministic choice (`$+$'), similar to \netkat.
In addition, \cnetkat programs may use the parallel composition operator (`$\parallel$').

The full syntax of \cnetkat is given in \cref{fig:syntax}.
Before giving a precise account of the semantics, we will go over some simple example programs.

\begin{figure}[t]
  \begin{minipage}{.3\textwidth}
 \includegraphics[scale=.15]{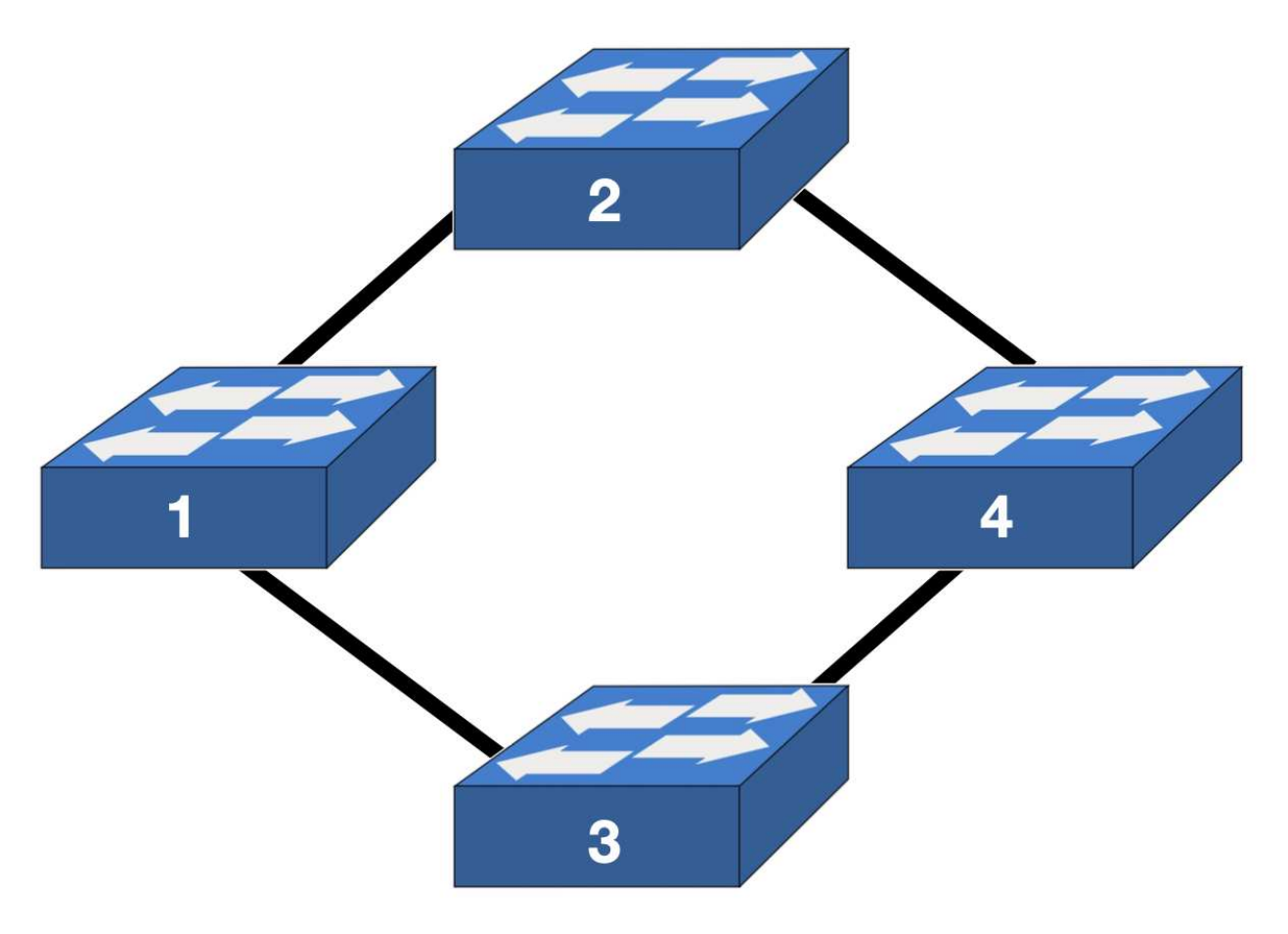}
 \end{minipage}
   \begin{minipage}{.67\textwidth}\small
   \[
   \begin{array}{ll}
  p_1 \defeq &\sw = 1 \mathbin{;} (({v=1} \mathbin{;} \type =\spadesuit \mathbin{;} \modify \sw  2 ) \\
  &\phantom{\sw = 1 \mathbin{;} \ \ }\parallel (\type =\heartsuit  \mathbin{;} \modify \sw 3 \mathbin{;} {\modify v 1}))\\
  p_2 \defeq & \sw =2 ; \modify \sw 4\\
  p_3 \defeq &\sw = 3; \modify \sw 4 \\
  p_4 \defeq &\sw=4\\ \\
  p \defeq & \modify v 0  \mathbin{;}  {(p_1 \parallel p_2 \parallel p_3 \parallel  p_4)}^*
  \end{array}
  \]
  \end{minipage}
\caption{Running example}\label{fig:running}
\end{figure}

\begin{example}[Packet forwarding]\label{ex:forwarding}
Consider the network depicted on the left in \Cref{fig:running}. Similar to \netkat, we assume packet movement and variable assignments are instantaneous.
Suppose there are two packet types: $\spadesuit$ and $\heartsuit$.
We want to write a program that transfers packets from node $1$ to node $4$ by sending $\spadesuit$ via node $2$, and $\heartsuit$ via node $3$.
The program running in switch 1 could be
\[
    p_1 := \match \sw 1 \mathbin{;} ((\match {\type }{\spadesuit} \mathbin{;} \modify \sw  2 ) \parallel (\match {\type }{\heartsuit}   \mathbin{;} \modify \sw 3))
\]
This program first filters out the packets at switch $1$.
Next, it launches two parallel threads, both of which receive a copy of the incoming packets.
The first thread filters out packets of type $\spadesuit$ and forwards them to switch $2$, while the second thread filters out packets of type $\heartsuit$, forwarding them to switch $3$.

We can write programs $p_2$, $p_3$ and $p_4$ for the other switches as well, and then compose all of those in parallel to obtain a program for the entire network.
\end{example}

\begin{remark}
Instant packet movement is not baked into \cnetkat, but rather a consequence of modeling packet location using the field $\sw$.
A more advanced model could use an additional field to mark a packet as being ``in-flight'' until it reaches the next hop.
Here, we opt for the simpler model.
\end{remark}

\begin{example}[Global behavior]
\cnetkat programs can read and write to a global store, letting earlier actions on packets affect later decisions.
For instance, suppose we need $\spadesuit$ packets to be forwarded only if a $\heartsuit$ packet already visited switch $3$.
We can use a global variable $v$ to implement this stateful behavior, writing:
\[
\match \sw 1 \mathbin{;} ((\textcolor{blue}{\match v 1} \mathbin{;} \match \type \spadesuit \mathbin{;} \modify \sw  2 )\parallel (\match \type \heartsuit  \mathbin{;} \modify \sw 3 \mathbin{;} \textcolor{blue}{\modify v 1}))
\]
We can program the other switches with $p_i$, as shown in \cref{fig:running}.
\end{example}

\begin{remark}[Concurrency and state]%
\label{rem:concurrency-and-state}
Actions involving global variables are more subtle than those that concern packet fields, due to concurrent threads accessing the global store.
For instance, we can write the program $\modify v 1 \mathbin{;} \match v 2$, which first sets $v$ to $1$ and then asserts that $v$ should have value $2$. This may seem inconsistent; however, there may be valid ways of executing this program if there are other threads that change the value of $v$ from $1$ to $2$ between the assignment $\modify v 1$ and the assertion $\match v 2$.
This possibility makes defining a compositional semantics somewhat tricky, as we will discuss below.
\end{remark}

\smallskip\noindent\textbf{Semantics of \cnetkat programs.}
A packet $\pk$ is a record of fields $f_1,\dots,f_N$.
We write $\pk(\sw)$ for the value of $\sw$ in $\pk$ and $\pk[1/\sw]$ for the packet obtained after updating the value of $\sw$ to $1$.
We denote the set of packets by $\Pk$.



\begin{figure}[t]
\hrule\vspace{1ex}
\begin{minipage}{.495\textwidth}\small
\textbf{Syntax}
\[
\begin{array}{@{~}rrl@{~}l}
\textrm{Values} & \Val &\ni n & ::= \mathrlap{0 \mid 1 \mid 2 \mid \cdots}\\
\textrm{Packet Fields} &\Field &\ni \field & ::= \mathrlap{\field_1 \mid \cdots \mid \field_k} \\
\textrm{Packets}  &\Pk &\ni \pk & ::= \mathrlap{\sset{\field_1=n_1, \dots }} \\
\textrm{Packet Sets} & 2^{\Pk} &\ni a,b  \\
\textrm{Packet} & \BB &\ni\preds & ::= \\ \textrm{predicates}
& \mid & \pfalse & \textit{False} \\
& \mid & \ptrue & \textit{True} \\
& \mid & \match{\field}{n} & \textit{Field Test} \\
& \mid & {\preda} \vee_{\BB} {\predb} & \textit{Disjunction} \\
& \mid & {\preda} \land_{\BB} {\predb} & \textit{Conjunction} \\
& \mid & \pnot{\preda} & \textit{Negation} \\[1ex]
\textrm{State} & \OO &\ni o,o' & ::= \\ \textrm{obs.}
& \mid & \cellcolor{LightYellow}\bot & \textit{Inconsistent} \\
& \mid & \cellcolor{LightYellow}\top & \textit{Neutral} \\
& \mid & \cellcolor{LightYellow}\match{\var}{n} & \textit{State test} \\
& \mid & \cellcolor{LightYellow}{o}\vee {o'} & \textit{Union} \\
& \mid & \cellcolor{LightYellow}{o}\land {o'} & \textit{Intersection} \\
& \mid & \cellcolor{LightYellow}\overline{o} & \textit{Complement}
\end{array}\]
\end{minipage}
\begin{minipage}{.495\textwidth}\small
\[
\begin{array}{@{~}rrl@{~}l}
\textrm{State Fields} & \Var &\ni \var  & ::= \mathrlap{\var_1 \mid \cdots \mid \var_i} \\
\textrm{Global State} &  \State&\ni \alpha,\beta & ::= \Var\rightharpoonup \Val\\
\textrm{State} &\Act & \ni e & ::= \\ \textrm{actions}
& \mid &\cellcolor{LightYellow} \modify{\var}{n} & \textit{Change} \\
& \mid &\cellcolor{LightYellow} \modify{\var}{\var'} & \textit{Copy} \\
\textrm{Programs} & \programs &\ni \pols & ::= \\
& \mid & \abort & \textit{Abort} \\
& \mid & \skp & \textit{Skip} \\
& \mid & \preda & \textit{Packet Filter} \\
& \mid & \cellcolor{LightYellow} o &  \textit{State Obs} \\
& \mid & \modify{\field}{n} & \textit{Packet Action} \\
& \mid & \cellcolor{LightYellow} e & \textit{State Action} \\
& \mid & \dup & \textit{Duplicate} \\
& \mid & \punion{\polp}{\polq} & \textit{Choice} \\
& \mid & \pseq{\polp}{\polq} & \textit{Sequence} \\
& \mid & \cellcolor{LightYellow} \polp \parallel \polq & \textit{Parallel} \\
& \mid & \pstar{\polp} & \textit{Iteration} \\
& \mid & \cellcolor{LightYellow} a &\textit{Packet Sets} \\
\end{array}
\]
\end{minipage}
\vspace{1ex}\hrule\vspace{1ex}
\caption{\cnetkat syntax. We highlight constructs not in \netkat.}%
\label{fig:syntax}
\end{figure}

\begin{figure}[t]
\hrule\vspace{1ex}
\begin{minipage}{.407\textwidth}\small
  \vspace{-.7em}
\textbf{Semantics}
\[
\begin{array}{r@{\ }c@{\ }l@{\quad}}
\lang{p} (\emptyset) & \defeq &  \{1\cdot \emptyset\} \\
\lang{\abort} (a)& \defeq & \emptyset\\
\lang{\skp} (a)& \defeq & \{\onepom \cdot a\}\\
\lang{t}(a) & \defeq & \{\onepom \cdot  \bsem{t}(a)\}\\
\lang{\modify{\field}{n}}(a) & \defeq &
  \{\onepom  \cdot {a(f\gets n)}\}\\
\lang{b}(a) & \defeq &
\{b  \cdot a\}\\
  \lang{\dup}(a) & \defeq &
  \{a  \cdot a\}\\
  \lang{\punion{\polp}{\polq}}(a) & \defeq &  \lang{\polp}(a) \cup \lang{\polq}(a)\\
    \end{array}
  \]
  \textbf{\quad Predicates}
    \[
  \begin{array}{@{\quad}r@{\,}c@{\,}l@{\ }r@{\,}c@{\,}l}
    \bsem t (a)&\colon&  \pPk \\
      \bsem{\pfalse}(a) & \defeq &
    \emptyset\\
  \bsem{\ptrue}(a) & \defeq &
    a \\
        \bsem{f=n}(a) & \defeq &
    a(f=n) \\
    \bsem{{\preda} \vee_{\BB} {\predb}}(a) & \defeq & \bsem{\preda}(a)  \cup \bsem{\predb}(a)   \\
    \bsem{{\preda} \land_{\BB} {\predb}}(a)  & \defeq & \bsem{\preda}(a)  \cap \bsem{\predb}(a)   \\
    \bsem{\pnot{\preda}}(a)  & \defeq & a\setminus \bsem{\preda}(a) \\
   \end{array}
  \]
  \vspace{0.7em}
  \end{minipage}
  \begin{minipage}{.588\textwidth}\small
    \fbox{$\lang p \colon 2^\Pk \to 2^{\Pom ( \State \cup \Act\cup 2^{\Pk}) \cdot 2^\Pk}$}
    \vspace{.7em}
    \[
    \begin{array}{r@{\ }c@{\ }l@{\quad}}
      \lang{o}(a) & \defeq &
       \State^* \odot \osem{o} \odot \State^* \times \{ a \}  \\
      \lang{e}(a) & \defeq &
      \State^*\odot \{e \}\odot \State^* \times \{ a \}  \\
    \lang{\pseq{\polp}{\polq}} (a)& \defeq &
     \left\{(\upom\cdot\vpom) \cdot b  \;\middle|\; \begin{array}{l}\upom \cdot a' \in \lang{\polp}(a), \\ \vpom \cdot b \in \lang{\polq}(a') \end{array} \right\} \\
     \lang{\polp \parallel \polq}(a)& \defeq &  \left\{(\upom \parallel \vpom)  \cdot (b \cup c) \;\middle|\;
    \begin{array}{l}\upom \cdot b \in \lang{\polp}(a), \\ \vpom \cdot c \in \lang{\polq}(a) \end{array} \right\} \\[1.4ex]
    \lang{\pstar\polp}(a)& \defeq & \displaystyle\bigcup \Bigl\{ \lang{\underbrace{p \cdots p}_{\text{$n$ times}}}(a) : n \in \N \Bigr\}\\
        \end{array}
      \]
%
\textbf{\quad Observations}
  \[
\begin{array}{r@{\,}c@{\,}l@{\ }r@{\,}c@{\,}l}
  \osem{o} &\colon& 2^{\State}\\
      \qquad\qquad  \osem{\bot} & \defeq &
  \emptyset\\
\osem{\top} & \defeq &
\State \\
\osem{v=n} &\defeq& \{ \alpha \in \State \mid \alpha(v) = n \} \\
\osem{{o} \vee {o'}} & \defeq & \osem{o} \cup \osem{o'} \\
\osem{{o} \land {o'}} & \defeq & \osem{o} \cap \osem{o'} \\
\osem{\overline{o}} & \defeq & \bigcup\{Z\in \mathcal{P}_{\leq}(\State) \mid \osem{o}\cap Z=\emptyset\}\\ \\ \\ \\
 \end{array}
\]
\end{minipage}
\vspace{-2.5em}
\hrule\vspace{1ex}
\begin{minipage}{.99\textwidth}\small
\textbf{Filtering, updates and downwards closure \hfill$ a\in  2^{\Pk}, Z\subseteq \State$}
\[
\begin{array}{ll}
a(f=n) = \{ \pi \in a \mid \pi(f) = n\} &
a(f\gets n) = \{\pi[n/f]\mid \pi \in a \} \\
\alpha\leq \beta \iff  \textsf{domain}(\beta)\subseteq \textsf{domain}(\alpha)
&\wedge\quad \forall x\in\textsf{domain}(\beta).\ \alpha(x)=\beta(x)\\
 Z_{\leq} = \{\alpha \mid \exists \beta\in Z \text{ s.t } \alpha \leq \beta \} &
\mathcal{P}_{\leq}(\State) = \{Z \mid Z\subseteq \State \wedge Z=Z_{\leq}\} \\
\end{array}
\]
\end{minipage}
\vspace{1ex}\hrule
\caption{\cnetkat semantics. Pairs $ \upom \cdot b$ in $\sem{p}(a)$ indicate that the program $p$ takes input $a$ and the global state change induced by $p$ is encoded in $\upom$ and constrains the final packet set $b$. We overload $\cdot$ for sequential composition of pomsets and pairs, while $\odot$ is the usual lifting from pomsets to languages.}%
\label{fig:semantics}
\end{figure}

The semantics of a \cnetkat program is represented as a function that takes a set of packets, potentially located in different nodes in the network, and returns a set of possible behaviors that those input packets might produce.
More precisely, the semantics function has type $\lang{-}\colon 2^{\Pk}\to2^{\Pom\cdot 2^{\Pk}}$.
Here, $\Pom$ is the set of \emph{pomsets}~\cite{grabowski-1981,gischer-1988}, which can be thought of as structures that record the causal order between concurrent events (details appear in \cref{sec:preliminaries}).
An element $\upom\cdot b \in \lang p (a)$ means ``there is an execution of $p$ that changes the global variables according to $\upom$, and the set of output packets produced is $b$''.\footnote{We use the notation $\cdot$ to denote pairs: $\upom\cdot b$ denotes the pair $(\upom,b)$.} 

The semantics is defined in \Cref{fig:semantics}. For instance, a packet filter $(\match f n)$ takes a set of packets $a$ and returns $\{\onepom \cdot a(\match f n)\}$, where $a(\match f n)$ contains all packets in $a$ where $f$
has value $n$ and $\onepom$ is the pomset representing that the global state did not change. A modification $(\modify f n)$ takes a set of input packets $a$ and returns $\{\onepom \cdot a(f\gets n)\}$, where $a(f\gets n) = \{ \pi[n/f] : \pi \in a \}$. These two basic packet actions manipulate the \emph{local state} of the program.

On the global state we have observations of the form ($\match v n$) and modifications ($\modify v n$), ($\modify v {v'}$). Each gives rise to a pair in the semantics---$\{v=n \cdot a\}$, $\{(\modify v  n )\cdot a\}$, $\{(\modify v {v'}) \cdot a\}$---in which the input set of packets $a$ is returned as output and the assertion or modification is recorded in the pomset.

Lastly, the primitive $a\in 2^{\Pk}$ is useful for writing specifications. This program copies the set of packets $a$ into the global pomset.
We will see that this is useful for checking inclusion of certain behaviors in a program's semantics, and in the proof of completeness. Formally, the behavior of $a$ on any input set $b$ is $\{a\cdot b\}$, where $a$ is the global state pomset with one node labeled by $a$.

To construct more complicated programs, we can combine the basic elements above using operators from Kleene algebra. For instance, $p+q$ is a program that represents a non-deterministic
choice between $p$ and $q$. Its semantics is obtained by taking the union of sets produced by both $p$ and $q$ on the input packets. We can also compose programs sequentially using $p\mathbin ; q$, where we first apply $p$ to the input packets and then $q$ to all sets of packets produced by $p$, and we compose the corresponding global pomsets sequentially. We can iterate a program finitely many times using $p^*$. Lastly, we can
combine programs with a parallel operator, $p\parallel q$, which denotes a program that, on input $a$, executes both $p$ and $q$ on $a$, and then combines the results: the pomsets denoting the global components are composed in parallel, and the corresponding sets of output packets joined.

\begin{remark}[Concurrency and state, continued]
Note that statements observing or modifying global variables are stored in the pomsets but not executed, that is, we do not actually \emph{check} immediately whether $v$ is indeed $1$ but rather simply record it.
This may seem like an odd choice at first: why does the semantics not also keep a record of the global store?
The reason is related to \cref{rem:concurrency-and-state}.

Consider the program $q = (\match v 0) \mathbin{;}( \match v 1)$, which asserts that $v$ has value $0$, and then that it has value $1$.
In isolation, $q$ does not have any valid behavior, as it sequentially executes two tests that cannot be valid without intermediate intervention.
However, the program $q \parallel (\modify v 1)$ \emph{does} have valid behavior on some interleavings---namely the ones where the assignment $\modify v 1$ is scheduled between the two tests.
It stands to reason that a compositional semantics of such programs should include traces with such local inconsistencies, as they may be explained by actions taken by other programs running in parallel~\cite{pocka}.
For \cnetkat, this is accomplished by placing the observations and modifications in the pomset.

This leaves us with the question of how to obtain the semantics of a program in isolation.
We take a page from POCKA~\cite{pocka}, which uses the set of \emph{guarded pomsets} to filter out the pomsets sensible in isolation; details appear in \cref{sec:example}.
\end{remark}

%

One final modification is needed to obtain the \cnetkat semantics from $\lang{-}$. The idea is to allow interleaving between parallel threads~\cite{hoare-moeller-struth-wehrman-2009}. This is accomplished by adding to the semantics all pomsets in which events are ``more ordered'' than the ones already present in $\lang{-}$. We denote this closed semantics by $\closure[]{\lang{-}}$; a precise definition is given in \cref{sec:cnetkat}.

\smallskip\noindent\textbf{Recording local behavior}
To apply \cnetkat to various verification tasks, we sometimes need to take snapshots of the local state at different points. For example, if we want to argue that $\heartsuit$ packets arrived at switch $3$ before $\spadesuit$ packets arrived at switch $2$, we need more than the information about inputs and outputs that have occurred so far. We therefore have to extend the language with an operator comparable to $\dup$ in \netkat.
On input $a$, the semantics of the $\dup$ operator is the set $\{a\cdot a\}$, where the first component is a single node pomset labeled with set of packets $a$.\footnote{We overload `a' as a set of packets, a programming primitive and a label used in pomsets, but it always denotes a set of packets in the latter two uses as well.} By recording packets inside the pomset, information about changes to packets also contains their relation to changes to global variables during the execution. Hence, using $\dup$, we can infer causality relations between local and global state changes.

The programs $p_1,p_2,p_3$ and $p_4$ used in our running example (see \Cref{fig:running}) can be instrumented with a $\dup$ on every entry to and exit from a switch. This encodes extra information in the semantics that can be used for reasoning about packet-forwarding paths as well as global state changes.
\[
\begin{array}{ll}
p_1 \defeq &\sw = 1 \mathbin{;}\textcolor{blue}{\dup} \mathbin{;} ((v=1 \mathbin{;}\type =\spadesuit \mathbin{;} \textcolor{blue}{\dup}\mathbin{;} \modify \sw  2 \mathbin{;} \textcolor{blue}{\dup} )\\
&\phantom{\sw = 1 \mathbin{;}\textcolor{blue}{\dup} \mathbin{;}x} \parallel (\type =\heartsuit  \mathbin{;} \textcolor{blue}{\dup}\mathbin{;} \sw\gets 3 \mathbin{;} \textcolor{blue}{\dup} \mathbin{;} v\gets 1))\\
p_2 \defeq & \sw =2 \mathbin{;} \textcolor{blue}{\dup} \mathbin{;} \modify \sw 4 \mathbin{;}\textcolor{blue}{\dup}\\
p_3 \defeq &\sw = 3\mathbin{;} \textcolor{blue}{\dup}\mathbin{;} \modify \sw 4 \mathbin{;} \textcolor{blue}{\dup}\\
p_4 \defeq &\sw=4 \mathbin{;} \textcolor{blue}{\dup}
\end{array}
\]
The overall program of the running example then becomes
\[
  p \defeq  \modify v 0   \mathbin{;}  {(p_1 \parallel p_2 \parallel p_3 \parallel p_4)}^*
\]
where the global variable $v$ is initialized to $0$, and the programs $p_1,p_2,p_3,p_4$ are executed in parallel, performing the actions of each individual switch. The Kleene star ensures that the packets may take multiple hops through the network, eventually reaching their final destination (switch $4$).

\begin{remark}
If a $\dup$ occurs in parallel to other threads, then these other parallel threads can only change the exact place of the $\dup$-recording in the pomset via possible interleavings, but not influence its content.
\end{remark}

\begin{remark}
  We model the collection of in-flight packets as a set, as opposed to e.g.\ a partially ordered set encoding their order of arrival. This is an abstraction of our framework. Not putting an order on packets simplifies the algebraic presentation and has the advantage that it enables modeling of switches that reorder packets without an additional primitive. If the order of packets is important, information about this order can be extracted from the semantics. In particular, when packets were forwarded can be deduced by inspecting the sets of packets recorded in the pomset component using $\dup$.
\end{remark}

\noindent\textbf{Two differences between \cnetkat and \netkat
}
Readers familiar with \netkat might wonder why \cref{ex:forwarding} uses $\parallel$ instead of $+$ to compose the branches of $p_1$.
The reason is that in \cnetkat, $\parallel$ is interpreted as multicast and $+$ is interpreted as non-deterministic composition.
In \netkat, programs act on a single input packet, so these coincide. But in \cnetkat, programs act on multiple packets concurrently, so they must be distinguished.

To illustrate the difference, consider wanting to filter the input packets so that only those where field $f$ has value $n$ or field $g$ has value $m$ remain.
In $\netkat$, we can use the program $\match f n + \match g m$, which can be understood in two different ways. First, we can think of it as using (angelic) non-determinism to select a test, yielding $\{\pk\}$ if at least one test passes and $\emptyset$ if both tests fail.
Alternatively, we can think of it as using multicast to copy the input to both $\match f n$ and $\match g m$, then using the tests to perform the required filtering, and finally taking the union of the resulting sets.
In \netkat, the net effect of both interpretations is identical, so multicast and non-determinism can be identified semantically.

However, when we generalize to \emph{sets} of packets, it is natural to expect that processing a set $a$ with $\match f n$ followed by $\match g m$ would yield the subset of $a$ where each packet satisfies at least one of the tests. Operationally, processing $a$ using these programs could be realized by making two copies of $a$, then using the tests to perform the required filtering, and taking the union of the resulting sets. This is reflected in the semantics: $\sem{\match f m \parallel \match g n}(a) = \{\onepom  \cdot (a(f=m)\cup a(g=n))\}$, where we get a single pair in the output. If instead we non-deterministically choose between the tests, the result would be the subset where $f=n$ \emph{or} the subset where $g=m$. Indeed, we have that $\sem {\match f m + \match g n}(a)=\{\onepom \cdot a(f=m) , \onepom \cdot a(g=n)\}$. Hence, multicast and non-determinism can no longer be identified in the context of multiple packets.
For readers familiar with \netkat, this means that the Boolean disjunction $\vee$ is now identified with $\parallel$ rather than $+$.

Lastly, we highlight that
\cnetkat's $\dup$ is fundamentally different from \netkat's $\dup$, which just records versions of the packet during execution. In \cnetkat, $\dup$ does two things: it implements the same functionality as in \netkat, but also structures the recording of packets inside the pomset.

\noindent\textbf{Proving properties with \cnetkat}
In \cref{sec:example}, we analyze the behavior of the running example in detail and show how to filter out the behaviors of $p$ that can be obtained when it is run \emph{in isolation}. In this overview, we establish a simpler property: namely, that $p$ exhibits executions where the packets were at switch $3$ before they were at switch $2$. We first argue this using the denotational semantics and then illustrate how we can establish the same fact with axiomatic reasoning.

Recall a pomset accounts for events and the ordering between them. In the following examples, we will depict pomsets as a graph with nodes labeled by state actions, observations and sets of packets, and the ordering indicated by arrows. For instance, $a\rightarrow b$ means that $a$ happened before $b$.

We evaluate $p$ on input $\{\heartsuit,\spadesuit\}$, where both packets start at switch $1$. In the closed semantics $\closure[]{\sem p} (\{\heartsuit,\spadesuit\})$ we find the following pomset (the $\cdots$ indicate that the pomset continues on the next line, not that nodes are omitted), in the first projection, with $\beta$ a partial function from $\Var$ to $\Val$ s.t. $\beta(v)=1$:

\begin{center}
\smallskip
\begin{tikzpicture}[node distance=0.25cm]
    \begin{scope}[every node/.style={anchor=center}]
    \node (gamma1) {$(\modify v 0)$};
    \node[right=of gamma1] (bla2) {$\{\heartsuit,\spadesuit\}$};
        \node[right=of bla2] (bla3) {$\{\heartsuit\}$};
          \node[right=of bla3] (bla4) {$\{\heartsuit[3/\sw]\}$};
        \node[right=of bla4] (bla5) {$(\modify v 1)$};
                \node[right=of bla5] (bla6) {$\beta$};
\node[right=of bla6] (bla7) {$\cdots$};
    \end{scope}
    \draw[->] (gamma1) edge (bla2);
      \draw[->] (bla2) edge (bla3);
      \draw[->] (bla3) edge (bla4);
      \draw[->] (bla4) edge (bla5);
      \draw[->] (bla5) edge (bla6);
        \draw[->] (bla6) edge (bla7);
\end{tikzpicture}

\begin{tikzpicture}[node distance=0.25cm]
    \begin{scope}[every node/.style={anchor=center}]
      \node (bla6) {$\cdots$};
\node[right=of bla6] (bla7) {$\{\spadesuit\}$};
\node[right=of bla7] (alpha) {$\{\spadesuit[2/\sw]\}$};
\node[yshift=4mm,right=of alpha] (gamma2) {$\{\spadesuit[2/\sw]\}$};
\node[yshift=-4mm,right=of alpha] (gamma5) {$\{\heartsuit[3/\sw]\}$};
\node[right=of gamma2] (x) {$\{\spadesuit[4/\sw]\}$};
\node[right=of gamma5] (y) {$\{\heartsuit[4/\sw]\}$};
\node[yshift=-4mm, right=of x] (z) {$\{\spadesuit[4/\sw],\heartsuit[4/\sw]\}$};
\end{scope}
\draw[->] (bla6) edge (bla7);
\draw[->] (bla7) edge (alpha);
\draw[->] (alpha) edge (gamma2);
\draw[->] (alpha) edge (gamma5);
\draw[->] (gamma2) edge (x);
\draw[->] (gamma5) edge (y);
\draw[->] (x) edge (z);
\draw[->] (y) edge (z);
\end{tikzpicture}
\end{center}

Every node labeled with a set of packets can be understood intuitively as ``at this point in
the execution these packets were a subset of the total packets present in the
network.'' We can observe in the pomset that the $\heartsuit$ packet was at switch
$3$, before the $\spadesuit$ packet reached switch $2$. We also see that $\modify
v 1$, happens between $\modify v 0$ and $\beta$. In the end, both packets are observed
at switch $4$.

The second projection in the semantics corresponding to this pomset is the set of output packets $\{\spadesuit[4/\sw],\heartsuit[4/\sw]\}$.

In
\ifarxiv%
\cref{sec:analysis},
\else%
the full version of this article~\cite[Appendix~E]{fullversion},
\fi%
we show something stronger: in all behaviors that can happen in isolation, the packet $\heartsuit[3/\sw]$ is recorded into the global pomset before the assignment $\modify v 1$, which precedes the observation that $v$ equals $1$ and the generation of the packet $\spadesuit[2/\sw]$.

We can write an axiomatic statement that captures that the above
behavior is in the closed semantics of $p$ on input $\{\heartsuit,\spadesuit\}$. To do this, we first need to capture the pictured global state pomset with corresponding set of output packets syntactically, for which we use an abbreviation. Namely, we can write a program that outputs, on any input, a specific packet: for a packet $\pk$, we write this program simply as $\pk$. The output of $\sem{\pk}$ on any input is $\{\onepom\cdot\{\pk\}\}$. This extends to sets of packets:  $\heartsuit\parallel\spadesuit$ denotes a program whose semantics is $\{\onepom\cdot \{\heartsuit\parallel\spadesuit\}\}$ on any input. This notation pairs well with the use of the letters $a\in 2^{\Pk}$ as programming syntax: if we know which set of packets we (want to) record into the global state pomset
with $\dup$, we can also directly write this set of packets in the program as a syntactic letter. For instance, the program $(\heartsuit\parallel\spadesuit)\mathbin{;}\dup$, has the same behaviors as $(\heartsuit\parallel\spadesuit)\mathbin{;}\{\heartsuit,\spadesuit\}$:
the moment we execute the $\dup$, we know the current set of packets is $\{\heartsuit,\spadesuit\}$, and thus writing this set of packets as a letter and recording that letter into the global state pomset will have the same result. Using these two pieces of information, we can write the program
\begin{align}
  q&
\defeq
\Big((\modify v 0)\mathbin{;}
\{\heartsuit,\spadesuit\} \mathbin{;} \{\heartsuit \}\mathbin{;} \{\heartsuit [3/\sw]\}\mathbin{;} (\modify v 1)\mathbin{;} (\match v 1)\mathbin{;} \{\spadesuit\} \mathbin{;}  \dots \label{programq} \\
&\dots \{\spadesuit[2/\sw]\}\mathbin{;}\Big((\{\spadesuit[2/\sw]\} \mathbin{;}\{\spadesuit[4/\sw]\}) \parallel ( \{\heartsuit [3/\sw]\} \mathbin{;}\{\heartsuit [4/\sw]\} )\Big)
\mathbin{;} \dots \nonumber \\
& \dots \{\spadesuit[4/\sw],\heartsuit[4/\sw]\}\Big)\mathbin{;} (\spadesuit[4/\sw]\parallel \heartsuit[4/\sw]) \nonumber
\end{align}
The first chunk of this program is the syntactic encoding of the desired global state pomset, where the $\heartsuit$ packet arrives at switch $3$ before the $\spadesuit$ packet arrives at switch $2$, and the final parallel of packets represents the set of output packets. We can prove using the axioms of \cnetkat that
\begin{equation}\label{eq:bla}
(\heartsuit\parallel\spadesuit) \mathbin{;} q \leqq (\heartsuit\parallel\spadesuit) \mathbin{;} p
\end{equation}

\eqref{eq:bla} states that the behavior of $q$ on input $\{\heartsuit,\spadesuit\}$, is included in the behavior of $p$ on the same input. In the behavior of $q$, it is clear that the $\heartsuit$ packets are observed at switch $3$ before the $\spadesuit$ packets appear at switch $2$. 

\begin{remark}[Generalized alphabet]\label{remark:letters}
Here we see the use of sets of packets as letters in the program syntax. Program $q$ is much closer to the behavior we try to capture, and therefore easier to analyze, than a program containing $\dup$.
\end{remark}

To check the validity of equivalences such as~\eqref{eq:bla}, we axiomitize \cnetkat and prove it sound and complete. The axioms include the axioms of \KA, extended with additional axioms for operations that manipulate packets and the global state. The full axiomatization appears in \Cref{sec:theory}. For instance, $\pfalse \mathbin{;} q \equiv \pfalse$ states that no outputs are produced in the absence of inputs. The program $\pfalse$ drops the set of inputs and returns $\{\onepom\cdot \emptyset\}$. Any program $q$ after $\pfalse$ outputs $\{\onepom\cdot\emptyset\}$, because $q$ is not executed when
the input is empty. In contrast, $q\mathbin{;}\pfalse\equiv\pfalse$ does not hold since $q$ might have changed the global state.


In addition to $\pfalse$, \cnetkat has a program $\abort$, which acts as a unit for non-deterministic choice ($+$). To illustrate the difference between $\abort$ and $\pfalse$ consider $(\match f n)\mathbin{;} (\match f m)$ and $(\match v n)\wedge (\match v m)$, where $m\neq n$.
The first program filters using $f=n$ and and then filters using $f=m$ where $m\neq n$. This yields $\{\onepom\cdot \emptyset\}$, since a packet cannot have different values for $f$. Hence, we can derive
$
(\match f n)\mathbin{;} (\match f m) \equiv  \pfalse
$.
The second program asserts the global state variable $v$ has value $n$ and $m$, which is inconsistent; we require variable $v$ to have two different values at the same time. Hence, from the axioms we can derive that
$
(\match v n)\wedge (\match v m) \equiv \bot \equiv \abort
$.

We prove in \Cref{sec:completeness} that the axiomatization presented in \cref{sec:theory} is not only sound but also complete---i.e., all programs with the same semantics can be proved equivalent using the axioms. The rest of the paper is devoted to presenting the \cnetkat syntax and semantics formally (\Cref{sec:cnetkat}), and establishing conservativity results over \netkat and \POCKA. Lastly we present a case study (\Cref{sec:example}). 

\section{Concurrent NetKAT}\label{sec:cnetkat}
This section defines the syntax and semantics of \cnetkat formally.

\subsection{Pomsets and pomset languages}\label{sec:preliminaries}
%

For a poset $(X, \leq)$ and a set $S \subseteq X$, define the \emph{downwards-closure} of $S$ by $S_{\leq} ::= \{x \mid \exists y\in S \text{ s.t } x\leq y\}$ and $P_{\leq}(X) ::= \{ Y \subseteq X \mid Y = Y_{\leq} \}$.
It is well-known that $P_{\leq}(X)$ carries the structure of a bounded distributive lattice, with intersection as meet, union as join, $X$ as top and $\emptyset$ as bottom.
Further, if $(X, \leq)$ is finite, the lattice is itself finite and thus carries a (necessarily unique) pseudocomplement defined by $\overline{Y} ::= \bigcup \{ Z \in P_{\leq}(X) \mid Y \cap Z = \emptyset\}$.
We provide a concrete lattice with a pseudocomplement below.

\medskip
\noindent\textbf{\textit{Pomsets}} are used to capture the different evolutions of the state as it is accessed concurrently by different threads. Pomsets are labeled posets (up to isomorphism), used as a generalization of words~\cite{gischer-1988,grabowski-1981}.
A \emph{labeled poset} over a finite alphabet $\alphabet$ is a triple $\lp{u} = \angl{S_\lp{u}, \leq_\lp{u}, \lambda_\lp{u}}$, where $(S_\lp{u},\leq_\lp{u})$ is a partially ordered set and $\lambda_\lp{u} \colon S \to \alphabet$ is the labeling function.
For $\lp{u}, \lp{v}$ labeled posets,
we say $\lp{u}$ is \emph{isomorphic} to $\lp{v}$, $\lp{u} \cong \lp{v}$, if there exists a bijection $h\colon S_\lp{u} \to S_\lp{v}$ that preserves labels --- $\lambda_\lp{v} \circ h = \lambda_\lp{u}$--- and preserves and reflects ordering--- $s \leq_\lp{u} s'$ if and only if $h(s) \leq_\lp{v} h(s')$.
A \emph{pomset} over $\alphabet$ is an isomorphism class of labeled posets over $\alphabet$, i.e., the class
$[\lp{v}] = \{ \lp{u} \mid \lp{u} \cong \lp{v} \}$ for some labeled poset $\lp{v}$. Because pomsets are label-preserving
isomorphism classes, the nature of the carrier is not relevant, only its cardinality and order.  The
triple $\lp{u} = \angl{S_\lp{u}, \leq_\lp{u}, \lambda_\lp{u}}$ is a representation of the pomset.
However, often we abuse terminology and call  $\upom$ the pomset.

We write $\Pom(\alphabet)$ for the set of pomsets over $\alphabet$, and $\onepom$ for the empty pomset.
When $a \in \alphabet$, we write $a$ for the pomset represented by the labeled poset with a single node labeled by $a$.
Pomsets can be composed sequentially and in parallel.

The \emph{parallel composition} of two pomsets is obtained by taking the disjoint union of the carriers, while keeping the ordering relations within each component. Formally, $\upom \parallel \vpom=\angl{S_{\lp{u} \parallel \lp{v}}, \leq_{\lp{u} \parallel \lp{v}},\lambda_{\lp{u} \parallel \lp{v}}}$, with
    $S_{\lp{u} \parallel \lp{v}} = S_\lp{u} + S_\lp{v}$,
    ${\leq_{\lp{u} \parallel \lp{v}}} = {\leq_\lp{u}} \cup {\leq_\lp{v}}$
    and $\lambda_{\lp{u} \parallel \lp{v}}(x)= \lambda_\lp{u}(x)$, for $x\in S_\lp{u}$, and $\lambda_{\lp{u} \parallel \lp{v}}(x)=\lambda_\lp{v}(x)$,  for $x\in S_\lp{v}$. Two pomsets are composed \emph{sequentially} by taking the disjoint union of the carriers and ordering all elements of the first before all elements of the second, keeping the ordering relations within each component. Formally, $\upom \cdot \vpom = \angl{S_{\lp{u} \cdot \lp{v}}, \leq_{\lp{u} \cdot \lp{v}},\lambda_{\lp{u} \cdot \lp{v}}}$, with
    $S_{\lp{u} \cdot \lp{v}} = S_\lp{u} + S_\lp{v}$,
    ${\leq_{\lp{u} \cdot \lp{v}}} = {\leq_\lp{u}} \cup {\leq_\lp{v}} \cup (S_\lp{u} \times S_\lp{v})$ and
    $\lambda_{\lp{u} \cdot \lp{v}} = \lambda_{\lp{u} \parallel \lp{v}}$.


Gischer introduced a notion of ordering on pomsets~\cite{gischer-1988}: $\upom \sqsubseteq \vpom$  means that $\upom, \vpom$ have the same events and labels, but $\upom$ is ``more sequential'' than $\vpom$ in the sense that more events are ordered. Formally, $\lp{u} \sqsubseteq \lp{v}$ if there exists a
label- and order-preserving bijection $h\colon S_\lp{v} \to S_\lp{u}$.

\textbf{\textit{Pomset languages}} are simply sets of pomsets. The operations on pomsets lift pointwise to pomset languages, see \cref{fig:semantics}. The semantics of concurrent threads requires ensuring a closure property. In particular, we will close pomset languages under the subsumption order of Gischer. Additionally, for pomsets that contain nodes labeled by observations, we make use of a \emph{contraction} order: $\upom \preceq \vpom$, capturing that $\upom$ results from $\vpom$ by eliminating consecutive observations that can be collapsed into one. As an example, consider
\begin{mathpar}
\begin{tikzpicture}[node distance=0.5cm]
  \begin{scope}[every node/.style={anchor=center}]
    \node (atom1) {$a$};
    \node[yshift=3mm, right=of atom1] (write1) {$\alpha$};
    \node[yshift=-3mm, right=of atom1] (write2) {$b$};
    \node[yshift=-3mm, right=of write1] (atom7) {$c$};
  \end{scope}
  \path (atom1) edge[->] (write1);
  \path (atom1) edge[->] (write2);
  \path (write1) edge[->] (atom7);
  \path (write2) edge[->] (atom7);
\end{tikzpicture}
\and
\begin{tikzpicture}[node distance=0.5cm]
 \begin{scope}[every node/.style={anchor=center}]
   \node (atom1) {$a$};
   \node[yshift=3mm, right=of atom1] (write1) {$\alpha$};
   \node[yshift=-3mm, xshift=5mm, right=of atom1] (write2) {$b$};
   \node[right=of write1] (atom3) {$\alpha$};
   \node[yshift=-3mm, right=of atom3] (atom7) {$c$};
 \end{scope}
 \path (atom1) edge[->] (write1);
 \path (atom1) edge[->] (write2);
 \path (write1) edge[->] (atom3);
 \path (write2) edge[->] (atom7);
 \path (atom3) edge[->] (atom7);
\end{tikzpicture}
\end{mathpar}
 Denote these pomset with $\upom$ and $\vpom$ respectively, and let $\alpha\in\State$. Then $\upom\preceq \vpom$. A formal definition can be found in
 \ifarxiv%
 \cref{app:preliminaries}.
 \else%
 the full version of this article~\cite[Appendix~A]{fullversion}.
 \fi%

\begin{textAtEnd}[allend, category=preliminaries]
  Formally, $\upom \preceq \vpom$ holds iff there exists a surjection $h\colon S_\lp{v} \to S_\lp{u}$ satisfying: (i) $\lambda_{\lp{u}} \circ h = \lambda_{\lp{v}}$; (ii) $v \leq_{\lp{v}} v'$ implies $h(v) \leq_{\lp{u}} h(v')$; (iii) if $h(v)
  \leq_{\lp{u}} h(v')$, then $\lambda_{\lp{v}}(v),
\lambda_{\lp{v}}(v') \in \State$ implies $v \leq_{\lp{v}} v'$ or $v'
\leq_{\lp{v}} v$, and $\lambda_{\lp{v}}(v)$ or
$\lambda_{\lp{v}}(v') \not\in \State$ implies $v \leq_{\lp{v}} v'$.

\end{textAtEnd}
%
\begin{definition}[Closure]\label{def:closure}
  Let $L$ be a pomset language.
\begin{mathpar}
\closure[\hexch]L=\{\upom\mid \exists \vpom\in L \text{ s.t. }\upom\sqsubseteq \vpom\} \and
\closure[\hcontr]{L}=\{\upom\mid \exists \vpom\in L \text{ s.t. }\upom\preceq \vpom\}
\end{mathpar}
We define $
\closure[\hcontr\cup\hexch]L$ as the smallest language containing $L$ and satisfying that if $\vpom\in \closure[\hcontr\cup\hexch]L$ and $\upom\preceq \vpom$ or $\upom\sqsubseteq \vpom$, then $\upom\in \closure[\hcontr\cup\hexch]L$.
\end{definition}

Closure under $\sqsubseteq$ is called $\hexch$ because it ensures soundness of the \emph{exchange law}, an axiom introduced in~\cite{hoare-moeller-struth-wehrman-2009} to capture the possibility of interleaving. Closure under contraction is motivated algebraically; it ensures soundness of one of the axioms necessary when adding a test algebra (a PCDL or a BA) to a KA~\cite{fossacs2020}.
%
%
%
%
%
\begin{textAtEnd}[allend,category=preliminaries]
The following lemma was shown in~\cite{fossacs2020}:
\begin{lemma}\label{closure-order}
 Let $L$ be a pomset language.
 We have $
 \closure[\hexch\cup\hcontr]{L}=\closure[\hcontr]{(\closure[\hexch]{L})}$.
\end{lemma}
Hence, we know that $U\in   \closure[\hexch\cup\hcontr]{L}$ if and only if there exists $W\in\closure[\hexch]{L}$ and $V\in L$ such that $U\preceq W\sqsubseteq V$.
\end{textAtEnd}

%
%

\subsection{\cnetkat: syntax and semantics}\label{sec:syntax}
\cnetkat expressions denote (possibly concurrent) packet processing programs that have access to a global state.
 Syntactically, \cnetkat is a language built from alphabets of tests and actions, each of which is divided in two categories. For packet tests, we firstly inherit \netkat's {\em packet predicates}, which are elements of a Boolean algebra generated by an alphabet of basic tests on packet fields. Packet predicates $t,u$ include constants $\drp$ and $\ptrue$, denoting false and true, basic tests $\match f n$, negation $\neg t$, disjunction $t\vee_{\BB}u$ and conjunction $t\wedge_{\BB}u$ operations.

Additionally, we have state observations, which do not have the structure of a Boolean algebra but instead form a
pseudocomplemented distributive lattice. Intuitively, the functions denoting the state are partial. State observations $o,o'$ include constants $\bot$ and $\top$, basic tests $\match v n$, pseudocomplement $\overline o$, intersection
$o\land o'$ and union $o\vee o'$. The other constructs were introduced in \Cref{sec:overview} (see \Cref{fig:syntax}).

The semantics of a program is a function $\lang{\cdot} \colon 2^\Pk \to 2^{\Pom ( \State \cup \Act \cup 2^{\Pk}) \cdot 2^\Pk}$ that takes a set of packets $a$ and produces a (possibly empty) set of pairs $\upom \cdot b$ consisting of a pomset $\upom$, recording the global state behavior and the storage of local packets whenever $\dup$ is used, and a set of packets $b$.
 On an empty input set, every program produces $\{\onepom\cdot\emptyset\}$, modeling that nothing can happen without packets. Producing the empty set when the input is non-empty models a program that has aborted, whereas producing a set $\{\onepom \cdot \emptyset\}$ models dropping all the packets without any change to the state. Most of the semantics was already explained in \Cref{sec:overview}; in the following we elaborate on some behaviors and illustrate subtleties concerning the units. See \Cref{fig:semantics}
for an overview of the full denotational semantics of \cnetkat.

On a non-empty input $a$, a packet filter $t$ removes packets in $a$ that do not satisfy predicate $t$ and does not touch the state --- this is captured by the set $\{\onepom \cdot \bsem{t}(a)\}$, where $\bsem{t}(a)$ is interpreted as an element of the Boolean algebra $(2^a, \cup,\cap,\emptyset,a,\setminus)
$ defined by the poset $(2^a,\subseteq)$, and $\bsem{t}(a)$ is defined as the homomorphic extension of $\bsem{\match f n}(a)=\{\pk\in a\mid \pi(f)=n\}$.

A state observation denotes a function that returns a set with elements
 $\upom \cdot a$ when applied to a set $a$. In case the original input
 set $a$ is empty, nothing happens and the output of $\sem{o}(a)$ is simply
 $\{\onepom\cdot \emptyset\}$. When $a$ is not empty, the semantics of $o$ makes use of an observation algebra developed in~\cite{jipsen-moshier-2016,pocka}.
 More formally, we take the
 pseudocomplemented bounded distributive lattice $(P_{\leq}(\State),
 \cup,\cap,\State,\emptyset, \pnegop,)$ generated by the poset
 $(\State,\leq)$ with $\alpha\leq \beta$ if and only if
 $\textsf{domain}(\beta)\subseteq \textsf{domain}(\alpha)$ and
$\forall x\in\textsf{domain}(\beta).\alpha(x)=\beta(x)$. Then, a state
observation is interpreted as $\State^* \cdot \osem{o} \cdot \State^* \times \{ a \} $, where $\osem{o}$ is an element of $P_{\leq}(\State)$ and defined as the homomorphic extension of the assignment $\osem{\match v n}=\{\alpha\in \State \mid \alpha(v)=n\}$.
Intuitively, in $\osem{o}$, we find all the partial functions (elements of $\State$) that agree with $o$. For instance, $\osem{\match v n}$ contains all partial functions that assign $n$ to $v$. This also illustrates the need for a pseudocomplement rather than a complement: if threads have only partial information about the state, an observation should be satisfied only
if there is \emph{positive evidence} for it. Hence, e.g. $\overline {\match v n}$ should be satisfied only if $v$ has a
value and it is not $n$, which is not captured by the complement from a Boolean algebra --- the complement would also include partial functions that do not assign a value to $v$ in the behavior of $\overline {\match v n}$. This is incorrect, because if $v$ has no value in a partial observation, we might learn later that the actual value of $v$ was in fact $n$, and it was therefore incorrect to assert $\overline{\match v n}$.

 State modifications are interpreted as a set of elements $\upom \cdot a$ when applied to a set $a$. The pomsets $\upom$ record the state modification surrounded by arbitrary state observations; in the first projection of the semantics of the assignment $\modify v n$ we get a set of possible pomsets:
$
 \State^*\odot \{v\leftarrow n \}\odot \State^*
$.

\begin{remark}\label{remark:stateobs}
  We surround state changes and observations with arbitrary sequences of states to include global pomsets that have alternating modifications and states in the semantics. Reasoning about behavior of programs is more practical using such alternating pomsets, because the states allow one to take stock of the configuration of the machine in between modifications. The semantics contains also non-alternating pomsets to ensure compositionality w.r.t the parallel.
\end{remark}
%

\cnetkat has six different syntactical units, some of which coincide semantically. There are two units for packets: $\drp$, which drops all the packets ($\{\onepom\cdot\emptyset\}$), and $\ptrue$, which passes the current packets without changing the state ($\{\onepom\cdot a\}$ on input $a$). Similarly, we have two units for state observations: $\bot$ and $\top$. The first one indicates an inconsistent state, and therefore the whole program exhibits no behavior; its behavior is $\emptyset$. The second one indicates any state observation is acceptable, and its behavior on input $a$ is $\{s\cdot a \mid s\in\State\}$. Lastly there are two
units for programs in general: $\abort$, the program without behavior, and $\skp$, the program where nothing happens (on input $a$ its semantics is
$\{\onepom \cdot a\}$). Hence, $\abort$ is equivalent to
$\bot$ and $\skp$ equivalent to $\ptrue$. All units behave as
$\{\onepom\cdot\emptyset\}$ when the input set is $\emptyset$,
because nothing happens when there are no packets.
%

  The \cnetkat semantics consists of pairs of global state pomsets and sets of output packets. It might be possible to
  encode the information of the output packets as a final node in the pomset, but keeping the set of output packets
  separated allows us to easily track the input-output behavior of a program in terms of packets. This brings \cnetkat
  closer to \netkat and its packet processing behavior. In particular, the \netkat packet processing axioms, can only be used because we track the input-output behavior of the program
  separately.
%

To obtain the full semantics, and ensure we capture correctly the intended behavior, we need to perform a closure on the state component.
\begin{definition}[Closed Semantics]\label{def:closed-semantics}
Given a \cnetkat policy $p$, we define the semantics of $p$ when applied to input $a\in 2^\Pk$ as
\[\closure[] {\lang{p}} (a) =
\left\{\upom \cdot b \mid \vpom \cdot b \in \lang{p}(a), \upom \in\closure[\hcontr\cup\hexch] { \{\vpom\}}\right\}
\]
\end{definition}
Closure under $\hexch$ and $\hcontr$ formalizes important intuitions about the semantics of concurrent threads. The closure under $\hexch$ ensures all traces resulting from interleaving threads are included, and the closure under $\hcontr$ specifies that if two observations hold simultaneously, then it is possible to observe them in sequence. Note that the converse should not hold as some action could happen in between the two observations in a parallel thread.

\begin{textAtEnd}[allend, category=syntax]

  The following properties concerning the interaction between closure and other operators can be useful. Here $\parallel, \cdot$ and the star all have their usual interpretation on pomset languages. That is, $L\parallel K=\{\upom\parallel
  \vpom\mid \upom\in L, \vpom\in K\}$, $L\cdot K=\{\upom\cdot
  \vpom\mid \upom\in L, \vpom\in K\}$ and the Kleene star of a pomset language $L$ is defined as
  $L^*=\bigcup_{n\in\mathbb{N}}L^n$, where $L^0=\{\onepom\}$ and
  $L^{n+1}=L^n\cdot L$. All properties are proven in~\cite{fossacs2020}, except for~\eqref{property:hypothesis-union}, which follows from~\cite[Lemma 3.7]{cka} and~\cite[Lemma 4.8]{kao}.
  \begin{lemma}\label{lemma:basicfacts}
    Let $L, K $ be pomset languages. The following hold.
  \begin{enumerate}
      \item\label{property:hypothesis-monotone}
      If $L \subseteq K$, then $\closure[\hexch\cup\hcontr] L \subseteq \closure[\hexch\cup\hcontr] K$.

      \item\label{property:hypothesis-union}
      $\closure[\hexch\cup\hcontr] {(L \cup K)} = \closure[\hexch\cup\hcontr] L \cup \closure[\hexch\cup\hcontr] K$

   \item\label{property:hypothesis-concat}
   $\closure [\hexch\cup\hcontr] {(L \cdot K)} = \closure[\hexch\cup\hcontr] {(\closure[\hexch\cup\hcontr] L \cdot \closure[\hexch\cup\hcontr] K)}$

      \item\label{property:hypothesis-parallel}
      $\closure[\hexch\cup\hcontr] {(L \parallel K)} = \closure[\hexch\cup\hcontr] {(\closure[\hexch\cup\hcontr] L \parallel \closure[\hexch\cup\hcontr] K)}$

  \end{enumerate}

  \end{lemma}

  We can prove the following properties about our closure definition.
  \begin{lemma}\label{lemma:basicfacts2}
    Let $a\in 2^{\Pk}$ and $p,q\in \programs$.
  \begin{enumerate}

      \item\label{property:monotone}
    $\sem{p}(a)\subseteq \closure[] {\sem{q}}(a)$
    if and only if $\closure[]{\sem p}(a)\subseteq \closure[] {\sem{q}}(a)$.

      \item\label{property:union}
      $\closure[] {(\sem {p + q})}(a) = \closure[]{\sem p} (a) \cup \closure[] {\sem q} (a)$

      \item\label{property:seq}
        $\closure[]{\sem{p\mathbin{;}q}}(a)=
        \{\wpom\cdot b\mid \upom
              \cdot a'\in \closure[]{\sem{p}}(a),\vpom \cdot b\in
      \closure[]{\sem{q}}(a') ,
      \wpom\in\closure[\hexch\cup\hcontr]{\{\upom\cdot\vpom\}}\}$

      \item\label{property:parallel}
        $\closure[]{\sem{p\parallel q}}(a)=$\\
        $\{\wpom\cdot (b\cup c)\mid \upom
              \cdot b\in \closure[]{\sem{p}}(a),\vpom \cdot c\in
      \closure[]{\sem{q}}(a) ,
      \wpom\in\closure[\hexch\cup\hcontr]{\{\upom\parallel\vpom\}}\}$

      \item\label{property:star}
      $  \closure[]{\sem{p^*}}(a) = \{\upom\cdot b \mid \vpom\cdot b\in
         \displaystyle\bigsqcup_{n \in \N} \closure[]{\sem{p^{(n)}}}(a), \upom\in \closure[\hexch\cup\hcontr]{\{\vpom\}}\} $

         \item\label{closure-com} For all $\vpom\cdot b \in
             \closure[]{\sem p}(a)$, if $\upom\in\closure[\hcontr\cup\hexch]{\{\vpom\}}$, then
             $\upom\cdot b\in \closure[]{\sem p}(a)$.
  \end{enumerate}
  \end{lemma}
  \begin{proof}
    For the first item, take $w\in\closure[]{\sem p}(a) $. Thus $w=\upom\cdot b$ such that
    $\vpom\cdot b \in \sem{p}(a)$
    and $\upom\in\closure[\hexch\cup\hcontr]{\{\vpom\}}$. Then $\vpom\cdot b \in \closure[]{\sem q }(a)$ and the desired result follows immediately. For the other direction we take $w\in \sem p (a)$. As $w=\upom\cdot b$ for some
    $\upom$ and $b$ and $\upom\in\closure[\hexch\cup\hcontr]{\{\upom\}}$, we also have that $w\in \closure[]{\sem p}(a)$. From
    the assumption it then immediately follows that $w\in  \closure[]{\sem q}(a)$.

    For the second item, we take $w\in \closure[] {\sem {p + q}}(a)$. Hence, $w=\upom\cdot b$ such that
    $\vpom\cdot b \in \sem{p+q}(a)=\sem p (a)\cup \sem q (a)$
    and $\upom\in\closure[\hexch\cup\hcontr]{\{\vpom\}}$.
    W.l.o.g.\ we can assume $\vpom\cdot b\in \sem p (a)$, and by definition we then immediately
    we obtain that $\upom\in \closure[]{\sem p}(a)$.
    For the other direction we take $w\in \closure[]{\sem p} (a)$, and thus $w=\upom\cdot b$ such that $\vpom \cdot b \in \sem p (a)$, and $\upom\in\closure[\hexch\cup\hcontr]{\{\vpom\}}$. Then immediately $\vpom\cdot b \in \sem{p+q}(a)$
    and thus $\upom\cdot b \in \closure[] {(\sem {p + q})}(a) $.

    For the third item,
    take a word $w\in \closure[]{\sem{p\mathbin{;}q}}(a)$.
    Hence, $w=\upom\cdot b$ such that $\vpom \cdot b\in \sem{p\mathbin{;}q}(a)$ and $\upom\in\closure[\hexch\cup\hcontr]{\{\vpom\}}$.
    From this we obtain that $\vpom=\mathbf{x}\cdot\mathbf{y}$ for $\mathbf{x}\cdot a'\in \sem{ p}(a)$ and $\mathbf{y}\cdot b\in
    \sem{q}(a')$. Using~\ref{property:monotone}, we then conclude that $\mathbf{x}\cdot a'\in \closure[]{\sem{ p}}(a)$ and $\mathbf{y}\cdot b\in
    \closure[]{\sem{q}}(a')$.
    We need that $\upom\in \closure[\hexch\cup\hcontr]{\{\mathbf{x}\cdot\mathbf{y}\}}$, but this is immediate as $\vpom=\mathbf{x}\cdot\mathbf{y}$.

    For the other direction, we take $\wpom\cdot b$ such that $\upom\cdot a'\in \closure[]{\sem{p}}(a)$, $\vpom \cdot b\in
\closure[]{\sem{q}}(a')$, and
$\wpom\in\closure[\hexch\cup\hcontr]{\{\upom\cdot\vpom\}}$. Hence,
from \cref{def:closed-semantics} we obtain that there exists
$\mathbf{x}$ and $\mathbf{y}$ such that $\mathbf{x}\cdot a'\in \sem{p}(a)$, $\mathbf{y}\cdot b\in \sem q (a')$,
 $\upom\in\closure[\hexch\cup\hcontr]{\{\mathbf{x}\}}$,
 and $\vpom\in\closure[\hexch\cup\hcontr]{\{\mathbf{y}\}}$.
We can conclude that $(\mathbf{x}\cdot\mathbf{y})\cdot b\in \closure[]{\sem{p\mathbin{;}q}}(a)$.
 The only
 thing left to prove for $\wpom\cdot b\in \closure[]{\sem{p\mathbin{;}q}}(a)$ is $\wpom\in\closure[\hexch\cup\hcontr]{\{\mathbf{x}\cdot\mathbf{y}\}}$.
  From $\upom\in\closure[\hexch\cup\hcontr]{\{\mathbf{x}\}}$ and property~\ref{property:hypothesis-monotone} from
  \cref{lemma:basicfacts}, we know that $\closure[\hexch\cup\hcontr]{\{\upom\}}\subseteq \closure[\hexch\cup\hcontr]{\{\mathbf{x}\}}$
  and $\closure[\hexch\cup\hcontr]{\{\vpom\}}\subseteq \closure[\hexch\cup\hcontr]{\{\mathbf{y}\}}$.
  Then using~\ref{property:hypothesis-concat} from \cref{lemma:basicfacts}, we can conclude that
  \begin{align*}
    \wpom\in\closure[\hexch\cup\hcontr]{\{\upom\cdot\vpom\}}&=\closurep[\hexch\cup\hcontr]{(\closure[\hexch\cup\hcontr]
   {\{\upom\}}\cdot \closure[\hexch\cup\hcontr]{\{\vpom\}})}\\
   &\subseteq
   \closurep[\hexch\cup\hcontr]{(\closure[\hexch\cup\hcontr] {\{\mathbf{x}\}}\cdot \closure[\hexch\cup\hcontr]{\{\mathbf{y}\}})}\\
   &=\closure[\hexch\cup\hcontr]{\{\mathbf{x}\cdot\mathbf{y}\}}\end{align*}

   The proof of the fourth item is analogous but instead uses~\ref{property:hypothesis-parallel} of \cref{lemma:basicfacts}.

   For the fifth item, we observe that by definition $\closure[]{\sem{p^*}}(a)=\{\upom\cdot b \mid \vpom\cdot b \in
   \displaystyle\bigsqcup_{n \in \N} \sem{p^{(n)}}(a), \upom\in\closure[\hexch\cup\hcontr]{\{\vpom\}}\}$. Hence, the left to right direction is trivial by~\ref{property:monotone}.
   For the other direction, we have $\upom\cdot b$ such that $\upom\in\closure[\hexch\cup\hcontr]{\{\vpom\}}$ and $\vpom\cdot b\in \closure[]{\sem{p^{(n)}}}(a)$ for some $n\in \N$.
   Hence $\wpom\cdot b \in \sem{p^{(n)}}(a)$ and $\vpom\in \closure[\hexch\cup\hcontr]{\{\wpom\}}$.
   As $\upom\in\closure[\hexch\cup\hcontr]{\{\vpom\}}\subseteq \closure[\hexch\cup\hcontr]{\{\wpom\}}$, we obtain the desired result.

For the last item, we derive the following.
      If $\vpom\cdot b \in \closure[]{\sem p}(a)$, by definition this means that there exists $\wpom\cdot b \in \sem p (a)$ and
      $\vpom\in\closure[\hexch\cup\hcontr]{\{\wpom\}}$. From this we obtain that $\{\vpom\}\subseteq \closure[\hexch\cup\hcontr]{\{\wpom\}}$ and
      using property~\ref{property:hypothesis-monotone} of \cref{lemma:basicfacts} we can conclude that
      $\closure[\hexch\cup\hcontr]{\{\vpom\}}\subseteq \closure[\hexch\cup\hcontr]{\{\wpom\}}$.
      As $\upom\in\closure[\hcontr\cup\hexch]{\{\vpom\}}$, we obtain that $\upom \in \closure[\hexch\cup\hcontr]{\{\wpom\}} $.
      This results immediately in $\upom\cdot b \in \closure[]{\sem p}(a)$ by definition.
  \end{proof}

\end{textAtEnd}

We distinguish state, packet and deterministic packet programs as follows.
\begin{definition}[State and deterministic packet programs]
Let $\mathcal{T}_{\mathsf{packet}}$ denote packet programs, which are programs generated by the following grammar:
\[
    p, q::= t \in \BB\cup \left\{\modify f n \mid f\in \Field, n\in \Val\right\} \pipe p + q \pipe p \mathbin{;} q \pipe p \parallel q \pipe p^*
\]

Let $\mathcal{T}_{\mathsf{state}}(\Sigma)$ denote state programs over alphabet $\Sigma$:
\[
    s, v::= \abort\pipe \skp \pipe u \in \Sigma \pipe s + v \pipe s \mathbin{;}  v \pipe s \parallel v \pipe s^*
\]

Let $\mathcal{T}_{\mathsf{det-pack}}$ denote deterministic packet programs:\footnote{
  Equivalently, we can define $\mathcal{T}_{\mathsf{packet}}$ by adding a predicate $H$ to the signature of our algebra that counts the number of $*$'s and $+$'s a term contains, and a packet program $p$ is an element of $\mathcal{T}_{\mathsf{det-pack}}$ if and only if $p\in\mathcal{T}_{\mathsf{packet}}$ and $H(p)=0$.
}:
\[
    x, y::= t \in \BB\cup \left\{\modify f n \mid f\in \Field, n\in \Val\right\} \pipe x \mathbin{;}  y \pipe x \parallel y
\]
\end{definition}
In this paper we mostly use state programs over alphabet $\OO\cup\Act\cup 2^{\Pk}\cup\{\dup\}$. Whenever we intend to use this alphabet, we simply write $\mathcal{T}_{\mathsf{state}}$.

We prove the following lemmas regarding the \cnetkat semantics.
\begin{theoremEnd}[default, category=syntax]{lemma}[State and packet program semantics]\label{lemma:statepacketterms}
Let $p\in \mathcal{T}_{\mathsf{packet}}$, $s\in \mathcal{T}_{\mathsf{state}}$
and $a\in 2^{\Pk}$. For all $w\in \sem p (a)$, $w$ is of the form
$\onepom\cdot b$ for $b\in 2^{\Pk}$. For all $w \in \sem s (a)$, $w$ is
of the form $\vpom \cdot a$ for $\vpom$ a pomset over $\State\cup\Act\cup 2^{\Pk}$.
\end{theoremEnd}
\begin{proofEnd}
  Both statements are proven with an easy induction on $p$ and $s$.
\end{proofEnd}

For non-empty sets of packets $a$ and $a'$, the global behavior of a state program without $\dup$ is identical on both inputs.
Let $2^{\Pk}_\nempty$ denote $2^{\Pk}\setminus{\{\emptyset\}}$.
\begin{theoremEnd}[default, category=syntax]{lemma}\label{lemma:stateterms}
Let $s\in \mathcal{T}_{\mathsf{state}}(\OO\cup\Act\cup 2^{\Pk})$. For all $a,a'\in 2^{\Pk}_\nempty$ we have   \linebreak[4]
$\big\{\upom \mid \upom\cdot b\in \sem s (a)\big\}=\big\{\upom \mid \upom \cdot b\in \sem s (a')\big\}$.
\end{theoremEnd}
\begin{proofEnd}
  The result is proven with an easy induction on $s$.
\end{proofEnd}

We characterize $\bsem{-}$ in terms of its behavior on subsets of the input set.
\begin{theoremEnd}[default, category=syntax]{lemma}\label{lemma:booleansetsofpackets}
Let $t\in\BB$ and $a,b\subseteq \Pk$. Then $\bsem t(a\cup b)=\bsem t(a)\cup \bsem t(b)$.
\end{theoremEnd}
\begin{proofEnd}
We prove this by induction on $t$. The cases where $t=\pfalse$ or $t=\ptrue$ are trivial. For $t=\match f n$, we have
that $\bsem t (a\cup b)=(a\cup b)(\match f n)=\{\pi\in (a\cup b) \mid \pi(f)=n\}=\{\pi\in a \mid
\pi(f)=n\}\cup\{\pi\in b \mid \pi(f)=n\}=\bsem t(a)\cup \bsem t(b)$.

For the negation we derive
\begin{align*}
\bsem {\neg t}(a\cup b)&=(a\cup b)\setminus {\bsem t(a\cup b)} \stackrel {\textrm{IH}}= (a\cup b)\setminus{(\bsem t (a) \cup \bsem t (b))} \\
&= a\setminus {\bsem t(a)}\cup a\setminus{\bsem t (b)} = \bsem {\neg t} (a)\cup \bsem {\neg t}(b)
\end{align*}
Note that if we have $\pk\in a$ such that $\pk\notin \bsem t (a)$, then also $\pk \notin \bsem t (b)$ for any set of packets $b$, because if $\pk\notin \bsem t (a)$ while $\pk\in a$, we know that $\pk$ does not satisfy predicate $t$.

\begin{align*}
\bsem {t\vee_{\BB}t'}(a\cup b)&= \bsem t (a\cup b) \cup \bsem t' (a\cup b)  \\
&\stackrel {\textrm{IH}}=( \bsem t (a) \cup \bsem t(b))\cup (\bsem {t'} (a)\cup \bsem {t'}(b)) \\
&= (\bsem t (a) \cup \bsem {t'}(a))\cup (\bsem t(b) \cup  \bsem {t'}(b)) \\
&= \bsem {t\vee_{\BB}t'} (a)\cup \bsem {t\vee_{\BB}t'} (b)
\end{align*}

\begin{align*}
\bsem {t\wedge_{\BB}t'}(a\cup b)&= \bsem t (a\cup b) \cap \bsem t' (a\cup b) \\
&\stackrel {\textrm{IH}}=( \bsem t (a) \cup \bsem t(b))\cap (\bsem {t'} (a)\cup \bsem {t'}(b)) \\
&= (\bsem t (a) \cap \bsem {t'}(a))\cup (\bsem t(b) \cap  \bsem {t'}(b)) \\
&= \bsem {t\wedge_{\BB}t'} (a)\cup \bsem {t\wedge_{\BB}t'} (b)
\end{align*}
In the middle step we use that if a packet $\pk\in \bsem t (a) \cup \bsem t(b)$ and $\pk\in \bsem {t'} (a)\cup \bsem
{t'}(b)$, then we know that $\pk$ satisfies predicate $t$ and $t'$, and w.l.o.g. $\pk$ is an element of $a$. Then also
$\pk \in \bsem t (a) \cap \bsem {t'}(a)$. Similarly, if $\pk \in \bsem t (a) \cap \bsem {t'}(a)$ or $\pk \in \bsem t (b) \cap \bsem {t'}(b)$ we can w.l.o.g.\ assume the former. We then know that $\pk$ satisfies both $t$ and $t'$ and
hence is an element of $\bsem t (a) \cup \bsem t(b)$ and of $\bsem {t'} (a) \cup \bsem {t'}(b)$.
\end{proofEnd}

Lastly, we have a lemma characterising the semantics of a deterministic packet program in terms of its behavior on subsets of the input.
\begin{theoremEnd}[default, category=syntax]{lemma}\label{lemma:splitsets}
Let $x\in\mathcal{T}_{\mathsf{det-pack}}$ and $a,b\subseteq \Pk$. Then \\
$\sem x (a\cup b) = \big\{\onepom\cdot (c\cup d) \mid \sem x (a)=\{\onepom \cdot c\},  \sem x (b)=\{\onepom \cdot d\}\big\}
$.
\end{theoremEnd}
\begin{proofEnd}
We prove this by induction on $x$. For $x=t\in \BB$ we use \cref{lemma:booleansetsofpackets}.
\begin{align*}
  \sem t (a\cup b) &= \{\onepom\cdot \bsem t (a\cup b)\} = \{\onepom \cdot \bsem t (a)\cup \bsem t (b)\} \tag{\cref{lemma:booleansetsofpackets}} \\
  &=\{\onepom \cdot (c\cup d)\mid c=\bsem t (a), d= \bsem t (b)\} \\
  &=\{\onepom \cdot (c\cup d)\mid \sem t(a)=\{\onepom\cdot c\}, \sem t(b)=\{\onepom\cdot d\}\}
\end{align*}
Let $x= \modify f n$.
\begin{align*}
  \sem {\modify f n} (a\cup b) &= \{\onepom\cdot  (a\cup b)(\modify f n)\} \\
  &= \{\onepom\cdot  \{\pk[n/f]\mid \pk\in (a\cup b)\}\} \\
  &=\{\onepom\cdot  (\{\pk[n/f]\mid \pk\in a\}\cup \{\pk[n/f]\mid \pk\in b\})\}\\
  &=\{\onepom\cdot (c\cup d)\mid c=\{\pk[n/f]\mid \pk\in a\}, d= \{\pk[n/f]\mid \pk\in b\}\}\\
  &=\{\onepom\cdot (c\cup d)\mid \sem{\modify f n}(a)=\{\onepom\cdot c\}, \sem{\modify f n}(b)=\{\onepom\cdot d\}\}
\end{align*}
For $ x \parallel y$ we derive
\begin{align*}
\sem {x\parallel y}(a\cup b)&=\{1\cdot (c\cup d)\mid 1\cdot c\in \sem x(a\cup b), 1\cdot d \in \sem y(a\cup b)\} \\
&\stackrel {\textrm{IH}}= \{1\cdot (c\cup d)\mid \\
& 1\cdot c\in
\{1\cdot (e\cup f)\mid \sem x(a)=\{1\cdot e\}, \sem x(b)=\{1\cdot f\}\}, \\
& 1\cdot d \in \{1\cdot (g\cup h)\mid \sem y(a)=\{1\cdot g\},\sem y(b)=\{1\cdot h\}\}\} \\
&=\{1\cdot (e\cup f\cup g\cup h) \mid \sem x(a)=\{1\cdot e\}, \sem x(b)=\{1\cdot f\},\\
& \sem y(a)=\{1\cdot g\},\sem y(b)=\{1\cdot h\}\} \\
&=\{1\cdot (c\cup d)\mid \sem {x\parallel y}(a)=\{1\cdot c\},
\sem {x\parallel y}(b)=\{1\cdot d\}\}
\end{align*}
where we use that $\sem {x\parallel y}(a)=\{1\cdot (e\cup g)\mid 1 \cdot e\in \sem x(a), 1\cdot g\in \sem y(a) \}$ and
$\sem {x\parallel y}(b)=\{1\cdot (f\cup h)\mid 1 \cdot f\in \sem x(b), 1\cdot h\in \sem y(b) \}$.

For $ \pseq x y$ we derive
\begin{align*}
&\sem {\pseq x y}(a\cup b)=\{1\cdot c\mid 1\cdot a'\in \sem x(a\cup b), 1\cdot c \in \sem y (a')\} \\
&\stackrel {\textrm{IH}}= \{1\cdot c\mid 1\cdot a'\in \{1\cdot (d\cup e)\mid \sem x(a)=\{1\cdot d\}, \sem x(b)=\{1\cdot e\}\}, 1\cdot c \in \sem y (a')\}\\
&=\{1\cdot c\mid  \sem x(a)=\{1\cdot d\}, \sem x(b)=\{1\cdot e\}, 1\cdot c \in \sem y (d\cup e)\} \\
&\stackrel {\textrm{IH}}= \{1\cdot c\mid  \sem x(a)=\{1\cdot d\},
\sem x(b)=\{1\cdot e\},\\
& 1\cdot c \in \{1\cdot (f\cup g)\mid \sem y(d)=\{1\cdot f\}, \sem y(e)=\{1\cdot g\}\}\} \\
&= \{1\cdot (f\cup g)\mid  \sem x(a)=\{1\cdot d\}, \sem x(b)=\{1\cdot e\}, \\
& \hspace{24mm} \sem y(d)=\{1\cdot f\}, \sem y(e)=\{1\cdot g\}\} \\
&= \{1\cdot (f\cup g)\mid  \sem {\pseq x y}(a)=\{1\cdot f\},
\sem {\pseq x y}(b)=\{1\cdot g\}\} \tag*{\qedhere}
\end{align*}
\end{proofEnd}

\subsection{Is {\normalfont\cnetkat} conservative over {\normalfont\netkat} and {\normalfont\textsf{POCKA}}?}\label{sec:relation}
\cnetkat combines \netkat and \POCKA, so it is natural to ask whether it is a conservative extension of either language. It turns out that the answer is positive for \POCKA, and for a fragment of \netkat. We start by recalling the semantics of \netkat\cite{netkat}. Note that \netkat expressions are packet programs without $\parallel$.
\begin{definition}[{\normalfont\netkat} semantics]
Let $\pk\in\Pk$, $t\in\BB$ and $p,q$ {\normalfont\netkat} terms.
\begin{mathpar}
\sem{t}_{\mathsf{NK}}(\pk)= \bsem t (\{\pk\}) \and
\sem{\ptrue}_{\mathsf{NK}}(\pk) = \{\pk\} \and
\sem {\pfalse}_{\mathsf{NK}}(\pk) = \{\} \and
\sem {\modify f n}_{\mathsf{NK}}(\pk) = \{\pk[n/f]\} \and
\sem{p\mathbin{;}q}_{\mathsf{NK}}(\pk)
    = \bigcup\nolimits_{\pk'\in \sem p(\pk)}\sem q_{\mathsf{NK}}(\pk') \and
\sem{p^*}_{\mathsf{NK}}(\pk)=\bigcup\nolimits_n \sem {p^n}_{\mathsf{NK}}(\pk) \and
\sem{p+q}_{\mathsf{NK}}(\pk) = \sem p_{\mathsf{NK}}(\pk)\cup\sem q_{\mathsf{NK}}(\pk)
\end{mathpar}
\end{definition}

\begin{theoremEnd}[default,category=relation]{theorem}
Take $\pi\in\Pk$ and \netkat term $p$.
$\sem p_{\mathsf{NK}}(\pi)=\bigcup_{\onepom\cdot a' \in \sem p (\{\pi\})} a'
$.
\end{theoremEnd}
\begin{proofEnd}
  We prove this by induction on the structure of $p$.
The case for predicates is straightforward as $\bsem{t} (\{\pk\}) =  \sem t_{\mathsf{NK}}(\pk)$.
We derive
\begin{align*}
\sem t_{\mathsf{NK}} (\pi)&= \bsem{t} (\{\pi\}) = \bigcup \{a'\mid \onepom \cdot a' \in \{\onepom\cdot
\bsem{t}(\{\pi\})\}\} \\
&= \bigcup \{a'\mid \onepom \cdot a' \in \sem t (\{\pi\})\}
\end{align*}

For the packet change we also get the result immediately as $\sem {\modify f n}_{\mathsf{NK}}(\pi)= \{\pi[n/f]\}$ and $\sem {\modify f n}(\{\pi\})=\{\onepom\cdot\{\pi[n/f]\}\}$.
For the inductive cases assume that $\sem p_{\mathsf{NK}}(\pi)=\bigcup \{a' \mid \onepom\cdot a' \in \sem p (a)\}$ and
$\sem q_{\mathsf{NK}}(\pi)=\bigcup \{a' \mid \onepom\cdot a' \in \sem q (a)\}$ (IH).

For $p+q$ we calculate:
\begin{align*}
\sem {p+q}_{\mathsf{NK}}(\pi) &=  \sem {p}_{\mathsf{NK}}(\pi) \cup \sem {q}_{\mathsf{NK}}(\pi) \\ \stackrel {\textrm{IH}}=& \bigcup \{a' \mid \onepom\cdot a' \in \sem p (a)\} \cup \bigcup \{a' \mid \onepom\cdot a' \in \sem q (a)\} \\
  &=  \bigcup \{a' \mid \onepom\cdot a' \in \sem p (a)\cup \sem q (a)\} \\
  &= \bigcup \{a' \mid \onepom\cdot a' \in \sem {\punion p q} (a)\}
\end{align*}

  For $\pseq p q$, we first note that for packet programs without parallel, the semantics of $p$ on a single input packet, consists of pairs $\onepom\cdot b$ such that $b=\{x\}$ for $x\in\Pk$ or $b=\emptyset$.
We derive
  \begin{align*}
  \sem {\pseq p q}_{\mathsf{NK}}(\pi)&=
  \{ b\mid b'\in \sem p_{\mathsf{NK}}(\pi), b\in \sem q_{\mathsf{NK}}(b') \} \\
    &\stackrel {\textrm{IH}}=\{ b\mid b'\in \bigcup \{a' \mid \onepom\cdot a' \in \sem p (\{\pi\})\},  b \in \bigcup \{c' \mid \onepom\cdot c' \in \sem q (\{b'\})\}\}\\
    &= \{b \mid \onepom \cdot \{b'\}\in \sem p (\{\pk\}), \onepom\cdot \{b\}\in \sem q (\{b'\})\} \tag{$a'=\{x\}$ for some $x\in \Pk$}\\ 
    &= \bigcup \{\{b\}\mid \onepom \cdot \{b'\}\in \sem p (\{\pk\}), \onepom\cdot \{b\}\in \sem q (\{b'\})\} \\
    &= \bigcup \{a' \mid \onepom\cdot c\in \sem p (a), \onepom\cdot a'\in \sem q (c)\}\\
    &= \bigcup\{a' \mid \onepom\cdot a' \in \{\onepom \cdot a' \mid \onepom \cdot c\in \sem p (a), \onepom \cdot a'\in \sem q (c)\}\}\\
    &= \bigcup \{ a' \mid \onepom\cdot a'\in \sem {\pseq p q}(a)\}
\end{align*}

  For $p^*$ we first prove by induction on $n$ that if $\sem p_{\mathsf{NK}}(\pi)=\bigcup \{a' \mid \onepom\cdot a' \in \sem p (a)\}$, then for all $n$ we have $\sem {p^{(n)}}_{\mathsf{NK}}(\pi)=\bigcup \{a'
  \mid \onepom\cdot a' \in \sem p^{(n)} (a)\}$. The base case is trival as $\sem {p^0}_{\mathsf{NK}}(\pi)= \{\pi\}$, and $\sem{p^{(0)}}(a)= \sem {\skp} (a)=\{ 1\cdot \{\pi\}\}$.
  For the inductive step, the induction hypothesis is
  $\sem {p^{(n)}}_{\mathsf{NK}}(\pi)=\bigcup \{a'
  \mid \onepom\cdot a' \in \sem{ p^{(n)}} (a)\}$.
  Then we derive

  \begin{align*}
   \sem {p^{(n+1)}}_{\mathsf{NK}}(\pi) &= \sem{ \pseq p p^{(n)}}_{\mathsf{NK}}(\pi) \\
    &= \bigcup\{a'\mid \onepom\cdot a'\in \sem {\pseq p p^{(n)}} (a)\} \tag{IH \text{ and the case for }$\pseq p q$} 
  \end{align*}

  Then we can derive:
  \begin{align*}
    \sem {p^*}_{\mathsf{NK}}(\pi) &= \displaystyle\bigsqcup_{n \in \N} \lang{p^{(n)}}_{\mathsf{NK}}(\pi)\\
    &= \displaystyle\bigsqcup_{n \in \N} \bigcup \{
   a'\mid \onepom \cdot a' \in \sem {p^{(n)}} (a)\} \tag{result above} \\
    &= \bigcup \{a' \mid \onepom\cdot a'\in \displaystyle\bigsqcup_{n \in \N} \sem {p^{(n)}}(a) \} \\
    &= \bigcup \{a' \mid \onepom\cdot a'\in  \sem {p^*}(a) \}\tag*{\qedhere}
  \end{align*}
\end{proofEnd}
We can derive a further relation between the semantics if we assume there is no use of $+$ and $*$ (the proof uses \cref{lemma:booleansetsofpackets}).
\begin{theoremEnd}[default,category=relation]{lemma}%
  \label{lemma:seqpacket}
  Let $p$ be built out of packet predicates and modifications ($\modify f n$), and their sequential composition. Then
  $\sem p (a) = \left\{\onepom \cdot \bigcup_{x\in a} \sem p_{\mathsf{NK}} (x)\right\}
$.
\end{theoremEnd}
\begin{proofEnd}
  If $a=\emptyset$, the result follows immediately. If $a\neq\emptyset$, we prove this by induction on $p$. For $p=t\in\BB$ we derive
    \begin{align*}
      \sem t (a) &= \{\onepom\cdot \bsem t (a)\} \\
      &= \{\onepom \cdot \bigcup_{x\in a} \bsem t (\{x\})\}\tag{\cref{lemma:booleansetsofpackets}} \\
      &= \{\onepom \cdot \bigcup_{x\in a} \sem{t}_{\mathsf{NK}}(x)\}\tag{Definition \netkat}
    \end{align*}

  For $p=\modify f n$ we derive
  \begin{align*}
    \sem {\modify f n} (a)&= \{\onepom\cdot a(\modify f n)\}\\
    &=\{\onepom\cdot \{x[n/f] \mid x\in a\}\}\\
    &=\{\onepom \cdot \bigcup_{x\in a}\{x[n/f]\}\}\\
    &=\{\onepom \cdot \bigcup_{x\in a}\sem{\modify f n}_{\mathsf{NK}}(x)\}
  \end{align*}

  In the inductive step we derive
  \begin{align*}
  \sem {p\mathbin{;} q} (a) &=\{(\upom\cdot\vpom )\cdot b \mid \upom \cdot a' \in \sem p (a), \vpom\cdot b \in \sem q (a')\}\\
  &\stackrel {\textrm{IH}}= \{\onepom\cdot b \mid \onepom\cdot a' \in \sem p (a)\\
  &= \{\onepom \cdot \bigcup_{x\in a} \sem
  p_{\mathsf{NK}} (x)\}, \onepom \cdot b \in \sem q (a')= \{\onepom \cdot \bigcup_{x\in a'} \sem q_{\mathsf{NK}} (x)\}\}\tag{\cref{lemma:statepacketterms}}\\
  &=\{\onepom \cdot \bigcup_{x\in a'} \sem q_{\mathsf{NK}} (x) \mid a'= \bigcup_{y\in a} \sem
  p_{\mathsf{NK}} (y)\}\\
  &=\{\onepom\cdot \bigcup \{\sem q_{\mathsf{NK}} (x) \mid x\in\bigcup_{y\in a} \sem
  p_{\mathsf{NK}} (y)  \}\} \\
  &=\{\onepom\cdot \bigcup_{y\in a} \bigcup\{\sem
  q_{\mathsf{NK}}(x)\mid x\in \sem p_{\mathsf{NK}}(y)\}\} \\
  &=\{\onepom\cdot \bigcup_{y\in a} \sem{p\mathbin{;}q}_{\mathsf{NK}}(y)\}  \tag*{\qedhere}
  \end{align*}
\end{proofEnd}

It is worth remarking that the equational theories of \netkat and \cnetkat are not equivalent: there are equivalent programs in \netkat, that cannot be proved equivalent with the \cnetkat axioms, as the following example illustrates.
Consider the program $p+\pfalse$ for $p$ a packet program without parallel. In \netkat, because the $+$ is interpreted as multicast, this program is provably equivalent to $p$: executing $p$ on your input packet while at the same time also dropping a copy of the input, has the same outcome as just executing $p$.
In \cnetkat, however, this is not the case. Instead, the $+$-operator is interpreted as non-deterministic choice and in the semantics of $p+\pfalse$ we get the trace
$\onepom\cdot\emptyset$, representing the choice of dropping all the packets, which is not present
in the semantics of $p$. Hence, this axiom is unsound ($p+\pfalse
\not\equiv p$), and instead the alternative axiom $p\parallel\drp = p$ holds, reflecting the fact that $\parallel$ is multicast.

We now show \cnetkat semantics is equivalent to the \POCKA semantics on state programs. In~\cite{pocka}, \POCKA terms
are what we defined as state programs over the alphabet $\OO\cup\Act$, and they are interpreted in terms of pomset languages over assignments and states, encoded as partial functions, similarly to separation logic~\cite{reynolds}.
The \POCKA semantics are defined in two steps: the first step results in a set containing all pomsets that can be derived directly from the terms, and in a second step this set is closed under two laws---$\hexch$ and $\hcontr$---that account for all traces that can be built in parallel threads (including simple interleaving).

\begin{definition}[\POCKA semantics]\label{def:pockasem}
Let $o\in \OO$, $e\in\Act$, $p,q \in \mathcal{T}_{\mathsf{state}}(\OO \cup \Act)$.
\[
\begin{array}{rcl@{\quad}rcl@{\quad}rcl@{\quad}rcl}
 \llparenthesis o
\rrparenthesis
&=&\State^*\odot\osem {o}\odot\State^*
&\llparenthesis p\mathbin{;} q
\rrparenthesis&=&\llparenthesis p
\rrparenthesis \odot \llparenthesis q
\rrparenthesis
& \llparenthesis \skp \rrparenthesis &=& \{\onepom\} &\llparenthesis \abort\rrparenthesis &=&\emptyset \\
\llparenthesis e
\rrparenthesis&=&\State^*\odot\{e\}\odot\State^*
&\llparenthesis p\parallel q
\rrparenthesis&=&\llparenthesis p
\rrparenthesis \parallel \llparenthesis q
\rrparenthesis
 &
 \llparenthesis p^*
 \rrparenthesis&=&\llparenthesis p
 \rrparenthesis^*
 & \llparenthesis p+q
 \rrparenthesis&=&\llparenthesis p
 \rrparenthesis\cup \llparenthesis q
 \rrparenthesis
\end{array}
\]
The semantics of a \POCKA expression $p$ is $\sem
{p}_{\textrm{\POCKA}}=\closure[\hexch\cup\hcontr]{\llparenthesis p
\rrparenthesis}$.
\end{definition}

\begin{theoremEnd}[default,category=relation]{theorem}\label{theorem:pockarel}
\cnetkat is a conservative extension of \POCKA:\@ if $p$ is a \POCKA term ($p\in \mathcal{T}_{\mathsf{state}}(\OO\cup\Act)$) then for $a\neq\emptyset$,
$
\closure[]{\sem p }(a) = \left\{ \upom \cdot a \mid \upom \in \sem p_{\textrm{\POCKA}}\right\}
$.
\end{theoremEnd}
\begin{proofEnd}
We first show by induction on the structure of $p$ that the theorem holds without closures. We show that
\begin{equation}\label{poceq}
\sem p (a) = \{ \upom \cdot a \mid \upom \in \llparenthesis p\rrparenthesis\}
\end{equation}
 The case for observations is straightforward as $\sem{o}(a) =  \State^* \cdot \osem{o} \cdot \State^* \times \{ a \}  $ and
$ \State^* \cdot \osem{o} \cdot \State^*  =  \llparenthesis o\rrparenthesis$. For state actions $e$ we also immediately obtain the result: $\sem{e} (a)$ is
exactly defined as $ \llparenthesis e \rrparenthesis\times
\{ a \} $. The same for the case for $\abort$ and
$\skp$, which are called $0$ and $1$ in \POCKA. For 
the inductive cases assume that $\sem p (a)
= \{\upom \cdot a \mid \upom \in \llparenthesis p\rrparenthesis\}$ and $\sem q (a) = \{ \vpom \cdot a \mid \vpom \in \llparenthesis q\rrparenthesis\}$ (IH).

For $p+q$ we calculate:
\begin{align*}
\sem {p+q} (a) &=   \lang{\polp}(a) \cup \lang{\polq}(a)\stackrel {\textrm{IH}}=  \{\upom \cdot a \mid \upom \in \llparenthesis p\rrparenthesis\} \cup \{ \vpom \cdot a \mid \vpom \in \llparenthesis q\rrparenthesis\}\\
 &=  \{\upom \cdot a \mid \upom \in \llparenthesis p\rrparenthesis\cup\llparenthesis q\rrparenthesis\} = \{\upom \cdot a \mid \upom \in \llparenthesis p+q\rrparenthesis\}
\end{align*}

For $\pseq p q$, we have:
\begin{align*}
\sem {\pseq p q} (a) &=  \{ (\upom\cdot \vpom) \cdot b \mid \upom\cdot  a' \in \lang{\polp}(a), \vpom \cdot b\in \lang{\polq}(a')\}\\
&\stackrel {\textrm{IH}}=  \{ (\upom\cdot \vpom) \cdot a \mid  \upom \in \llparenthesis p\rrparenthesis,  \vpom \in\llparenthesis q\rrparenthesis\}\\
&=  \{ (\upom\cdot \vpom) \cdot a \mid  (\upom\cdot \vpom)  \in \llparenthesis p\mathbin{;}q\rrparenthesis\}\\
\end{align*}

For $p \parallel q$, we have:
\begin{align*}
\sem { p \parallel q} (a) &=  \{ (\upom\parallel \vpom) \cdot (b_1\cup b_2) \mid \upom\cdot  b_1 \in \lang{\polp}(a), \vpom \cdot b_2\in \lang{\polq}(a)\}\\
&\stackrel {\textrm{IH}}=  \{ (\upom\parallel \vpom) \cdot a  \mid  \upom \in \llparenthesis p\rrparenthesis,  \vpom \in \llparenthesis q\rrparenthesis\}\\
&=  \{ (\upom\parallel \vpom) \cdot a \mid  (\upom\parallel \vpom)  \in \llparenthesis p\parallel q\rrparenthesis \}\\
\end{align*}

For the case with the star, first note that from $\sem p (a)
= \{\upom \cdot a \mid \upom \in \llparenthesis p\rrparenthesis\}$ and the case for sequential composition, we can derive
$\sem {p^{(n)}} (a) = \{\upom\cdot a \mid \upom \in \llparenthesis p^{(n)}\rrparenthesis \}$.
\end{proofEnd}

\subsection{Axiomatization}\label{sec:theory}
We introduce notation to describe packets and sets of packets axiomatically. Let $f_1,\dots,f_k$ be a list of all fields of a packet in some fixed order. Then for each tuple $\overline{n}=n_1,\dots,n_k$ we obtain expressions $f_1 = n_1 \cdots f_k = n_k$ and $f_1\gets n_1 \cdots f_k \gets n_k$, which, similar to \netkat, we call \emph{complete tests} and \emph{complete assignments}. Complete tests are also referred to as atoms, because they
are the atoms of the Boolean algebra generated by the tests. We denote the set of atoms by $\At$, complete tests with $\alpha$ and
complete assignments with $\pi$. There is a one-to-one correspondence between complete tests and assignments according to the values of $\overline{n}$. For $\alpha\in\At$ we denote the corresponding complete assignment by $\pi_{\alpha}$, and if $\pi$ is a complete assignment we denote the corresponding atom by $\alpha_{\pi}$.

There is also a link between sets of packets and terms of the form $\Vertt_{i\in I} \pi_i$.
For each set of packets $a$, we take the set $\{\pi_i \mid i\in I\}$ of complete assignments such that each $\pi_i$ corresponds to a packet of $a$, and combine them in parallel. Formally, for a set of packets $a$ there exists an expression $\Vertt_{i\in I} \pi_i$, that we denote with $\Pi_a$, such that on any input $b\neq\emptyset$, $\sem{\Pi_a}(b)=\{\onepom \cdot a\}$.
Similarly, the semantics of an expression of the form $\Vertt_{i\in I} \pi_i$ on any input is always $\{\onepom \cdot a\}$ for some $a\in 2^{\Pk}$.
 We use the notation $\Pi_a$ as a syntactic representation of set of packets $a$.

\cnetkat has the structure of a Kleene algebra on state programs,
enriched with additional axioms. Tests form a Boolean algebra and state observations a pseudocomplemented distributive lattice (PCDL). The test and observation structures are subject to interaction constraints. The packet processing behavior is captured by the packet axioms, which contain axioms for individual packets and sets of packets. The axioms governing the parallel operator are partially familiar from earlier work on \BKA~\cite{hoare-moeller-struth-wehrman-2009,laurence-struth-2014}. There is also the exchange law familiar from \CKA. Lastly, we have axioms for the interactions between state programs and packet programs. The full set of axioms is described in \Cref{fig:axioms}. We write $\equiv$ for the smallest congruence on 
$\programs$ generated by the axioms in~\cref{fig:axioms}.

\begin{remark}[When is $\Pi_a$ equal to $\pfalse$?]\label{remark:piadrop}
  $\Pi_a\equiv \pfalse$  if and only if $a$ is empty.   $\Pi_\emptyset = \Vert_{i\in \emptyset} \pi_i\equiv \Vert \emptyset \equiv \bigvee \emptyset \equiv\pfalse$. For all other $a$, we have $\Pi_a\not\equiv \pfalse$.
\end{remark}

\begin{figure*}[!ht]\small\centering\noindent%
  \fbox{\noindent%
    \begin{minipage}[t]{.479\textwidth}
      \noindent%
      \textbf{Kleene Algebra axioms}
      \hfill $s\in\mathcal{T}_{\mathsf{state}}$
              \vspace{-.8em}
      \[
      \begin{array}{rcl}
        p + (q + r) & \equivbka & (p + q) + r \\
        p + q  & \equivbka & q + p   \\
        p + \abort                             & \equivbka & p\\
        p + p                       &       \equivbka & p\\
        \pseq p  {(\pseq q r)}               & \equivbka&  \pseq {(\pseq p q)} r     \\
                \pseq s  \abort                             & \equivbka & \abort  \\
               \pseq  \abort p     &\equivbka& \abort \\
       \pseq p \skp                         & \equivbka & p \equivbka \pseq  \skp  p
        \\
       \pseq  p {(q+r)} & \equivbka & \pseq p q+ \pseq p r \\
        \pseq {(p + q)} r & \equivbka & \pseq p  r + \pseq q r\\
        p^* &\equiv& \skp + pp^* \\
        \pseq {p + q}  r \leqqbka q & \Rightarrow & p \cdot r^* \leqqbka q\\
        p^* &\equiv& \skp + p^*p \\
        p + q \cdot r \leqqbka r  &\Rightarrow & q^* \cdot p \leqqbka r
      \end{array}\]
         \vspace{-.7em}
     \hrule
     \vspace{.5em}
         \noindent%
      \textbf{Packet axioms}
            \hfill $x\in\mathcal{T}_{\mathsf{det-pack}}$
                  \vspace{-.8em}
                  \[
     \begin{array}{rcl}
       \pseq {\match f n} {\modify {f'} m}& \equiv &  \pseq {\modify {f'} m} {\match {f} n}  \quad (f\neq f')  \\
       \pseq {\modify f n} {\modify {f'} m}& \equiv  & \pseq {\modify {f'} m} {\modify {f} n}  \quad (f\neq f')  \\
       \pseq {\match f n} {\modify f n} &\equiv &\match f n  \\
       \pseq  {\modify f n} {\match f n} &\equiv &{\modify f n}\\
              \pseq  {\modify f m} {\modify f n} &\equiv &{\modify f n}\\
          x \parallel x &\equiv &x  \\
          \pseq  x {(p\parallel q)} & \equiv& (\pseq  x p)\parallel (\pseq  x q)  \\
          \pseq {(p\parallel q)} x & \equiv& (\pseq p x)\parallel (\pseq q x)
        \end{array}\]
        \vspace{-.8em}
  \hrule
  \vspace{.5em}
  \noindent%
  \textbf{Local vs global state} $y,z\in\mathcal{T}_{\mathsf{packet}}$,  $s,v\in\mathcal{T}_{\mathsf{state}}$,
   $w\in\mathcal{T}_{\mathsf{state}} (\OO\cup\Act\cup 2^{\Pk})$
         \vspace{-.8em}
         \[
  \begin{array}{rcl@{\quad}l}
    \pseq {\Pi_a} {\dup} &\equiv & \pseq {\Pi_a} {a} &(a\in 2^{\Pk})\\
   \pseq {\Pi_a} {w} & \equiv & \pseq {w} {\Pi_a} &(a\in 2^{\Pk}_\nempty )\\
   \pseq {\pfalse} {p} &\equiv & \pfalse &
   \pseq y \pfalse  \equiv \pfalse  \\
    s\parallel \skp &\equiv &s \\
    (\pseq s y) \parallel (\pseq v z) & \equiv &\multicolumn{2}{l}{\pseq {(s\parallel v)} {(y\parallel z)}}
  \end{array} \]
  \vspace{-.3em}
  \hrule
  \vspace{.5em}
  \noindent
  \textbf{Extensionality}\\
  $\forall a\in 2^{\Pk}.  (\Pi_a \mathbin{;} p \equiv \Pi_a \mathbin{;}q) \Rightarrow p\equiv q$
    \end{minipage}\hspace{.005\linewidth}\vrule\hspace{.005\linewidth}%
    \begin{minipage}[t]{.482\linewidth}
      \noindent%
      \textbf{Parallel axioms}
        \vspace{-.8em}
        \[
      \begin{array}{rcl}
        p \parallel (q \parallel r) & \equivbka &  (p \parallel q) \parallel r \\
        p \parallel \abort & \equivbka&  \abort \\
        \drp \parallel p & \equiv& p \\
        p \parallel (q + r) & \equivbka & p \parallel q + p \parallel r \\
       p \parallel q & \equiv & q \parallel p
     \end{array}
     \]
       \vspace{-.8em}
      \hrule
      \vspace{.3em}
      \noindent%
      \textbf{Exchange law}
      \hfill $s,s',v,v'\in \mathcal{T}_{\mathsf{state}}$
        \vspace{-1em}
     \begin{align*}
       \pseq {(s \parallel s')} { (v \parallel v')} &\leqqpocka (\pseq s v) \parallel (\pseq {s'}  {v'})
     \end{align*}
     \vspace{-1.4em}
    \hrule
          \vspace{.5em}
      \noindent%
      \textbf{Packet pred., state obs.\ axioms}\\
     \hfill $\boldsymbol{\vee}\in\{\vee,\vee_{\BB}\}$, $\boldsymbol{\wedge}\in\{\wedge,\wedge_{\BB}\}$, $a,b,c\in \BB\cup\OO$
           \vspace{-.8em}
           \[
      \begin{array}{rcl}
        a \boldsymbol{\wedge} b & \equivpl & b \boldsymbol{\wedge} a \\
        a \boldsymbol{\wedge} (b \boldsymbol{\wedge} c) & \equivpl & (a \boldsymbol{\wedge} b) \boldsymbol{\wedge} c \\
        a \boldsymbol{\vee} (a \boldsymbol{\wedge} b) & \equivpl & a \equivpl
                             a \boldsymbol{\wedge} (a \boldsymbol{\vee} b) \\
        a \boldsymbol{\vee} (b\boldsymbol{\wedge} c) & \equivpl&  (a \boldsymbol{\vee} b) \boldsymbol{\wedge} (a \boldsymbol{\vee} c) \\
        a \boldsymbol{\wedge} (b \boldsymbol{\vee} c) & \equivpl  &(a\boldsymbol{\wedge} b) \boldsymbol{\vee} (a \boldsymbol{\wedge} c)
      \end{array}
      \]
      \vspace{-.4em}
      \hrule
      \vspace{.5em}
  \noindent%
  \textbf{Additional state obs.\ axioms}
  \vspace{-.8em}
  \[
  \begin{array}{rcl@{\ \ }l}
      o &\equivpl&
                   o \wedge \top \\ 
    o\leqqpl \overline{o'} &\Leftrightarrow &o\wedge o'\equivpl \bot\\
    v = n &\wedge & v = m \equivpl  \bot  & (n\neq m )\\ 
    \overline{v=n} & \leqqpl & \bigvee\limits_{n\neq m}v=m \\
    \overline{\bigwedge_{i} v_i=n_i} &\leqqpl &\bigvee_{i} \overline{v_i=n_i}  &(i\neq j. v_i \neq v_j) 
  \end{array}
  \]
    \vspace{-.6em}
  \hrule
  \vspace{0.5em}
  \noindent%
  \textbf{Additional packet pred.\ axioms}
    \vspace{-.8em}
    \[
  \begin{array}{rcl}
    t\vee_{\BB}\ptrue &\equivpl & \ptrue \equivpl
                 t \vee_{\BB} \neg t \\
     t \land_{\BB} {\neg t} &\equivpl& \drp  \\ 
              \match f n \wedge_{\BB} \match f m & \equiv &\drp \quad (n\neq m)  \\
                  \bigvee_i \match f i &\equiv &\ptrue
  \end{array}
  \]
  \vspace{-.8em}
  \hrule
  \vspace{.5em}
    \noindent%
    \textbf{Interface axioms}
          \vspace{-.8em}
    \[
    \begin{array}{rlrlll}
      o \wedge o' & \leqqpocka \pseq o  {o'} &
      o \vee o'  &\equivpocka  o + o' & {(o,o'\in\termspl)} \\
     \abort &\equivpocka \bot&  \skp &\equiv \ptrue &{(e\in\Act)} \\
      \pseq \top  o & \leqqpocka   o &\pseq o  \top &\leqqpocka  o &{(t,t'\in\BB)} \\
      \pseq \top  e & \leqqpocka   e &\pseq e  \top &\leqqpocka e\\
               t\land_{\BB} t' &\equiv \pseq t t' & t\vee_{\BB}t' &\equiv t\parallel t'
    \end{array}
    \]
    \end{minipage}
  }
  \caption{%
    Axioms of \cnetkat.
    The left column contains the \KA axioms, the packet axioms, the axioms for the interaction between the local and global state, and an extensionality axiom. The right column axiomatizes the $\parallel$,
   the algebra of packet tests (which is a Boolean algebra), and the algebra of partial state observations (which is a PCDL). The interface axioms connect both the lattice operators to the Kleene algebra ones.
 We write $ e\leqqpl f $ as a shorthand for $e+f \equivpl f$.
  }%
  \label{fig:axioms}
\end{figure*}

There are a few subtleties to notice in \cref{fig:axioms}. First, we point out the interaction between $\pfalse$ and $\abort$. When no packets are present, not even $\abort$ can be executed. Hence, if we drop all packets and then $\abort$, the abort does not happen:
$
\pfalse \mathbin{;}\abort \equiv \pfalse
$.
On the other hand, if we first $\abort$ and then drop all the packets, the behavior is equal to just aborting:
$
\abort \mathbin{;}\pfalse \equiv \abort
$.

In the axioms of the parallel operator, the axiom $s\parallel \skp\equiv \skp$ from \BKA is missing; it only holds when $s$ is a state program, and can be found in the local state vs global state axioms. In addition to the familiar \BKA axioms, there is
the axiom $\pfalse\parallel p \equiv p$, in contrast with $\abort \parallel p \equiv \abort$.

The local state vs global state axioms capture the interactions between the global pomset and the output packets. The first one, $\Pi_a\mathbin{;}\dup \equiv \Pi_a\mathbin{;}a$, captures the intuition that if we know the input is $a$ (due to $\Pi_a$, which, as a parallel of complete assignments, essentially overwrites any non-empty input set to $a$), then we know the $\dup$ is recording an ``$a$''. The second axiom, $\Pi_a\mathbin{;}w \equiv w\mathbin{;}\Pi_a$ states that
for dup-free state program $w$, we can flip the order between changing the set of output packets or performing the state
changes in $w$, as long as $\Pi_a$ is not the parallel representing the empty set. This latter condition is crucial: if
$a=\emptyset$, then $\Pi_a\equiv\pfalse$, and $\pfalse \mathbin{;} w \equiv \pfalse$ (the global state changes in $w$ do not get executed if we have no packets).

The axiom $\pfalse \mathbin{;} p \equiv \pfalse$ for any program $p$ captures the intuition that if there are no packets, nothing happens anymore. The other way around, $y \mathbin{;}\pfalse\equiv \pfalse$ is only true for $y$ a packet program; if it was a state program, the global state changes get executed if we start with a non-empty set of input packets, making the behavior of $y\mathbin{;}\pfalse$ not equivalent to $\pfalse$.

Lastly, extensionality says that if two programs are equivalent on all inputs (i.e., $a\in 2^\Pk$), then the programs are equivalent. It is not clear whether this axiom is derivable from the others; we hope to settle this question in the future.
%

\begin{textAtEnd}[allend, category=completeness]
\begin{lemma}\label{lemma:funny-derivation}
  Let $a\in 2^{\Pk}_\nempty$ and $b\in 2^{\Pk}$. Then we can derive using the axioms of \cnetkat that
  \[\Pi_a \mathbin{;}\Pi_b \equiv \Pi_b\]
\end{lemma}
\begin{proof}
  We know there exists finite index sets $I,J$ and complete assignments $\pi_i,\pi_j'$ such that $\Pi_a\equiv\Vertt\limits_{i\in I} \pi_i$ and $\Pi_b\equiv\Vertt\limits_{j\in J} \pi'_j$. Then we can derive using the axioms of \cnetkat that
  \[\Vertt\limits_{i\in I} \pi_i \mathbin{;}\Vertt\limits_{j\in J} \pi'_j \equiv \Vertt\limits_{j\in J} \pi'_j\]
  \begin{align*}
    &\Vertt\limits_{i\in I} \pi_i \mathbin{;}\Vertt\limits_{j\in J} \pi'_j \\
    &\equiv \Vertt\limits_{j\in J}\big(\Vertt\limits_{i\in I} \pi_i \mathbin{;} \pi'_j\big) \tag{$\pseq  x {(p\parallel q)}  \equiv (\pseq  x p)\parallel (\pseq  x q)$ }\\ 
    &\equiv \Vertt\limits_{j\in J}\Vertt\limits_{i\in I} \big(\pi_i \mathbin{;} \pi'_j\big) \tag{$\pseq {(p\parallel q)} x  \equiv (\pseq p x)\parallel (\pseq q x)$ } \\ 
    &\equiv \Vertt\limits_{j\in J}\Vertt\limits_{i\in I}  \pi'_j \tag{$\pi\mathbin{;}\pi'\equiv \pi'$ (shown to follow from packet axioms in~\cite{netkat}) }\\ 
    &\equiv \Vertt\limits_{j\in J} \pi'_j \tag*{\qedhere}
  \end{align*}
\end{proof}
\end{textAtEnd}
\section{Soundness and Completeness}\label{sec:completeness}

In this section we prove soundness and completeness of the \cnetkat
semantics w.r.t.\ the axiomatization from
\Cref{fig:axioms}. For soundness, we prove that if programs $p$ and $q$ are provably equivalent using the axioms, they have the same semantics:

\begin{theoremEnd}[default,category=completeness]{theorem}[Soundness]\label{thm:sound}
For all $p,q\in\programs$, if $p\equiv q$, then $\closure[]{\sem p} = \closure[]{\sem q}$.
\end{theoremEnd}
\begin{proofEnd}
  We proceed by induction on $\equiv$.
In most cases we prove soundness of the axioms before closure, in which case the semantics is also sound after closure. For two axioms the closure is crucial for soundness, and we provide a soundness result only after the closure is computed. We prove that for all $a\in 2^{\Pk}$ we have $\closure[]{\sem p}(a)=\closure[]{\sem q}(a)$. If $a=\emptyset$, the result follows immediately as then $\sem p (a) =\{\onepom\cdot\emptyset\}=\sem q (a)$. If $a\neq\emptyset$, we give direct proofs for all the axioms.
Note that regardless of whether $a$ is $\emptyset$ or not $\sem {p+q}(a)=\sem p (a)\cup\sem q(a)$, $\sem{p\mathbin{;}q}(a)=\{(\upom\cdot\vpom)\cdot b \mid \upom \cdot a'\in \sem p (a),
\vpom\cdot b \in \sem q (a')\}$,
$\sem{p\parallel q}(a)=\{(\upom\parallel\vpom)\cdot (b\cup c )\mid \upom \cdot b\in \sem p (a),
\vpom\cdot c \in \sem q (a)\}$ and $\lang{\pstar\polp}(a) = \displaystyle\bigsqcup_{n \in \N} \lang{p^{(n)}}(a)$.

We start with the axioms of \KA. As $+$ is defined as $\cup$, associativity, commutativity and idempotence of the $+$ follow immediately. The axiom $p+\abort \equiv p$ is also immediately satisfied. For the other axioms we give a direct proof. 

We consider the axiom $\pseq p  {(\pseq q r)}\equivbka  \pseq {(\pseq p q)} r    $.
Let $a$ be a set of packets. We derive
\begin{align*}
\sem {\pseq p  {(\pseq q r)}} (a) &= \{(\upom \cdot \vpom)\cdot b \mid \upom \cdot a'\in \sem p (a), \vpom\cdot b\in \sem {\pseq q r}(a')\} \\
&= \{(\upom \cdot \vpom)\cdot b \mid \upom \cdot a'\in \sem p (a), \\
&\vpom\cdot b\in \{(\upom' \cdot \vpom')\cdot c \mid \upom' \cdot b'\in \sem q (a'), \vpom'\cdot c\in \sem r (b')\}\} \\
&= \{(\upom \cdot (\upom' \cdot \vpom ))\cdot b \mid \upom \cdot a'\in \sem p (a), \upom' \cdot b'\in\sem q (a') , \vpom \cdot b \in \sem r (b')\} \\
&= \{((\upom \cdot \upom') \cdot \vpom )\cdot b \mid \upom \cdot a'\in \sem p (a), \upom' \cdot b'\in\sem q (a') , \vpom \cdot b \in \sem r (b')\} \\
&= \{((\upom \cdot \upom') \cdot \vpom )\cdot b \mid (\upom \cdot \upom') \cdot b'\in \sem {\pseq p q} (a), \vpom \cdot b \in \sem r (b')\} \\
& = \{(\upom \cdot \vpom )\cdot b \mid \upom \cdot a'\in \sem {\pseq p q} (a), \vpom \cdot b \in \sem r (a')\} \\
&= \sem {\pseq {(\pseq p q)}  r} (a)
\end{align*}
In the third step we used the fact that sequential composition on pomsets is associative.

The next axiom we verify is $ \pseq s  \abort \equivbka  \abort  $ for $s$ a state program and $\pseq \abort p \equiv p$ for all programs $p$. We derive
\begin{align*}
\sem {\pseq s  \abort } (a)&=\{(\upom \cdot \vpom)\cdot b \mid \upom \cdot a'\in \sem s (a), \vpom\cdot b\in \sem {\abort}(a')\} \\
&=\{(\upom \cdot \vpom)\cdot b \mid \upom \cdot a\in \sem s (a), \vpom\cdot b\in \sem {\abort}(a)=\emptyset\} \tag{\cref{lemma:statepacketterms}}\\
&=\emptyset= \sem {\abort} (a)
\end{align*}

\begin{align*}
\sem {\abort}(a) &= \emptyset =\{(\upom \cdot \vpom)\cdot b \mid \upom \cdot a'\in \emptyset, \vpom\cdot b\in \sem p (a')\} = \sem {\pseq   \abort  p}(a)
\end{align*}

\begin{align*}
\sem {\pseq p  \skp } (a)&=\{(\upom \cdot \vpom)\cdot b \mid \upom \cdot a'\in \sem p (a), \vpom\cdot b\in \sem {\skp}(a')=\{\onepom\cdot a'\}\} \\
&=\{\upom \cdot a' \mid \upom \cdot a'\in \sem p (a)\} = \sem {p} (a) \\
&= \{(\upom \cdot \vpom)\cdot b \mid \upom \cdot a'\in \{\onepom\cdot a\}, \vpom\cdot b\in \sem p (a')\}\\
& = \{(\upom \cdot \vpom)\cdot b \mid \upom \cdot a'\in \sem \skp (a), \vpom\cdot b\in \sem p (a')\}= \sem {\pseq \skp p} (a)
\end{align*}

\begin{align*}
\sem {\pseq p {(\punion q r)}}(a)&=\{(\upom\cdot\vpom)\cdot b \mid \upom\cdot a'\in \sem p (a), \vpom \cdot b \in \sem q(a')\cup \sem r (a')\} \\
&=\{(\upom\cdot\vpom)\cdot b \mid \upom\cdot a'\in \sem p (a), \vpom \cdot b \in \sem q(a')\}\\
&\cup \{(\upom\cdot\vpom)\cdot b \mid \upom\cdot a'\in \sem p (a), \vpom \cdot b \in \sem r (a')\} \\
&= \sem {\pseq p q}(a)\cup \sem {\pseq p r}(a)= \sem {\punion {\pseq p q} {\pseq p r}}(a)
\end{align*}

\begin{align*}
\sem {\pseq {(\punion p q)} r}(a)&=\{(\upom\cdot\vpom)\cdot b \mid \upom\cdot a'\in \sem p (a)\cup\sem q (a), \vpom \cdot b \in \sem r (a')\} \\
&=\{(\upom\cdot\vpom)\cdot b \mid \upom\cdot a'\in \sem p (a), \vpom \cdot b \in \sem r(a')\}\\
&\cup \{(\upom\cdot\vpom)\cdot b \mid \upom\cdot a'\in \sem q (a), \vpom \cdot b \in \sem r (a')\} \\
&= \sem {\pseq p r}(a)\cup \sem {\pseq q r}(a)= \sem {\punion {\pseq p r} {\pseq q r}}(a)
\end{align*}

Next is $p^*\equiv \skp + pp^*$
Take an element $x\in \sem {p^*}(a)$. Thus $x\in \sem {p^{(n)}}(a)$ for some $n\in\N$. If $n=0$, then $x\in \sem {\skp}(a)=\{\onepom\cdot a\}\subseteq \sem{\punion \skp {\pseq p {p^*}}}(a)$.
If $n>0$, then $x\in \sem {\pseq p {p^{(n-1)}}}$. Hence, $x=(\upom\cdot \vpom)\cdot b$ such that
$\upom\cdot a'\in\sem p(a)$ and $\vpom\cdot b\in \sem {p^{(n-1)}}(a')$. From the latter we obtain that
$\vpom\cdot b \in \sem {p^*} (a')$. As $\upom\cdot a'\in\sem p(a)$, we then get that $x\in \sem {\pseq p {p^*}}(a)\subseteq \sem{\punion \skp {\pseq p {p^*}}}(a)$.
For the other direction, we first observe that $\sem {\skp}(a)\subseteq \sem {p^*}(a)$. Then take an $x\in \sem {\pseq p {p^*}}(a)$. Hence, $x=(\upom\cdot \vpom)\cdot b$ such that
$\upom\cdot a'\in\sem p(a)$ and $\vpom\cdot b\in \sem {p^*}(a')$. From the later we can conclude that there exists an $n\in\N$ such that $\vpom\cdot b \in \sem {p^{(n)}}(a')$. Hence, $x\in \sem
{\pseq p {p^{(n)}}}(a)=\sem {p^{n+1}}(a)$ and we obtain subsequently that $x\in \sem {p^*}(a)$.
The other star-axiom is verified in a similar way.

For the first least fixpoint axiom we assume that
$\closure[]{\sem{p+q\mathbin{;}r}}(a)\subseteq\closure[]{\sem q} (a)$ for all $a$. If $w\in \sem
{p\mathbin{;}r^*}(a)$, then $w=(\upom\cdot \vpom) \cdot b$ such that $\upom \cdot a' \in \sem {p}(a)$
and $\vpom\cdot b \in \sem {r^*}(a')$. Either $\vpom\cdot b = \onepom\cdot a'$, in which case
$w\in
\closure[]{\sem q}(a)$ follows immediately from the premise, or $\vpom\cdot b =(v_1\cdots v_n)c_n$ for
$v_i c_i \in \sem {r}(c_{i-1}')$ and $c_0'=a'$. Note that $\upom\cdot a' \in \closure[]{\sem {q }}(a)$
because of the premise. Then $(\upom\cdot v_1)\cdot c_1\in \closure[] {\sem q}(a) $ because of the premise as well. Repeating this, we get that $w\in \closure[]{\sem{q }}(a)$. The other least fixpoint axiom is verified similarly.

Next are the axioms for the parallel, which we will also prove directly.
For associativity, where we use associativity of union and of parallel composition of pomsets, we derive
\begin{align*}
  \sem {p \parallel (q\parallel r)} (a)&= \{(\upom \parallel \vpom) \cdot (b \cup c)  \mid  \upom \cdot b \in \lang{\polp}(a) , \vpom \cdot c \in \lang{q\parallel r}(a)\}\\
  &= \{(\upom \parallel \vpom)  \cdot  (b \cup c)  \mid  \upom \cdot b \in \lang{\polp}(a) , \\
  & \vpom \cdot
  c \in \{(\upom' \parallel \vpom') \cdot (b' \cup c')  \mid  \upom' \cdot b' \in \lang{\polq}(a) , \vpom' \cdot c' \in \lang{r}(a)\}\}\\
  &= \{(\upom \parallel (\upom' \parallel \vpom'))  \cdot  (b \cup (b' \cup c') )  \mid \\&  \upom \cdot b \in \lang{\polp}(a) ,  \upom' \cdot b' \in \lang{\polq}(a) , \vpom' \cdot c' \in \lang{r}(a)\}\\
    &= \{((\upom \parallel \upom' )\parallel \vpom')  \cdot  ((b \cup b') \cup c' )  \mid  \\& \upom \cdot b \in \lang{\polp}(a) ,  \upom' \cdot b' \in \lang{\polq}(a) , \vpom' \cdot c' \in \lang{r}(a)\}\\
  &= \{((\upom \parallel \upom' )\parallel \vpom')  \cdot  ((b \cup b') \cup c' )  \mid \\&  (\upom\parallel \upom')\cdot (b\cup b') \in \lang{\polp \parallel \polq}(a) ,  \vpom' \cdot c' \in \lang{r}(a)\}\\
    &= \{(\upom\parallel \vpom')  \cdot  (b \cup c' )  \mid  \upom\cdot b \in \lang{\polp \parallel \polq}(a) ,  \vpom' \cdot c' \in \lang{r}(a)\}\\
  &= \sem {(p \parallel q) \parallel r} (a)
\end{align*}

Commutativity of the parallel follows in a similar manner.

\begin{align*}
  \sem {p \parallel \pfalse}(a) & = \{(\upom \parallel \vpom) \cdot (b \cup c)  \mid  \upom \cdot b \in \lang{\polp}(a) , \vpom \cdot c \in \lang{\pfalse}(a)\}\\
  &=\{(\upom \parallel \vpom  )\cdot (b\cup c)\mid  \upom \cdot b \in \lang{\polp}(a), \vpom\cdot c\in \{\onepom\cdot \emptyset\} \} \\
  &=\{\upom \cdot b  \mid  \upom \cdot b \in \lang{\polp}(a) \} =\lang{\polp} (a)
\end{align*}

\begin{align*}
  \sem {p \parallel \abort}(a) & = \{(\upom \parallel \vpom) \cdot (b \cup c)  \mid  \upom \cdot b \in \lang{\polp}(a) , \vpom \cdot c \in \lang{\abort}(a)\}\\
  &= \{(\upom \parallel \vpom) \cdot (b \cup c)  \mid  \upom \cdot b \in \lang{\polp}(a) , \vpom \cdot c \in \emptyset\}\\
  & = \emptyset = \sem{\abort}(a)
\end{align*}

The next axiom for the parallel is distributivity over plus.
\begin{align*}
\sem{p \parallel (\punion q r)}(a) & = \{(\upom \parallel \vpom) \cdot (b \cup c)  \mid  \upom \cdot b \in \lang{\polp}(a) , \vpom \cdot c \in \lang{q + r}(a)\}\\
&= \{(\upom \parallel \vpom) \cdot (b \cup c)  \mid  \upom \cdot b \in \lang{\polp}(a) , \vpom \cdot c \in \lang{q}(a) \cup \sem{r}(a)\}\\
&= \{(\upom \parallel \vpom) \cdot (b \cup c)  \mid  \upom \cdot b \in \lang{\polp}(a) , \vpom \cdot c \in \lang{q}(a) \} \\
&\cup
\{(\upom \parallel \vpom) \cdot (b \cup c)  \mid  \upom \cdot b \in \lang{\polp}(a) , \vpom \cdot c \in  \sem{r}(a)\} \\
& = \sem {p \parallel q}(a) \cup \sem{p\parallel r}(a) \\
&= \sem {\punion {p \parallel q} {p \parallel r}}(a)
\end{align*}

For the exchange law we give a direct proof.
 Let $s,s',v$ and $v'$ be state programs.
 It is sufficient to prove $\sem {\pseq {(s\parallel s')} {(v\parallel v')}}\subseteq  \closure[]{\sem {(\pseq s
 v)\parallel (\pseq {s'} {v'})}}(a)$ (property~\ref{property:monotone} of \cref{lemma:basicfacts2}).
Take a word $w$ in $\sem {\pseq {(s\parallel s')} {(v\parallel v')}}(a)$. Via \cref{lemma:statepacketterms} we know that $w=\vpom\cdot a$, and $\vpom =(X\parallel Y)\cdot (Z\parallel W)$ for $X\cdot a \in \sem
{s}(a)$, $Y\cdot a\in \sem{s'}(a)$, $Z\cdot a\in \sem v(a)$ and $W\cdot a\in \sem{v'}(a)$. Subsequently we can conclude that
$((X\cdot Z)\parallel (Y\cdot W))\cdot a\in \sem {(\pseq s v)\parallel (\pseq {s'} {v'})}(a)$.
As $\vpom\sqsubseteq ((X\cdot Z)\parallel (Y\cdot W))$, we can conclude with \cref{def:closure} that $\vpom\in\closure[\hexch\cup\hcontr]{\{((X\cdot Z)\parallel (Y\cdot W))\}}$.
Then by definition of the closed \cnetkat semantics we can conclude that $\vpom\cdot a \in \closure[]{\sem {(\pseq s v)\parallel (\pseq {s'} {v'})}}(a)$.

Soundness of all but the last three of the packet axioms follows directly from \cref{lemma:seqpacket} and soundness of \netkat for those axioms~\cite[Theorem 1]{netkat}.

For the last three packet axioms let $x$ be a deterministic packet program. Observe that this means that $\sem x(a)=\{\onepom\cdot b\}$ for some (possibly empty) set of packets $b$ (\cref{lemma:tinydetails}, \cref{detpacket}).

\begin{align*}
\sem {x\parallel x} (a) &= \{\onepom\cdot (b\cup c)\mid \onepom\cdot b\in \sem x (a), \onepom \cdot c \in \sem x(a)\}
= \sem x (a)
\intertext{}
\sem {\pseq {(p\parallel q)} x}(a) &= \{\upom\cdot b\mid \upom\cdot a'\in\sem {p\parallel q}(a),\onepom\cdot b\in \sem x(a')\} \\
&= \{\upom\cdot b\mid \upom \cdot a'\in\{(z\parallel y)\cdot (c\cup d)\mid z\cdot c\in \sem p (a), \\&y\cdot d\in \sem q(a)\},\onepom\cdot b\in \sem x(a')\}\\
&=\{(z\parallel y)\cdot b\mid  z\cdot c\in \sem p (a), y\cdot d\in \sem q(a),\onepom\cdot b\in \sem x(c\cup d)\} \\
&=\{(z\parallel y)\cdot (e\cup f)\mid  z\cdot c\in \sem p (a), y\cdot d\in \sem q(a), \\& \sem x(c)=\{\onepom\cdot e\}, \sem x(d)=\{\onepom\cdot f\}\} \tag{\cref{lemma:splitsets}}\\
&= \{(z\parallel y)\cdot (e\cup f)\mid z \cdot e\in \sem {\pseq p x}(a), y \cdot f\in \sem{\pseq q x}(a)\}\\
&=\sem {(\pseq p x)\parallel (\pseq q x)}(a)
\intertext{}
\sem {\pseq x {(p\parallel q)}}(a) &= \{\upom\cdot b\mid \onepom\cdot a'\in\sem x(a),\upom\cdot b\in \sem {p\parallel q}(a')\} \\
&= \{\upom\cdot b\mid \onepom \cdot a'\in \sem x (a), \\&\upom\cdot b\in \{(z\parallel y)\cdot (c\cup d)\mid z\cdot c\in \sem p(a'), y \cdot d\in\sem q (a') \}\}\\
&= \{(z\parallel y)\cdot (c\cup d)\mid \sem x (a)= \{1\cdot a'\},z\cdot c\in \sem p (a'), \\
& \hspace{32mm} y\cdot d\in \sem q(a')\}\\
&=\{(z\parallel y)\cdot (c\cup d)\mid  z\cdot c\in \sem {\pseq x p} (a), y\cdot d\in \sem {\pseq x q}(a)\}\\
&=\sem {(\pseq x p)\parallel (\pseq x q)}(a)
\end{align*}

For the extensionality axiom we derive for any $a\in 2^{\Pk}$ and $b\in 2^{\Pk}_\nempty$
$
\closure[]{\sem{p}}(a)=\closure[]{\sem{\Pi_a\mathbin{;}p}}(b)=\closure[]{\sem{\Pi_a\mathbin{;}q}}(b)=\closure[]{\sem{q}}(a)
$
using \cref{lemma:packet-specified}.

Next we prove the interface axioms, which we will also prove directly.

Take a word $w\in \sem {o\wedge o'} (a)$. Hence $w=(u_1\alpha
u_2)\cdot a$ and $u_1,u_2\in \State^*$ and $\alpha\in \osem
{o}\cap \osem{o'}$. Then $(u_1\alpha)\cdot a\in \sem{o}(a)$ and
$(\alpha u_2)\cdot a\in \sem{o'}(a)$. Because of
\cref{lemma:statepacketterms} we can then conclude that
$(u_1\alpha\alpha u_2)\cdot a\in \sem {o\mathbin{;}o'}(a)$. As $u_1\alpha u_2\preceq u_1\alpha\alpha u_2 $,
we obtain from \cref{def:closure} that $u_1\alpha u_2 \in \closure[\hexch\cup\hcontr]{\{u_1\alpha \alpha u_2\}}$.
From \cref{def:closed-semantics}, we can then immediately conclude that $(u_1 \alpha u_2)\cdot a=w\in \closure[]{\sem {o\mathbin{;}o'}}(a)$.

For the next axiom we derive
\begin{align*}
  \sem{o\vee o'}(a) &= \State^* \odot \osem{o}\cup\osem{o'}\odot \State^* \times \{a\} \\
  &=\State^* \odot \osem{o}\odot \State^* \times \{a\} \cup\State^* \odot \osem{o'}\odot \State^* \times \{a\}\\
  &= \sem{o}(a)\cup \sem{o'}(a)=\sem{o+o'}(a)
\end{align*}

The axiom $\abort \equiv \bot$ is verified immediately as they both are interpreted as the empty set. The axiom $\skp\equiv\ptrue$ is also verified immediately by definition.

For the next axiom we derive
\begin{align*}
  \sem{\top\mathbin{;} o} &= \{(\upom\cdot \vpom )\cdot b \mid \upom\cdot a' \in \sem{\top}(a)=\State^*\odot \State\odot
  \State^*\times\{a\}, \vpom \cdot b\in \sem{o}(a)\\&= \State^*\odot
  \osem{o}\odot\State^*\times\{a\}\} \\&\subseteq \{\upom \cdot b \mid \upom \cdot b \in \State^*\odot
  \osem{o}\odot\State^*\times\{a\}\} = \sem{o}(a)
\end{align*}

The axioms $o\mathbin{;}\top \leqq o$, $e\mathbin{;}\top \leqq e$ and $\top \mathbin{;}e \leqq e$ are verified similarly.

The last two interface axioms are the following.
\begin{align*}
  \sem {t\wedge_{\BB}t'} (a)&=\{1\cdot \bsem t (a)\cap \bsem {t'}(a)\}\\
  &=\{1\cdot b \mid b= \bsem t (a)\cap \bsem {t'}(a)\}
  \\
  &= \{1\cdot b\mid 1\cdot a' \in \{1\cdot \bsem t(a)\},1\cdot b \in \{1\cdot \bsem {t'}(a')\}\} \\
  &=\{(\upom \cdot \vpom)\cdot b \mid \upom\cdot a'\in \sem t(a), \vpom\cdot b\in \sem{t'}(a')\}\\
  &=\sem {\pseq t {t'}}(a)
\end{align*}

\begin{align*}
  \sem {t\vee_{\BB}t'} (a)&=\{1\cdot (\bsem t (a)\cup \bsem {t'}(a))\}\\
  &=\{1\cdot (c\cup d)\mid c=\bsem t (a), d=\bsem {t'}(a)\}
  \\
  &= \{1\cdot (c\cup d)\mid \sem t (a)=\{1\cdot c\}, \sem {t'}(a)=\{1\cdot d\}\} \\
  &=\{(\upom \parallel \vpom)\cdot(c\cup d) \mid \upom\cdot c\in \sem t(a), \vpom\cdot d\in \sem{t'}(a)\}\\
  &=\sem {t\parallel t'}(a)
\end{align*}

We now check the local state vs global state axioms. Let $x$ be a deterministic packet program.
\begin{align*}
  \sem{\Pi_a\mathbin{;}\dup}(a)&= \{(\upom\cdot\vpom)\cdot b\mid
  \upom \cdot a'\in \sem{\Pi_a}(a)=\{\onepom\cdot a\},\vpom \cdot b\in
  \sem{\dup}(a')\} \\
  &= \{\vpom \cdot b \mid \vpom\cdot b\in \sem{\dup}(a)=\{a\cdot a\}\} \tag{$a\neq\emptyset$}\\
  &= \{a\cdot a \} \\
  &= \{\vpom \cdot b \mid \vpom\cdot b\in \sem{a}(a)=\{a\cdot a\}\} \\
  &= \{(\upom\cdot\vpom)\cdot b\mid
  \upom \cdot a'\in \sem{\Pi_a}(a)=\{\onepom\cdot a\},\vpom \cdot b\in
  \sem{a}(a')\} \tag{$a\neq\emptyset$} \\
  &= \sem{\Pi_a\mathbin{;}a}(a)
\end{align*}

Take $\Pi_b$ for $b\neq\emptyset$ and $w\in \mathcal{T}_{\mathsf{state}}(\OO\cup\Act\cup 2^{\Pk})$:
\begin{align*}
\sem {\pseq {\Pi_b} w} (a) & = \{(\upom\cdot \vpom)\cdot c\mid \upom \cdot a' \in \sem {\Pi_b}(a), \vpom\cdot c \in \sem w(a')\}\\
&= \{\vpom\cdot b \mid \sem {\Pi_b}(a)=\{\onepom\cdot b\}, \vpom\cdot b \in \sem w (b)\} \tag{Def $\Pi_b$, \cref{lemma:statepacketterms}} \\
&= \{\vpom \cdot b \mid \vpom \cdot a \in \sem w (a), \onepom\cdot b\in \sem {\Pi_b} (a)\} \tag{\cref{lemma:stateterms}} \\
&= \{(\vpom\cdot \upom) \cdot c \mid \vpom \cdot a' \in \sem w (a), \upom\cdot c\in \sem {\Pi_b} (a')\}\tag {\cref{lemma:statepacketterms}} \\
&= \sem {\pseq w {\Pi_b}} (a)
\end{align*}

For any program $p$ we derive
\begin{align*}
  \sem {\pseq \pfalse p}(a)
  &=\{(\upom\cdot\vpom)\cdot b \mid \upom\cdot a' \in \sem \pfalse (a), \vpom \cdot b \in \sem p (a')\} \\
    &= \{(\upom\cdot\vpom)\cdot b \mid \upom\cdot a' \in \{\onepom \cdot \emptyset\}, \vpom \cdot b \in \sem p (a')\} \\
    &=\{(\onepom\cdot\vpom)\cdot b \mid \vpom \cdot b \in \sem p (\emptyset)=\{\onepom\cdot\emptyset\}\} \tag{Def of $p$ on empty input}\\
      &= \{\onepom \cdot \emptyset \} \\
  &=\sem {\pfalse}(a)
  \end{align*}

Take $y$ a packet program.
\begin{align*}
\sem {\pseq y \pfalse}(a)&= \{(\upom\cdot\vpom)\cdot b \mid \upom\cdot a' \in \sem y(a), \vpom \cdot b \in \sem \pfalse (a')\} \\
&= \{(\upom\cdot\vpom)\cdot b \mid \onepom\cdot a' \in \sem y(a), \vpom \cdot b \in \{\onepom\cdot \emptyset \}\} \tag{\cref{lemma:statepacketterms}} \\
&= \{\onepom \cdot \emptyset \} \\
&= \sem {\pfalse}(a)
\end{align*}

Let $s$ be a state program.
\begin{align*}
  \sem{s\parallel\skp} (a)&= \{(\upom\parallel \vpom)\cdot b\cup c\mid \upom\cdot b\in \sem s (a), \vpom\cdot c\in \skp(a)=\{\onepom\cdot a\}\}\\
  &= \{(\upom\parallel \onepom)\cdot a\cup a\mid \upom\cdot a\in \sem s (a)\}\tag{\cref{lemma:statepacketterms}} \\
  &=\{\upom \cdot a \mid \upom \cdot a\in \sem s (a)\} \\
  &=\{\upom \cdot b \mid \upom \cdot b\in \sem s (a)\} \tag{\cref{lemma:statepacketterms}}\\
  &=\sem{s}(a)
\end{align*}

Let $y,z$ be packet programs and $s,v$ state programs.
\begin{align*}
\sem{(s\mathbin{;}y)\parallel(v\mathbin{;}z)}(a) &= \{(\upom\parallel \vpom) \cdot b\cup c\mid \upom \cdot b \in \sem{s\mathbin{;}y}(a),\vpom\cdot c\in \sem{v\mathbin{;}z}(a)\}\\
&= \{(\upom\parallel \vpom) \cdot b\cup c\mid \\& \upom \cdot b \in \{(\mathbf{d}\cdot \mathbf{e})\cdot f\mid
\mathbf{d}\cdot a'\in\sem s(a), \mathbf{e}\cdot f \in \sem y (a')\},\\
&\vpom\cdot c\in
\{(\mathbf{g}\cdot \mathbf{h})\cdot i\mid \mathbf{g}\cdot a'\in\sem v(a), \mathbf{h}\cdot i \in \sem z (a')\}\}\\
&=\{(\upom\parallel \vpom) \cdot b\cup c\mid \\&\upom \cdot b \in \{ \mathbf{d}\cdot f\mid
\mathbf{d}\cdot a\in\sem s(a), \onepom\cdot f \in \sem y (a)\},\\
&\vpom\cdot c\in
\{\mathbf{g}\cdot i\mid \mathbf{g}\cdot a\in\sem v(a), \onepom\cdot i \in \sem z (a)\}\} \tag{\cref{lemma:statepacketterms}}\\
&=\{(\upom\parallel \vpom) \cdot b\cup c\mid \upom \cdot a\in\sem s(a), \onepom\cdot b \in \sem y (a),\\&\vpom\cdot  a\in\sem v(a), \onepom\cdot c \in \sem z (a)\}\\
&=\{\wpom \cdot d \mid \wpom \cdot a \in\{(\upom\parallel \vpom)\cdot a\mid \upom\cdot a\in \sem s (a),\\&\vpom \cdot a \in \sem v (a)\}=\sem{s\parallel v}(a), \\
&\onepom \cdot d
\in \{\onepom\cdot b\cup c\mid \onepom\cdot b \in \sem y (a), \onepom\cdot c\in \sem z (a)\}=\sem {y\parallel z}(a)\} \tag{\cref{lemma:statepacketterms}}\\
&= \{(\upom\cdot \vpom)\cdot b\mid \upom\cdot a'\in \sem{s\parallel v}(a), \vpom\cdot b\in \sem {y\parallel z}(a')\} \tag{\cref{lemma:statepacketterms}} \\
&=\sem{(s\parallel v)\mathbin{;}(y\parallel z)}(a)
\end{align*}

The packet predicate and state observation axioms need to be proven both
for elements from $\BB$ and for elements from $\OO$. For the latter we
use that $\osem{-}$ is identical to $\sem{-}_{\mathsf{OA}}$ and soundness for
those semantics is proven in~\cite[Theorem 3.8]{pocka}. For the former, we
use that $\cup$ and $\cap$ (interpretations of $\vee_\BB$ and
$\wedge_\BB$) satisfy all those properties.
The additional state observation axioms also follow from~\cite[Theorem 3.8]{pocka}.
The only axioms left to check are the additional packet predicate axioms.
\begin{align*}
\sem{t\vee_\BB \ptrue} (a)&=\{\onepom \cdot \bsem {t\vee_\BB \ptrue}(a)\} \\
&=\{\onepom \cdot \bsem{t}(a)\cup a\}\\
&=\{\onepom \cdot a\} \tag{$\bsem{t}(a)\subseteq a$}\\ 
&=\{\onepom\cdot \bsem{\ptrue}(a) \}\\
&=\sem {\ptrue}(a)\\
&= \{\onepom \cdot \bsem{t}(a)\cup (a \setminus {\bsem{t}(a)})\}\\
&=\sem{t\vee_\BB \neg t}(a)
\end{align*}

\begin{align*}
  \sem{t\wedge_\BB \neg t} (a)&= \{\onepom \cdot \bsem{t}(a)\cap (a\setminus {\bsem{t}(a)}) \} \\
  &={\onepom\cdot \emptyset} \tag{$\bsem{t}(a)\subseteq a$} \\ 
  &=\sem{\pfalse}(a)
\end{align*}

Let $m\neq n$. Observe that this entails for a set of packets $a$ that $a(\match f n)\cap a(\match f m)=\emptyset$, as packets cannot have two different values for one field.
\begin{align*}
  \sem {\match f n \wedge_{\BB} \match f m}(a)&= \{1\cdot \bsem {\match f n \wedge_{\BB} \match f m}(a)\} \\
  &=\{1\cdot \bsem {\match f n } (a) \cap \bsem {\match f m}(a)\}\\
  &=\{1\cdot  (a(\match f n) \cap a(\match f m))\}\\
  &= \{1\cdot \emptyset\} = \sem \pfalse (a)
\end{align*}

\begin{align*}
  \sem {\bigvee_i \match f i} (a)&= \{1\cdot \bsem {\bigvee_i \match f i}(a)\} \\
  &= \{1\cdot a\}=\sem \ptrue (a)
\end{align*}
This is because $\bsem {\bigvee_i \match f i}(a)$ is the union of all the filters $\match f i$ on $a$ for all possible values $i$ of $f$, and each packet in $a$ has a value for $f$. Hence, each packet is an element of $\bsem{\match f i}(a)$ for some $i$.

In the inductive step we need to check whether the closure rules for
congruence have been preserved. We distinguish four cases. If
$p=p_0+p_1$ and $q=q_0+q_1$ with $p_0\equiv q_0$ and $p_1\equiv q_1$,
then by induction we know that
$\closure[]{\sem{p_0}}(a)=\closure[]{\sem{q_0}}(a)$ and
$\closure[]{\sem{p_1}}(a)=\closure[]{\sem{q_1}}(a)$ for all $a\in
2^{\Pk}$.
Via property~\ref{property:union} of \cref{lemma:basicfacts2}, we
obtain that \[\closure[]{\sem{p_0+p_1}}(a)=\closure[]{\sem
{p_0}}(a)\cup\closure[]{\sem {p_1}}(a)=\closure[]{\sem{q_0}}(a)\cup
\closure[]{\sem{q_1}}(a)=\closure[]{\sem{q_0+q_1}}(a)\]

For the next case we consider $p=p_0\mathbin{;}p_1$ and $q=q_0\mathbin{;}q_1$ with $p_0\equiv q_0$ and $p_1\equiv q_1$.
Using property~\ref{property:seq} from \cref{lemma:basicfacts2} and the induction hypothesis we immediately obtain
\begin{align*}
&  \closure[]{\sem{p_0\mathbin{;}p_1}}(a)\\&=\{\wpom\cdot b\mid \upom
\cdot a'\in \closure[]{\sem{p_0}}(a),\vpom \cdot b\in
\closure[]{\sem{p_1}}(a') ,
\wpom\in\closure[\hexch\cup\hcontr]{\{\upom\cdot\vpom\}}\} \\
&= \{\wpom\cdot b\mid \upom
\cdot a'\in \closure[]{\sem{q_0}}(a),\vpom \cdot b\in
\closure[]{\sem{q_1}}(a') ,
\wpom\in\closure[\hexch\cup\hcontr]{\{\upom\cdot\vpom\}}\} \\
&= \closure[]{\sem{q_0\mathbin{;}q_1}}(a)
\end{align*}

The case for $p=p_0\parallel p_1$ and $q=q_0\parallel q_1$ is analogous but instead uses~\ref{property:parallel} of \cref{lemma:basicfacts2}.

For $p=p_0^*$ and $p_0\equiv q_0$. We derive
\begin{align*}
  \closure[]{\sem{p_0^*}}(a) &= \{\upom\cdot b \mid \vpom\cdot b\in
   \displaystyle\bigsqcup_{n \in \N} \closure[]{\sem{p_0^{(n)}}}(a), \upom\in \closure[\hexch\cup\hcontr]{\{\vpom\}}\}
  \tag{\ref{property:star} of \cref{lemma:basicfacts2}} \\
  &=  \{\upom\cdot b \mid \vpom\cdot b\in  \displaystyle\bigsqcup_{n \in \N}
  \closure[]{\sem{q_0^{(n)}}}(a), \upom\in \closure[\hexch\cup\hcontr]{\{\vpom\}}\}
  \tag{IH and previous case} \\
  &= \closure[]{\sem{q_0^*}}(a) \tag{\ref{property:star} of \cref{lemma:basicfacts2}}
\end{align*}

\end{proofEnd}
%

Conversely, we will prove that if $p$ and $q$ have the same semantics on all inputs $a$, then $p\equiv q$.  
\noindent We structure the \textbf{\textit{completeness}} proof in four parts:
\begin{enumerate}
  \item Define a normal form for \cnetkat programs, and show that for every input set $a$, every
  program is provably equivalent to a program in normal form in which $a$ is incorporated. In other words, the
  normal form of a program is dependent on the input. Similar to \netkat, normal form programs are \cnetkat expressions over \emph{complete
  assignments}. We show that we have a simplified set of axioms on complete assignments and tests.
  \item Obtain completeness for $\Pi_a$-shaped programs from \netkat completeness.
  \item Using completeness of \POCKA, obtain completeness for
  programs of the form $s\mathbin{;}\Pi_a$ (and sums thereof), where $s$ is a
  state program.
  \item Lastly, we combine these results to prove that if $p$ and $q$ have the same behavior on input $a$, the program $\Pi_a\mathbin{;}p$ is provably equivalent to $\Pi_a\mathbin{;}q$.
\end{enumerate}
\smallskip

\textbf{Step 1: Normal form}
We prove that for every $a\in 2^{\Pk}$, we can write any program $p$ as $\Pi_a$ followed by a sum of state programs followed by a parallel of complete assignments. This is the most difficult step in the completeness proof.

We derive a few equivalences from~\cref{fig:axioms} regarding complete tests and assignments that make the proof of the normal form easier. We refer to these axioms as the $\emph{reduced}$ axioms. For $\alpha$ and $\beta$ complete tests such that $\alpha \neq \beta$, $\pi$ and $\pi'$ complete assignments, and $a\in 2^{\Pk}_\nempty,b\in 2^{\Pk}$, we can derive:
\begin{mathpar}
  \pi \equiv \pseq \pi {\alpha_{\pi}}
  \and
  \alpha \equiv \pseq \alpha {\pi_{\alpha}}
  \and
  \pseq \pi {\pi'} \equiv \pi'
  \and \pseq \alpha \beta \equiv \pfalse
  \and \Pi_a \mathbin {;}\Pi_b\equiv \Pi_b
\end{mathpar}

All of these equivalences are easy consequences of the packet axioms, the packet predicate axioms, the
axiom $t\land_{\BB} t' \equiv \pseq t t'$ and the fact that for all packet programs $p$
we have $p\mathbin{;}\pfalse\equiv\pfalse\equiv \pfalse \mathbin{;}p$~\cite{netkat}. The last reduced axiom is derived
in
\ifarxiv%
\cref{lemma:funny-derivation} in \cref{app:completeness}.
\else%
the full version of this article~\cite[Lemma~14]{fullversion}.
\fi%

\begin{theoremEnd}[default,category=completeness]{theorem}[Normal form]\label{lemma:normalform}
  Let $p\in\programs$ and $a\in 2^{\Pk}$. There exists a finite set $J$, and elements $u_j\in \mathcal{T}_{\mathsf{state}}(\OO\cup\Act\cup 2^{\Pk})$ and $b_j\in 2^{\Pk}$ for each $j\in J$
  s.t.
  \[\pseq {\Pi_a} {p}\equiv \pseq {\Pi_a} {\sum\limits_{j\in J}\left(\pseq {u_j} {\Pi_{b_j}}\right)}\]
\end{theoremEnd}
\begin{proofEnd}
  Let $\Pi_a = \Vertt\limits_{k\in K}\pi_k$. In case $K=\emptyset$, the equivalence follows immediately, as $\pfalse \mathbin{;}p\equiv \pfalse$ for any program $p$. In the rest of the proof we assume $K\neq\emptyset$.
We prove the claim by induction on the structure of $p$. For the case $p=t\in\BB$, we note that the complete tests are the atoms (minimal nonzero elements) of the Boolean algebra generated by the tests. Hence, we know that each element $t$ can be written as a disjunction of all the atoms below $t$ (\cite[Chapter 5.9]{birkhoff-bartee-1970}).
We derive
\begin{align*}
  \Pi_a \mathbin{;} t &\equiv \Pi_a \mathbin{;}
  \bigvee_{\alpha\leq_{\BB}t}\alpha \tag {Property Boolean algebra}\\
  &\equiv \Pi_a \mathbin{;}
  (\Vert_{\alpha\leq_{\BB}t}\alpha)\tag{axiom
  $t\vee_{\BB} t' \equiv t\parallel t'$}\\ 
  &= (\Vert_{k\in K}\pi_k )\mathbin{;}  (\Vert_{\alpha\leq_{\BB}t}\alpha) \\
  &\equiv \Vert_{k\in K}\big(\pi_k \mathbin{;} (\Vert_{\alpha\leq_{\BB}t}\alpha)\big) \tag{$\pseq {(p\parallel q)} x \equiv (\pseq p x)\parallel (\pseq q x)$}\\ 
  &\equiv \Vert_{k\in K}\Vert_{\alpha\leq_{\BB}t}\big(\pi_k \mathbin{;} \alpha\big) \tag{  $\pseq  x {(p\parallel q)}  \equiv (\pseq  x p)\parallel (\pseq  x q)$}\\ 
  &\equiv \Vert_{k\in K}\Vert_{\alpha\leq_{\BB}t}\big(\pi_k\mathbin{;}\alpha_{\pi_k}\mathbin{;}\alpha\big) \tag{$\pi\equiv \pi \mathbin{;}\alpha_{\pi}$}\\ 
  &\equiv \Vert_{k\in K}\Vert_{\alpha_{\pi_k}\leq_{\BB}t}\big(\pi_k\mathbin{;}\alpha_{\pi_k}\mathbin{;}\alpha_{\pi_k}\big) \tag{If $\alpha\neq\beta$ then
  $\alpha\mathbin{;}\beta\equiv\pfalse$, $p\parallel\pfalse\equiv p$}\\ 
    &\equiv \Vert_{k\in K}\Vert_{\alpha_{\pi_k}\leq_{\BB}t}\big(\pi_k\mathbin{;}\alpha_{\pi_k}\big) \tag{$\alpha\cdot\alpha\equiv \alpha\wedge\alpha\equiv\alpha$}\\ 
  &\equiv \Vert_{k\in K, \alpha_{\pi_k} \leq_{\BB} t}(\pi_k) \tag{ $\pi\equiv \pi \mathbin{;}\alpha_{\pi}$} \\ 
  &\equiv \Vert_{k\in K}\pi_k \mathbin{;} \Vert_{k\in K, \alpha_{\pi_k} \leq_{\BB} t}(\pi_k) \tag{$\Pi_a \mathbin{;}\Pi_b \equiv \Pi_b$} \\ 
  &\equiv\Pi_a \mathbin{;} \skp \mathbin{;} \Vert_{k\in K,\alpha_{\pi_k} \leq_{\BB} t}(\pi_k) \tag{$p\mathbin{;}\skp \equiv p$} 
\end{align*}

For $p=\skp$, we just use that $\skp \equiv \ptrue$ and the case above.

In the next base case we let $p=\modify f n$. We derive
\begin{align*}
\Pi_a \mathbin{;} \modify f n &\equiv \Pi_a \mathbin{;}\Pi_a \mathbin{;} \modify f n \tag{$\Pi_a \mathbin{;}\Pi_b \equiv \Pi_b$}\\ 
&= \Pi_a \mathbin{;} (\Vert_{k\in K}\pi_k )\mathbin{;} \modify f n \\
&\equiv \Pi_a \mathbin{;} \Vert_{k\in K}(\pi_k \mathbin{;} \modify f n ) \tag{axiom $\pseq {(p \parallel q)} x \equiv (\pseq p x)\parallel (\pseq q x)$} \\ 
&\equiv \Pi_a \mathbin{;} \skp \mathbin{;} \Vert_{k\in K}\pi_k'\tag{$p\mathbin{;}\skp \equiv p$} 
\end{align*}
where $\pi'_{k}$ is $\pi_k$ with the assignment for $f$ replaced by $(\modify f n)$.

The next three bases cases are for state actions and the rest of the state observations.
Let $p=\modify v n$. The case for $p=b\in 2^\Pk$ and $p=o\in\OO$ is identical.

 We derive
\[\Pi_a \mathbin{;} \modify v n \equiv \Pi_a \mathbin{;}\Pi_a \mathbin{;}\modify v n \equiv \Pi_a \mathbin{;} \modify v n \mathbin{;} \Pi_a \]

For the case where $p=\abort$, we use that $\bot\equiv\abort$ and the case above.

The last base case is for $\dup$. We derive
\begin{align*}
  \Pi_a \mathbin{;} \dup & \equiv \Pi_a \mathbin{;} \Pi_a\mathbin{;}\dup \tag{$\Pi_a \mathbin{;}\Pi_b \equiv \Pi_b$} \\ 
  &\equiv \Pi_a \mathbin{;} \Pi_a\mathbin{;} a \tag{  $ \pseq {\Pi_a} {\dup} \equiv \pseq {\Pi_a} {a}$} \\ 
  &\equiv \Pi_a \mathbin{;} a\mathbin{;}  \Pi_a \tag{$\Pi_a\mathbin{;}w\equiv w\mathbin{;}\Pi_a$} 
\end{align*}

We have four inductive cases.

For $p+q $ we derive
\begin{align*}
  \Pi_a \mathbin{;} (p+q) &\equiv \Pi_a \mathbin{;} p + \Pi_a \mathbin{;}q \tag{Distributivity} \\
  &\stackrel {\textrm{IH}}\equiv \Pi_a \mathbin{;} \sum_{j\in J}\left(\pseq {u_j} {\Pi_{b_j}}\right) + \Pi_a \mathbin{;}\sum_{m\in M}\left(\pseq {v_m} {\Pi_{c_m}}\right)\\
 &\equiv \Pi_a \mathbin{;}\left( \sum_{j\in J}\left(\pseq {u_j} {\Pi_{b_j}}\right) +\sum_{m\in M}\left(\pseq {v_m} {\Pi_{c_m}}\right) \right) \tag{Distributivity}
\end{align*}

For $ p \parallel q$ we derive
\begin{align*}
  \Pi_a \mathbin{;} \left(p \parallel q\right) &\equiv \left(\Pi_a \mathbin{;} p \right)\parallel \left(\Pi_a \mathbin{;} q \right)\tag{axiom $  \pseq  x {(p\parallel q)}  \equiv (\pseq  x p)\parallel (\pseq  x q) $} \\ 
   &\stackrel {\textrm{IH}}\equiv \left(\Pi_a \mathbin{;} \sum_{j\in J}\left(\pseq {u_j} {\Pi_{b_j}}\right)
   \right)\parallel \left(\Pi_a \mathbin{;}\sum_{m\in M}\left(\pseq {v_m} {\Pi_{c_m}}\right)\right) \\
  &\equiv \Pi_a \mathbin{;} \left(\left(\sum_{j\in J}\left(\pseq {u_j} {\Pi_{b_j}}\right)\right) 
  \parallel
  \left(\sum_{m\in M}\left(\pseq {v_m} {\Pi_{c_m}}\right)\right)\right)\tag{axiom $  \pseq  x {(p\parallel q)} \equiv (\pseq  x p)\parallel (\pseq  x q) $} \\ 
 &\equiv\Pi_a \mathbin{;} \left(
  \sum_{j\in J}\sum_{m\in M} \left( \left(\pseq {u_j} {\Pi_{b_j}
  }\right) \parallel \left(\pseq {v_m} {\Pi_{c_m}}\right) \right)\right)\tag{Distributivity of
  parallel over sum} \\
  &\equiv \Pi_a \mathbin{;} \left(\sum_{j\in J}\sum_{m\in M} \left( \pseq {\left(u_j\parallel v_k\right)} { \left( \Pi_{b_j}
  \parallel \Pi_{c_m}\right)} \right)\right) \tag{axiom  $(\pseq s y) \parallel (\pseq v z) \equiv \pseq {(s\parallel v)} {(y\parallel z)}$} 
\end{align*}
This last line is already in the desired format, using associativity of the parallel.

For $\pseq p q$ we derive
\begin{align*}
  \Pi_a \mathbin{;} \left(\pseq p q\right) &\equiv   \left(\Pi_a \mathbin{;} p\right) \mathbin{;} q \tag{Associativity of $\mathbin{;}$}\\
  &\stackrel {\textrm{IH}}\equiv \Pi_a \mathbin{;}
  \pseq {\sum_{j\in J}\left(\pseq {u_j} {\Pi_{b_j}}\right)}  {q} \\
&\equiv \Pi_a \mathbin{;}
\sum_{j\in J}\left(\pseq {u_j} {\Pi_{b_j}}\mathbin{;}q\right) \tag{Distributivity of $+$ and $\mathbin{;}$} \\
&\equiv \Pi_a \mathbin{;}
\sum_{j\in J} \left( \pseq {\pseq {u_j} {\Pi_{b_j}}}  {\sum_{m\in M}\left(\pseq {v_m} {\Pi_{c_m}}\right)} \right)
\tag{IH on $\Pi_{b_j}\mathbin{;}q$} \\
&\equiv \Pi_a \mathbin{;}
\sum_{j\in J} \sum_{m\in M}\left( \pseq {\pseq {u_j} {\Pi_{b_j}}} {\pseq {v_m} {\Pi_{c_m}}} \right)
\tag{Distributivity of $+$ and $\mathbin{;}$}
\end{align*}

Now we make a case distinction for each element from the sum. For each $j$ we have $\Pi_{b_j}$ equal to $\pfalse$ or not. In case it is equal to $\pfalse$, we have $b_j=\emptyset$ (see \cref{remark:piadrop}), we can derive for this $j$ and any $m$ that
\[
\pseq {\pseq {u_j} {\Pi_{b_j}}} {\pseq {v_m} {\Pi_{c_m}}} \equiv \pseq {u_j} {\pfalse}
\]
In case $\Pi_{b_j}$ is not equal to $\pfalse$, we have $b_j\neq\emptyset$ and we derive
\begin{align*}
  \pseq {\pseq {u_j} {\Pi_{b_j}}} {\pseq {v_m} {\Pi_{c_m}}}
  &\equiv
  \pseq {\pseq {u_j} {v_m}} {\pseq {\Pi_{b_j}} {\Pi_{c_m}}} \tag{$\Pi_a\mathbin{;}w\equiv w\mathbin{;}\Pi_a$}\\ 
&\equiv \pseq {\pseq {u_j} {v_m}} {\Pi_{c_m}}\tag{$\Pi_a\mathbin{;}\Pi_b\equiv \Pi_b$} 
\end{align*}

Hence, every element in the sum in
\[
\Pi_a \mathbin{;}
\sum_{j\in J} \sum_{m\in M}\left( \pseq {\pseq {u_j} {\Pi_{b_j}}} {\pseq {v_m} {\Pi_{c_m}}} \right)
\]
is provably equivalent to an expression such that the sum becomes an expression in normal form.

For $p^*$, we need a more complicated argument. We are going to use the fact that matrices over a Kleene algebra form a Kleene algebra~\cite{kozen94}. We first construct a set $Q$ of syntactic representations of sets of packets.
\begin{mathpar}
    \inferrule{~}{%
        \Pi_a\in Q
    }
    \and
    \inferrule{%
\Pi_b \in Q\\
      \Pi_b \mathbin{;} p \stackrel {\textrm{IH}}\equiv  \Pi_b \mathbin{;}\sum_{k\in K}\left(\pseq {u_k} {\Pi_{c_k}}\right)\\
    k\in K
    }{%
    \Pi_{c_k}\in Q
    }
\end{mathpar}

$Q$ is finite because from the induction hypothesis we know that $K$ is finite for each sum, and $Q$ exists out parallel products of packets, and we only have finitely many of such terms (up-to-equivalence), because parallel composition is
idempotent on complete assignments (direct representation of a packet), and we only have finitely many complete assignments, because we have finitely many fields and values.

Next we define a $Q\times Q$-matrix $M$.
If $R \neq \Vert\emptyset$ and $ R\mathbin{;}p \stackrel {\textrm{IH}}\equiv R \mathbin{;}\sum_{k\in K}\left(u_k\mathbin{;}\Pi_{c_k}\right)$,
\[ M(R, R'):= \sum \{ u_k : k \in K \text{ and } \Pi_{c_k} \equiv R' \} \]
in all other cases, we set $M(R, R') = \drp$.

As usual, multiplication of a $Q\times Q$ matrix $M$ with another $Q\times Q$-matrix $M'$ is given by
\[
(M\mathbin{;}M')(q_1,q_2)=\sum_{R\in Q}M(q_1, R)\mathbin{;}M'(R,q_2)
\]
and addition is defined pointwise.

We now prove that

\[
\Pi_a \mathbin{;} p^* \equiv \Pi_a\mathbin{;}\sum_{R\in Q}M^*(\Pi_a,R)\mathbin{;}R
\]

where the expression on the right has the correct format for the normal form, because all elements of $M*$ are state programs: they are formed out of elements of $M$, which are state programs, and the Kleene algebra operators.

We start with $\leqq$.
To apply one of the fixpoint axioms, we need to prove that
\begin{equation}\label{nummer2}
\Pi_a + \Pi_a\mathbin{;}\sum_{R\in Q}(M^*(\Pi_a,R)\mathbin{;}R)\mathbin{;}p \leqq  \Pi_a\mathbin{;}\sum_{R\in Q}M^*(\Pi_a,R)\mathbin{;}R
\end{equation}
Then the desired result follows from one of the fixpoint axioms.
We first show that
\begin{equation}\label{nummer1}\Pi_a \leqq \Pi_a\mathbin{;}\sum_{R\in Q}M^*(\Pi_a,R)\mathbin{;}R\end{equation}
Let $1$ be the identity matrix. Because all the Kleene algebra axioms hold on state programs, the elements of the matrix $M$ have a Kleene algebra structure. Then, because matrices over a Kleene algebra form a Kleene algebra, we have $1+M\mathbin{;}M^*\equiv M^*$ (axiom of KA)\cite{kozen94}. From this we can conclude that $1\leqq M^*$, and thus $1(\Pi_a,\Pi_a)\leqq M^*(\Pi_a,\Pi_a)$. We also know that $1(\Pi_a,\Pi_a)=\skp$.

We derive
\begin{align*}
  \Pi_a &\equiv \Pi_a\mathbin{;}\skp =\Pi_a\mathbin{;}1(\Pi_a,\Pi_a)\leqq \Pi_a\mathbin{;}M^*(\Pi_a,\Pi_a)\\
  &\equiv \Pi_a\mathbin{;}\Pi_a\mathbin{;}M^*(\Pi_a,\Pi_a) \tag{$\Pi_a\mathbin{;}\Pi_b\equiv\Pi_b$} \\
  &\equiv\Pi_a\mathbin{;}M^*(\Pi_a,\Pi_a)\mathbin{;}\Pi_a \tag{$\Pi_a\mathbin{;}w\equiv w\mathbin{;}\Pi_a$} 
\end{align*}
The last step we can do because $\Pi_a\not\equiv\pfalse$ and $M^*(\Pi_a,\Pi_a)$ is a state program. We now have proved~\eqref{nummer1}.

Next we show that
\begin{equation}\label{nummer3}\Pi_a\mathbin{;}\sum_{R\in Q}(M^*(\Pi_a,R)\mathbin{;}R)\mathbin{;}p \leqq  \Pi_a\mathbin{;}\sum_{R\in Q}M^*(\Pi_a,R)\mathbin{;}R
\end{equation}

We first apply distributivity and the induction hypothesis to obtain
\[
\Pi_a\mathbin{;}\sum_{R\in Q}\left(M^*\left(\Pi_a,R\right)\mathbin{;}R\right)\mathbin{;}p\equiv
\Pi_a\mathbin{;}\sum_{R\in Q}\left(M^*\left(\Pi_a,R\right)\mathbin{;}R\mathbin{;}\sum_{k\in K}\left(u_k\mathbin{;}\Pi_{b_k}\right)\right)
\]

In order to prove~\eqref{nummer3}, it thus suffices to show that for all $R$ and all $k$ we have
\[
\Pi_a\mathbin{;}M^*(\Pi_a,R)\mathbin{;}R\mathbin{;}u_k\mathbin{;}\Pi_{b_k}\leqq \Pi_a\mathbin{;}\sum_{R\in Q}M^*(\Pi_a,R)\mathbin{;}R
\]
If $R\equiv\pfalse$, the result follows immediately. If $R\not\equiv\pfalse$, we observe that $u_k\leqq M(R, \Pi_{b_k})$ (because $R\mathbin{;}p \stackrel
{\textrm{IH}}\equiv R\mathbin{;}\sum_{k\in K}(u_k\mathbin{;}\Pi_{b_k})$ and $u_k \mathbin{;}\Pi_{b_k}$ is one of the entries of that sum).
Then we derive
\begin{align*}
  &\Pi_a\mathbin{;}M^*(\Pi_a,R)\mathbin{;}R\mathbin{;}u_k\mathbin{;}\Pi_{b_k}\\
  &\leqq \Pi_a\mathbin{;}M^*(\Pi_a,R)\mathbin{;}R\mathbin{;}M(R, \Pi_{b_k})\mathbin{;}\Pi_{b_k} \\
  &\equiv \Pi_a\mathbin{;}M^*(\Pi_a,R)\mathbin{;}M(R, \Pi_{b_k})\mathbin{;}R\mathbin{;}\Pi_{b_k} \tag{$\Pi_a\mathbin{;}w\equiv w\mathbin{;}\Pi_a$}\\ 
  &\equiv \Pi_a\mathbin{;}M^*(\Pi_a,R)\mathbin{;}M(R, \Pi_{b_k})\mathbin{;}\Pi_{b_k} \tag{$\Pi_a\mathbin{;}\Pi_b\equiv \Pi_b$} 
\end{align*}

From the fact that matrices over a Kleene algebra form a Kleene algebra we know that $M^*\mathbin{;} M\leqq M^*$.
Hence \begin{equation}\label{nummer4}(M^*\mathbin{;}M)(\Pi_a,\Pi_{b_k})\leqq M^*(\Pi_a,\Pi_{b_k})\end{equation}
We also know that
\begin{equation}\label{nummer5}
(M^*\mathbin{;}M)(\Pi_a,\Pi_{b_k}) =\sum_{S\in Q}M^*(\Pi_a,S)\mathbin{;}M(S,\Pi_{b_k})
\end{equation}
Putting~\eqref{nummer4} and~\eqref{nummer5} together we obtain
\[
M^*(\Pi_a,R)\mathbin{;}M(R,\Pi_{b_k})\leqq M^*(\Pi_a,\Pi_{b_k})
\]
From this we obtain
\begin{equation}\label{nummer6}\Pi_a\mathbin{;}M^*(\Pi_a,R)\mathbin{;}M(R, \Pi_{b_k})\mathbin{;}\Pi_{b_k} \leqq \Pi_a\mathbin{;}M^*(\Pi_a, \Pi_{b_k})\mathbin{;}\Pi_{b_k}
\end{equation}

As we already had that
\begin{equation}\label{nummer7}
\Pi_a\mathbin{;}M^*(\Pi_a,R)\mathbin{;}R\mathbin{;}u_k\mathbin{;}\Pi_{b_k} \leqq \Pi_a\mathbin{;}M^*(\Pi_a,R)\mathbin{;}M(R, \Pi_{b_k})\mathbin{;}\Pi_{b_k}
\end{equation}
we can now conclude from~\eqref{nummer6} and~\eqref{nummer7} that
\[\Pi_a\mathbin{;}M^*(\Pi_a,R)\mathbin{;}R\mathbin{;}u_k\mathbin{;}\Pi_{b_k} \leqq \Pi_a\mathbin{;}M^*(\Pi_a, \Pi_{b_k})\mathbin{;}\Pi_{b_k} \leqq
\Pi_a\mathbin{;}\sum_{R\in Q}M^*(\Pi_a,R)\mathbin{;}R\]
which concludes the proof of~\eqref{nummer3}. Then~\eqref{nummer1} and~\eqref{nummer3} together lead to~\eqref{nummer2}. This concludes the $\leqq$-direction.

For the other direction ($\geqq$) we need to prove that
\[
 \Pi_a \mathbin{;} p^* \geqq \Pi_a\mathbin{;}\sum_{R\in Q}M^*(\Pi_a,R)\mathbin{;}R
\]

Take a vector $V$ with $V(R)=R$ for all $R\in Q$ and a vector $V'$ with $V'(R)=R\mathbin{;}p^*$ for all $R\in Q$.
We claim that if $M^*\mathbin{;}V\leqq V'$, we obtain the desired result immediately. The reasoning for that is as follows. If $M^*\mathbin{;}V\leqq V'$, then in particular
$(M^*\mathbin{;}V)(\Pi_a)\leqq V'(\Pi_a)$. Writing out both of these terms, we see:
\[
(M^*\mathbin{;}V)(\Pi_a) = \sum_{S\in Q}M^*(\Pi_a,S)\mathbin{;}V(S)=\sum_{S\in Q}M^*(\Pi_a,S)\mathbin{;}S
\]
and $V'(\Pi_a)=\Pi_a\mathbin{;}p^*$.
Hence, if $M^*\mathbin{;}V\leqq V'$, then $\sum_{S\in
Q}M^*(\Pi_a,S)\mathbin{;}S\leqq \Pi_a\mathbin{;}p^*$.
From that we derive
\begin{align*}
 \Pi_a \mathbin{;} p^* &\equiv \Pi_a \mathbin{;}\Pi_a \mathbin{;} p^* \tag{$\Pi_a\mathbin{;}\Pi_b\equiv \Pi_b$}\\ 
 &\geqq \Pi_a \mathbin{;} \sum_{S\in
 Q}M^*(\Pi_a,S)\mathbin{;}S
\end{align*}
This is exactly what we needed to prove.

We now prove that indeed $M^*\mathbin{;}V\leqq V'$. We use one of the fixpoint axioms. We need to show that $V+M\mathbin{;}V'\leqq V'$. We prove this pointwise for all $R\in Q$.
We first show that $V(R)\leqq V'(R)$. This follows immediately as $V(R)=R$ and $V'(R)=R\mathbin{;}p^*$.
Now we need to prove that for any $R\in Q$ we have
$(M\mathbin{;}V')(R)\leqq V'(R)$.
We derive
\[
(M\mathbin{;}V')(R) = \sum_{S\in Q}M(R,S)\mathbin{;}V'(S)=\sum_{S\in Q}M(R,S)\mathbin{;}S\mathbin{;}p^*
\]
Thus in order to prove $(M\mathbin{;}V')(R)\leqq V'(R)$, we need to show that for all $S\in Q$ we have
\[
M(R,S)\mathbin{;}S\mathbin{;}p^* \leqq R\mathbin{;}p^*
\]
If $R\equiv\pfalse$, we are done immediately as $M(R,S)=\pfalse$ in that case. If $R\not\equiv\pfalse$, we have
\[
M(R,S)\mathbin{;}S\mathbin{;}p^* = \sum u_k\mathbin{;}S\mathbin{;}p^*
\]
for $S\equiv\Pi_{b_k} $ and $R \mathbin{;}u_k\mathbin{;}S\leqq R\mathbin{;}p$.
Thus we need to prove that for all $u_k$ such that $S\equiv\Pi_{b_k} $ and $R \mathbin{;}u_k\mathbin{;}S\leqq R\mathbin{;}p$ we have
\[
u_k\mathbin{;}S\mathbin{;}p^* \leqq R\mathbin{;}p^*
\]
We derive
\begin{align*}
  u_k\mathbin{;}S\mathbin{;}p^* &\equiv   u_k\mathbin{;}R\mathbin{;}S\mathbin{;}p^* \tag{$\Pi_a\mathbin{;}\Pi_b\equiv \Pi_b$} \\ 
  &\equiv R\mathbin{;}u_k\mathbin{;}S\mathbin{;}p^* \tag{$\Pi_a\mathbin{;}w\equiv w\mathbin{;}\Pi_a$ and $R\not\equiv\pfalse$} \\ 
  &\leqq R\mathbin{;}p \mathbin{;}p^*\leqq R\mathbin{;}p^*
\end{align*}

\end{proofEnd}
\begin{proof}[Sketch]
  The proof proceeds by induction on the structure of $p$. For instance, for an assignment $\modify f n$, where we take $\Pi_a = \Vert_{k\in K}\pi_k$ for some non-empty finite index set $K$ and complete assignments $\pi_k$, we derive
  \begin{align*}
  \Pi_a \mathbin{;} \modify f n &\equiv \Pi_a \mathbin{;}\Pi_a \mathbin{;} \modify f n & \tag{$\Pi_a \mathbin {;}\Pi_b\equiv \Pi_b$}\\ 
  &= \Pi_a \mathbin{;} (\Vertt\limits_{k\in K}\pi_k )\mathbin{;} \modify f n \\
  &\equiv \Pi_a \mathbin{;} \Vertt\limits_{k\in K}(\pi_k \mathbin{;} \modify f n ) \tag{$\pseq {(p \parallel q)} x \equiv (\pseq p x)\parallel (\pseq q x)$} \\ 
  &\equiv \Pi_a \mathbin{;} \skp \mathbin{;} \Vertt\limits_{k\in K}\pi_k' & \tag{$p\mathbin{;}\skp \equiv p$} 
  \end{align*}
  where $\pi'_{k}$ is $\pi_k$ with the assignment for $f$ replaced by $\modify f n$. If $K=\emptyset$ then $\Pi_a\equiv\pfalse$ and the equivalence above follows immediately.
  The most difficult case is the star; we use an argument that relies on the fact that matrices over a Kleene algebra form a Kleene algebra~\cite{kozen94}.
  A proof can be found in
  \ifarxiv%
  \cref{app:completeness}.
  \else%
  the full version of this article~\cite[Appendix~D]{fullversion}
  \fi%
\end{proof}
\smallskip

\textbf{Step 2: Completeness for $\Pi_a$-shaped programs}
As mentioned, $\Pi_a$-shaped programs are syntactic representations of packet sets. We prove that if two such programs result in the same set of packets on any non-empty input, they are provably equivalent, using that $\Pi_a$ describes a unique set of packets.
\begin{textAtEnd}[allend, category=completeness]
  The following lemma rewrites the elements in a sum such that the packet elements are all unique.
\begin{lemma}\label{lemma:unique-in-sum}
Let $J$ be a finite index set, $u_j$ state programs and $\Pi_{b_j}$  finite parallels of complete assignments for some sets of packets $b_j$. Then any expression $\sum_{j\in J}\left(\pseq {u_j} {\Pi_{b_j}}\right)$ is provably equivalent
  to an expression $\sum_{j'\in J'}\left(\pseq {u_{j'}} {\Pi_{b_{j'}}}\right)$ for finite index set $J'$, state programs $u_{j'}$ and finite parallels of complete assignments $\Pi_{b_{j'}}$ such that for all
  $x,y\in J'$, it holds that $\Pi_{b_x}\not\equiv
\Pi_{b_y}$.
\end{lemma}
\begin{proof}
  We prove this by induction on the size of $J$. If $J=\emptyset$, we are done immediately. If $\mid J\mid = n+1$, we let $K$ be $J\setminus{\{l\}}$ for $l\in J$. We derive
  \begin{align*}
\sum_{j\in J}\left(\pseq {u_j} {\Pi_{b_j}}\right) &\equiv \sum_{j\in K}\left(\pseq {u_j} {\Pi_{b_j}}\right) + u_l \mathbin{;}\Pi_{b_l} \\
&\stackrel {\textrm{IH}}\equiv \sum_{j'\in K'}\left(\pseq {u_{j'}} {\Pi_{b_{j'}}}\right) + u_l \mathbin{;}\Pi_{b_l}
  \end{align*}
where for all $k, k'\in K'$ we have that $\Pi_{b_k}\not\equiv
\Pi_{b_{k'}}$ via the induction hypothesis.
If there does not exists $k\in K'$ such that $\Pi_{b_l}\equiv \Pi_{b_k}$, we are immediately done. If such a $k$ does exist, we know via the induction hypothesis there is only one.
We derive
\begin{align*}
\sum_{j'\in K'}\left(\pseq {u_{j'}} {\Pi_{b_{j'}}}\right) + u_l
\mathbin{;}\Pi_{b_l} &\equiv \sum_{j'\in K'\setminus{\{k\}}}\left(\pseq {u_{j'}}
{\Pi_{b_{j'}}}\right) +u_k \mathbin{;}\Pi_{b_k} +u_l \mathbin{;}\Pi_{b_l}\\
&\equiv \sum_{j'\in K'\setminus{\{k\}}}\left(\pseq {u_{j'}}
{\Pi_{b_{j'}}}\right) +(u_k +u_l)\mathbin{;}\Pi_{b_k}
\end{align*}
Hence, for all $l,l'\in K'$ it is the case that $\Pi_{b_l}\not\equiv
\Pi_{b_{l'}}$ and
\[\sum_{j\in J}\left(\pseq {u_j} {\Pi_{b_j}}\right)\equiv \sum_{j'\in K'\setminus{\{k\}}}\left(\pseq {u_{j'}}
{\Pi_{b_{j'}}}\right) +(u_k +u_l)\mathbin{;}\Pi_{b_k}\]
Because $(u_k+u_l)$ is a state program, this proves the required result.
\end{proof}

The next lemma states some basic facts about the semantics of state and packet programs, and their sequential composition.
\begin{lemma}\label{lemma:tinydetails}
  Let $s\in\mathcal{T}_{\mathsf{state}}$, $x\in\mathcal{T}_{\mathsf{packet}}$, and $a\in 2^{\Pk}$.
  Then we have
  \begin{enumerate}
    \item\label{soannoying} $\closure[]{\sem{s\mathbin{;}x}}(a)=\{\upom\cdot b \mid \upom\in\closure[\hcontr\cup\hexch] { \{\vpom\}},\vpom\cdot
  a\in\sem{s}(a),\onepom\cdot b\in \sem{x}(a)  \}$.
  \item\label{detpacket} If $x\in\mathcal{T}_{\mathsf{det-packet}}$, we have  $\closure[]{\sem{x}}(a)=\sem{x}(a)=\{1\cdot b\}$ for some $b\in 2^{\Pk}$.
\end{enumerate}
\end{lemma}
\begin{proof}
  For the first item (\ref{soannoying}) we derive
  \begin{align*}&\closure[]{\sem{s\mathbin{;}x}}(a)\\&=\{\upom\cdot b \mid \vpom \cdot b \in \sem{s\mathbin{;}x}(a), \upom\in\closure[\hcontr\cup\hexch] { \{\vpom\}} \}\\
    &=\{\upom\cdot b \mid \vpom \cdot b \in \{(\upom'\cdot\vpom')\cdot b'\mid \upom'\cdot a' \in \sem{s}(a), \vpom'\cdot b'\in \sem{x}(a')\}, \\&\upom\in\closure[\hcontr\cup\hexch] { \{\vpom\}} \}
\end{align*}
Because $s$ is a state program, we know that $a'=a$, and because $x$ is a packet program, we know that $\vpom'=\onepom$.
Hence
\begin{align*}
  &\closure[]{\sem{s\mathbin{;}x}}(a)\\&=
\{\upom\cdot b \mid \vpom \cdot b \in \{(\upom'\cdot\vpom')\cdot b'\mid \upom'\cdot a' \in \sem{s}(a), \vpom'\cdot b'\in \sem{x}(a')\},\\& \upom\in\closure[\hcontr\cup\hexch] { \{\vpom\}} \}\\
&=
\{\upom\cdot b \mid \vpom \cdot a\in \sem{s}(a), \onepom\cdot b\in \sem{x}(a), \upom\in\closure[\hcontr\cup\hexch] { \{\vpom\}} \}
\end{align*}

The second item (\ref{detpacket}) can be proven with induction on the structure of $x$. The result follows easily because packet programs do not influence the state pomset (\cref{lemma:statepacketterms}), and $\parallel$ and $\mathbin{;}$ do not introduce multiple elements in the output.
\end{proof}
\end{textAtEnd}

\begin{textAtEnd}[allend,category=completeness]
\begin{lemma}\label{lemma:complete-assignments}
Let $a\in 2^{\Pk}_\nempty$ and let $\pi,\pi'$ be complete assignments. If $\closure[]{\sem{\pi}}(a)=\closure[]{\sem{\pi'}}(a)$ then $\pi\equiv \pi'$.
\end{lemma}
\begin{proof}
As $\pi$ and $\pi'$ are complete assignments, they assign a value to each field. Because $\closure[]{\sem{\pi}}(a)=\{\onepom\cdot\{x\}\}=\closure[]{\sem{\pi'}}(a)$ for some packet $x$, both $\pi$ and $\pi'$ must assign to each field the same value. Using commutativity of independent assignments, we obtain that $\pi\equiv\pi'$.
\end{proof}
\end{textAtEnd}

\begin{theoremEnd}[default,category=completeness]{lemma}\label{lemma:partial-completeness-sets}
  Let $a\in 2^{\Pk}_\nempty$, and $b,c\in2^{\Pk}$.  If
  $\closure[]{\sem{\Pi_b}}(a)=\closure[]{\sem{\Pi_c}}(a)$ then
  $\Pi_b\equiv \Pi_c$.
\end{theoremEnd}
\begin{proofEnd}
  Take finite parallels of complete
  assignments $\Vert_{i\in I}\pi_i$ and $\Vert_{j\in J}\pi'_j$ such that $\Vert_{i\in I}\pi_i\equiv \Pi_b$ and $\Vert_{j\in J}\pi'_j\equiv\Pi_c$.
  Let $\semp{-}:C\to \Pk$, where $C$ is the set of complete assignments, be
  defined as $\semp{\pi}:=x$ for $\sem{\pi}(b)=\{\onepom\cdot\{x\}\}$ for some non-empty set of packets $b$ (this is independent of the choice of $b$). The function $\semp{-}$ is surjective, and
  thus has at least one right inverse. We denote it with $r : \Pk\rightarrow
  C$. Hence, we have $\semp{r(x)}=x$.
We now show that
 \[\closure[]{\sem{
 \Vertt\limits_{i\in I}\pi_i}}(a) =  \{\onepom \cdot \{\semp{\pi_i}\mid i\in I\}\}
 \]
  Note that for a complete assignment and non-empty set of packets $\sem{\pi_i}(a)=\{\onepom\cdot \{x\}\}$ for some $x\in\Pk $.
  We derive
\begin{align*}   \closure[]{\sem
{\Vertt\limits_{i\in I}\pi_i}}(a) &=\sem{
    \Vertt\limits_{i\in I}\pi_i}(a)\\ & = \{\onepom \cdot \cup_{i\in I} \{c_i \mid \sem{\pi_i}(a)=\{\onepom\cdot c_i\}\}\} \\
  &= \{\onepom \cdot \{x \mid i\in I, \sem{\pi_i}(a)=\{\onepom\cdot \{ x \}\}\}\} \\
&= \{\onepom\cdot \{\semp{\pi_i}\mid i\in I\}\}\end{align*}

  Then we obtain using soundness that
  \begin{align*}
  \closure[]{\sem{\Pi_b}}(a)
    &= \closure[]{\sem{\Vertt\limits_{i\in I}\pi_i}}(a) \\
    &= \{\onepom \cdot \{\semp{\pi_i}\mid i\in I\}\} \\
    &= \{\onepom \cdot \{\semp{\pi_j'}\mid j\in J\}\} \\
    &=\closure[]{\sem{\Pi_c}}(a)
  \end{align*}
    As $\semp{\pi_i}$ and $\semp{\pi_j'}$ are always just a packet for all $i\in I$ and $j\in J$,  we can conclude from associativity, idempotence and commutativity of the parallel on deterministic packet programs that
    \[\Vertt\limits_{i\in I} r(\semp{\pi_i}) \equiv \Vertt\limits_{j\in J} r(\semp{\pi_j'})
      \]

      As $r$ is a right inverse we can derive
      \[\semp{r(\semp{\pi_i}
      )}=\semp{\pi_i}
        \]
        Using \cref{lemma:complete-assignments}, we obtain that $r(\semp{\pi_i}) \equiv \pi_i$.
        Hence, we can conclude that
        \[\Pi_b\equiv
        \Vertt\limits_{i\in I} \pi_i\equiv \Vertt\limits_{i\in I} r(\semp{\pi_i})\equiv \Vertt\limits_{j\in J} r(\semp{\pi_j'}) \equiv \Vertt\limits_{j\in J} \pi_j'\equiv \Pi_c \qedhere
            \]
\end{proofEnd}
%
\smallskip

\textbf{Step 3: Completeness of sums in the normal form}
%
\begin{textAtEnd}[allend, category=completeness]
We now slightly generalize the completeness proof from~\cite{pocka}.
\begin{lemma}\label{pockaupdate}
Take the \POCKA terms over alphabet $\OO\cup\Act\cup 2^{\Pk}$, let $s,v$ be such terms, and extend
\cref{def:pockasem} with $\llparenthesis a\rrparenthesis = \{a\}$ for $a\in
2^{\Pk}$. If $\sem{s}_{\textrm{\POCKA}}=\sem{v}_{\textrm{\POCKA}}$, then
$s\equiv_{\textrm{\POCKA}} v$.
\end{lemma}
\begin{proof}
  We need to extend the result of~\cite[Theorem 4.6]{pocka}. This result is built on Lemma 4.1, Lemma 4.2, Lemma 4.3 and Lemma 4.4 from the same paper, which we need to extend first. The homomorphism $\hat{r}$ gets an extra clause and is defined instead as the homomorphic extension of
  \[
    r(a) =
    \begin{cases}
      \sum_{\alpha \leq_{\OO} a} \alpha & a \in \OO\\
      a & a \in \Act\cup 2^{\Pk}
    \end{cases}
\]
Lemma 4.1 from~\cite{pocka}, which states that for all \POCKA programs $e$, we have $e\equiv_{\textrm{\POCKA}} \hat{r}(e)$, can be proven in the exact same way as in~\cite{pocka}, except that there is an extra base case for $a\in 2^{\Pk}$, which follows immediately.
Lemma 4.2 from the same paper shows the unclosed \POCKA semantics of a program $e$ is equivalent to the \BKA semantics of $\hat{r}(e)$ after closure under a set of hypotheses called $\mathsf{top}$. In the current paper we do not provide this definition, but for the reader familiar with this framework: letters $a\in 2^{\Pk}$ do not occur in $\mathsf{top}$, and hence the \BKA-semantics of $\hat{r}(a)$ for $a\in 2^{\Pk}$ closed under $\mathsf{top}$ is just $\{a\}$.
This is a new base case in the proof, and the rest of the proof follows in the same way.
Lemma 4.3 is independent of the \POCKA terms, so can be copied immediately.
For Lemma 4.4, we need to extend $s$ with the clause $s(a)=a$ for $a\in 2^{\Pk}$.
The result that the homomorphism $\hat{s}$ generated by $s$ is such that for all \POCKA terms we have $e\equiv \hat{s}(e)$ still holds as the only added base case is trivial. The other claim, that the \BKA semantics of $\hat{s}(e)$ are the same as the \BKA semantics of $e$ closed under $\mathsf{top}$, also follows immediately as the \BKA semantics of $\hat{s}(a)=\{a\}$ and the \BKA semantics of $a$ closed under $\mathsf{top}$ is $\{a\}$ as well. Then we can repeat the exact proof as is given in~\cite{pocka} for Lemma 4.6 (completeness).
\end{proof}
\end{textAtEnd}
We first prove completeness for state programs, where we use completeness of \POCKA. To do so, some caution is needed; \POCKA terms are state terms over the alphabet $\OO\cup\Act$. However, the state terms relevant here also include elements $a\in 2^{\Pk}$. 
\begin{theoremEnd}[default, category=completeness]{lemma}\label{alternativecompl}
Let $s,v\in \mathcal{T}_{\mathsf{state}}(\OO\cup\Act\cup 2^{\Pk})$ and $a\in 2^{\Pk}_\nempty$. If $\closure[]{\sem{s}}(a)=\closure[]{\sem{v}}(a)$, then $s\equiv v$.
\end{theoremEnd}
\begin{proofEnd}
  The proof of \cref{theorem:pockarel} can trivially be extended with an extra base case for $b\in 2^{\Pk}$, as
  $\sem{b}(a)=\{b\cdot a\}$. Hence, from $\closure[]{\sem{s}}(a)=\closure[]{\sem{v}}(a)$ we obtain via
  \cref{theorem:pockarel} that $\sem{s}_{\textrm{\POCKA}}=\sem{v}_{\textrm{\POCKA}}$. Now we use \cref{pockaupdate} to conclude that $s\equiv_{\textrm{\POCKA}}v$, and as all the axioms of \POCKA on state terms are axioms of \cnetkat, we obtain that $s\equiv v$.
\end{proofEnd}

Next we prove completeness for expressions of the form $s\mathbin{;}\Pi_a$, and then extend this to arbitrary finite sums of such programs:
\begin{theoremEnd}[default,category=completeness]{lemma}\label{lemma:equivalence-simple}
  Let $b,c\in 2^{\Pk}$, $u,v$ state programs, and $a\in 2^{\Pk}_\nempty$. Then we have:
$
      \closure[]{\sem{u\mathbin{;}\Pi_b}}(a)=\closure[]
      {\sem{v\mathbin{;}\Pi_c}}(a) \Rightarrow u \mathbin{;}\Pi_b \equiv v\mathbin{;}\Pi_c
$.
\end{theoremEnd}
\begin{proofEnd}
  Via \cref{lemma:tinydetails} item~\ref{soannoying} and~\ref{detpacket}, we
  obtain that $\sem{\Pi_b}=\{\onepom\cdot d\}=\sem{\Pi_c}(a)$ for some $d\in 2^{\Pk}$. We can use \cref{lemma:partial-completeness-sets} to obtain that
  \begin{equation}\label{eq5}
  \Pi_b\equiv \Pi_c
  \end{equation}
  Next we show that
  \begin{equation}\label{eq3}
  \closure[]{\sem{u}}(a)=\closure[]{\sem{v}}(a)
  \end{equation}
  For $w\in \closure[]{\sem{u}}(a)$
   we know via
  \cref{lemma:statepacketterms} and the definition of closure that $w=\upom \cdot a$,
  $\upom\in\closure[\hcontr\cup\hexch] { \{\upom'\}}$ and $\upom'\cdot a \in \sem{u}(a)$.
  Via \cref{lemma:tinydetails} we obtain that $\upom\cdot d\in\closure[]{\sem{u\mathbin{;}\Pi_b}}(a)$ and via our assumption that
  $\upom\cdot d\in\closure[]{\sem{v\mathbin{;}\Pi_c}}(a)$. Via \cref{lemma:tinydetails} we can conclude that $\upom\cdot a\in \closure[]{\sem{v}}(a)$.
  %
  The other direction is argued symmetrically. Via
  \cref{alternativecompl}, we obtain from~\eqref{eq3} that
  \begin{equation}\label{eq4}
    u\equiv v
  \end{equation}
Combining~\eqref{eq5} and~\eqref{eq4} we obtain
  \[
    u\mathbin{;} \Pi_b\equiv v \mathbin{;}\Pi_c
\qedhere\]
\end{proofEnd}
%
\begin{theoremEnd}[default,category=completeness]{lemma}\label{lemma:sumsemantics}
  If
  \(\closure[]{\sem{
  \sum_{j\in J}(\pseq {u_j} {\Pi_{b_j}})}}(a)=\closure[]{\sem{\sum_{k\in K}(\pseq {v_k} {\Pi_{c_k}})}}(a)\)
  for some $a\in 2^{\Pk}_\nempty$,
  then
  \(
  \sum_{j\in J}(\pseq {u_j} {\Pi_{b_j}})\equiv \sum_{k\in K}(\pseq {v_k} {\Pi_{c_k}})\,,
  \)
  where $J,K$ are finite;
  $u_j, v_k$ are state programs and $b_j, c_k \in 2^{\Pk}$  for each $j,k$.
\end{theoremEnd}
\begin{proofEnd}
  We can use \cref{lemma:unique-in-sum} to obtain a sum that is provably equivalent to \linebreak[4] $\sum_{j\in J}\left(\pseq {u_j} {\Pi_{b_j}}\right)$ where it holds that for all $j,j'$ in the sum, we have $\Pi_{b_j}\not\equiv \Pi_{b_{j'}}$.
  Via soundness, the semantics of this expression is equivalent to the semantics of $\sum_{j\in J}\left(\pseq {u_j} {\Pi_{b_j}}\right)$. For simplicity we will not rename variables; we simply assume that in $\sum_{j\in J}\left(\pseq {u_j} {\Pi_{b_j}}\right)$ it is the case that for all $j,j'\in J$, we have $\Pi_{b_j}\not\equiv \Pi_{b_{j'}}$.

  Take any $j\in J$.
  Take $w\in \closure[]{\sem{
  \sum_{j\in J}\left(\pseq {u_j} {\Pi_{b_j}}\right)}}(a)$. Because of property~\ref{property:union} from \cref{lemma:basicfacts2}, we know that $\closure[]{\sem{
  \sum_{j\in J}\left(\pseq {u_j} {\Pi_{b_j}}\right)}}(a)=\bigcup_{j\in J} \closure[]{\sem{\pseq {u_j} {\Pi_{b_j}}}}(a)$.
  Hence, there exists a $j\in J$ such that $w\in \closure[]{\sem{u_j\mathbin{;}\Pi_{b_j}}}(a)$.
  Via \cref{lemma:tinydetails}, we get that $w=\upom \cdot d$ for some state pomset $\upom$ and
  $\sem{\Pi_{b_j}}(a)=\{\onepom\cdot d\}$. Because of
our assumption, we know there exists a $k\in K$ such that $w\in
  \closure[]{\sem{v_k\mathbin{;}\Pi_{c_k}}}(a)$. Via \cref{lemma:tinydetails} we get that $\sem{\Pi_{c_k}}(a)=\{\onepom\cdot d\}$.
  Now we prove that in fact
  \begin{equation}\label{eq2}
      \closure[]{\sem{u_j\mathbin{;}\Pi_{b_j}}}(a)=\closure[]{\sem{v_k\mathbin{;}\Pi_{c_k}}}(a)
  \end{equation}
The left-to-right inclusion is already proven. For the right-to-left
inclusion we take \\
$w'\in \closure[]{\sem{v_k\mathbin{;}\Pi_{c_k}}}(a)$. Via \cref{lemma:tinydetails} and the fact that
$\sem{\Pi_{c_k}}(a)=\{\onepom\cdot d\}$, we obtain that
$w'=\vpom\cdot d$ for some state pomset $\vpom$. Because of our assumption, we know there exists an $n\in J$ such that $w'\in \closure[]{\sem{u_n\mathbin{;}\Pi_{b_n}}}(a)$. Suppose that
$n\neq j$. Via \cref{lemma:tinydetails} and the fact that $w'=\vpom\cdot d$ we obtain that $\sem{\Pi_{b_n}}(a)=\{\onepom\cdot d\}$.
Hence $\closure[]{\sem{\Pi_{b_n}}}(a)=\closure[]{\sem{\Pi_{b_j}}}(a)$ and we obtain
$\Pi_{b_n}\equiv \Pi_{b_j}$ via
\cref{lemma:partial-completeness-sets}. This is a contradiction with \cref{lemma:unique-in-sum}. Hence, $j=n$. This proves~\eqref{eq2}.

 Now we can use \cref{lemma:equivalence-simple} to obtain
 that for all $j\in J$ there exists $k\in K$ such that
\begin{equation} \label{eq6}
  u_j\mathbin{;} \Pi_{b_j}\equiv v_k \mathbin{;}\Pi_{c_k}
\end{equation}
Similarly, we can prove~\eqref{eq6} for all $k\in K$.
Via idempotence, cummatativity and associativity of $+$ we can then conclude that
\[
\sum_{j\in J}\left(\pseq {u_j} {\Pi_{b_j}}\right)\equiv \sum_{k\in K}\left(\pseq {v_k} {\Pi_{c_k}}\right)
\qedhere\]
\end{proofEnd}
\smallskip

\textbf{Step 4: Completeness}
The last lemma before proving completeness relates the semantics of $p$ on input $a$ to the semantics of $\Pi_a\mathbin{;}p$ on any non-empty input.
\begin{theoremEnd}[default,category=completeness]{lemma}\label{lemma:packet-specified}
  Let $b\in 2^{\Pk}_\nempty$, $a\in 2^{\Pk}$.
  For all $p\in\programs$, $\closure[] {\sem{\Pi_a\mathbin{;}p}}(b)=\closure[]{\sem{p}}(a)$.
\end{theoremEnd}
\begin{proofEnd}
  We prove instead that $\sem{\Pi_a\mathbin{;}p}(b)=\sem{p}(a)$,
  which directly implies the statement we need to show. Using the
  fact that $\sem{\Pi_a}(x)=\{\onepom\cdot a\}$ for any non-empty set
  of packets $x$, we have directly that
  $\sem{\Pi_a\mathbin{;}p}(b)=\{(\upom\cdot\vpom) \cdot c  \mid
  \upom \cdot a' \in \lang{\Pi_a}(b) , \vpom \cdot c \in
  \lang{p}(a')\}=\{(\onepom\cdot\vpom) \cdot c  \mid  \vpom \cdot c \in \lang{p}(a)\}=\sem{p}(a)$.
\end{proofEnd}

\begin{theorem}[Completeness]\label{thm:complete2}
      Let $p,q\in\programs$. For all $a\in 2^{\Pk}$ we have that if  $\closure[]{\sem{p}}(a)=\closure[]{\sem{q}}(a)$, then $p\equiv q$.
\end{theorem}
\begin{proof}
  We first show that $\Pi_a\mathbin{;}p\equiv \Pi_a\mathbin{;}q$ for all $a\in 2^{\Pk}$.
  In case $a=\emptyset$, $\Pi_a$ must be the empty parallel. Hence, $\Pi_a\mathbin{;}p\equiv\pfalse\equiv \Pi_a\mathbin{;}q$. In the rest of the proof we assume $a\neq\emptyset$.
  Via \cref{lemma:packet-specified}, we obtain that $\closure[]{\sem{p}}(a)=\closure[]{\sem{\Pi_a \mathbin{;}p}}(a)=\closure[]{\sem{\Pi_a \mathbin{;}q}}(a)=\closure[]{\sem{q}}(a)$.
We obtain a normal form such that
  $\pseq {\Pi_a} {p}\equiv \pseq {\Pi_a} {\sum_{j\in J}(\pseq {u_j} {\Pi_{b_j}})}$ (\cref{lemma:normalform}).
  Similarly, $\pseq {\Pi_a} {q}\equiv \pseq {\Pi_a} {\sum_{k\in K}(\pseq {v_k} {\Pi_{c_k}})}$.
  Via soundness we derive $\closure[]{\sem{\pseq {\Pi_a}
  {\sum_{j\in J}(\pseq {u_j} {\Pi_{b_{j}}})}}}(a)=\closure[]{\sem{\pseq {\Pi_a} {\sum_{k\in K}(\pseq {v_k} {\Pi_{c_k}})}}}(a)$, and via
  \cref{lemma:packet-specified} that
$\closure[]{\sem{
  \sum_{j\in J}(\pseq {u_j} {\Pi_{b_j}})}}(a)=\closure[]{\sem{\sum_{k\in K}(\pseq {v_k} {\Pi_{c_k}})}}(a)
$.
With the partial completeness result from \cref{lemma:sumsemantics}, we obtain that $\sum_{j\in J}\left(\pseq {u_j} {\Pi_{b_j}}\right)\equiv \sum_{k\in K}\left(\pseq {v_k} {\Pi_{c_k}}\right)$. This leads to  \[\Pi_a\mathbin{;}p\equiv \Pi_a\mathbin{;}\sum_{j\in J}\left(\pseq {u_j} {\Pi_{b_j}}\right)\equiv\Pi_a\mathbin{;}\sum_{k\in K}\left(\pseq {v_k} {\Pi_{c_k}}\right)\equiv \Pi_a\mathbin{;} q\]

Hence, we have derived that $\Pi_a\mathbin{;}p\equiv \Pi_a\mathbin{;}q$ for all $a\in 2^{\Pk}$. With the extensionality axiom we can conclude that $p\equiv q$.
\end{proof}
%

\section{Examples}\label{sec:example}
This section shows how we can use \cnetkat to model and analyze several concurrent programs. We start by analyzing the running example from \cref{sec:overview}, and then proceed to a more involved example that combines the behavior of a stateful firewall, a load balancer, and an in-network cache.


\subsection{Running Example}

Consider again the running example from \cref{sec:overview}. Because we are
ultimately interested in the behavior of the program when the packets have
reached their final destination, switch $4$, we will add a test
$\match \sw 4$ at the end of the program:
 \[
 p \defeq (\modify v 0 )\mathbin{;}{(p_1\parallel p_2\parallel p_3\parallel p_4)}^* \mathbin{;} (\match \sw 4)
 \]

Recall that the \cnetkat semantics of a program contains traces that are only required to model executions where the program is composed in parallel with another program, to ensure a compositional semantics for the language. However, to analyze the behavior of a program in isolation, we want to eliminate these extra traces. To do this, we follow the same strategy used in~\cite{pocka}, where so-called \emph{guarded} pomsets were proposed. Guarded pomsets are a subclass of pomsets that captures the characteristics of behaviors of (concurrent) programs running in isolation. For example, in a guarded pomset, if one assertion, say $\match v 0$, occurs before another assertion, say $\match v 1$, there must be an assignment $\modify v 1$ between the two asserts to account for the change. That is, in an isolated execution every change to variables must be explained by an action in the program.

To illustrate the difference between pomsets and guarded pomsets,
consider our example. We unfold the Kleene star twice and evaluate
the resulting program; we obtain a pair with output
$\{\spadesuit[4/\sw],\heartsuit[4/\sw]\}$ and corresponding pomset,

\begin{tikzpicture}[node distance=0.35cm]\small
    \begin{scope}[every node/.style={anchor=center}]
    \node (gamma1) {$(\modify v 0)$};
    \node[right=of gamma1] (bla2) {$\{\heartsuit,\spadesuit\}$};
    \node[yshift=5mm,right=of bla2] (gamma2) {$\beta$};
    \node[below=5mm of gamma2] (gamma5) {$\{\heartsuit\}$};
        \node[right=of gamma2] (bla3) {$\{\spadesuit\}$};
        \node[right=5mm of gamma5] (bla4) {$\{\heartsuit[3/\sw]\}$};
        \node[right=15mm of bla3] (bla5) {$\spadesuit[2/\sw]$};
        \node[below=5mm of bla5] (bla6) {$(\modify v 1)$};
        \node[right=of bla5] (bla7) {$\{\spadesuit[2/\sw]\}$};
        \node[below=5mm of bla7] (bla8) {$\{\heartsuit[3/\sw]\}$};
        \node[right=of bla7] (bla9) {$\{\spadesuit[4/\sw]\}$};
        \node[right=of bla8] (bla10) {$\{\heartsuit[4/\sw]\}$};
    \end{scope}
\path  (gamma1) edge[->] (bla2);
\path (bla2) edge[->] (gamma2);
\path (bla2) edge[->] (gamma5);
\path (gamma2) edge[->] (bla3);
\path (gamma5) edge[->] (bla4);
\path  (bla3) edge[->] (bla5);
\path (bla4) edge[->] (bla6);
            \path (bla5) edge[->] (bla7);
            \path  (bla6) edge[->] (bla8);
              \path  (bla7) edge[->] (bla9);
                  \path  (bla8) edge[->] (bla10);
                          \path (bla6) edge[->] (bla7);
                              \path (bla5) edge[->] (bla8);
\end{tikzpicture}

\noindent where $\beta(v)=1$. This pomset is unguarded: $\beta(v) =1$ occurs without a cause.

The semantics also contains a pair with $\{\spadesuit[4/\sw],\heartsuit[4/\sw]\}$ and pomset,

\noindent\begin{tikzpicture}[node distance=0.35cm]\small
    \begin{scope}[every node/.style={anchor=center}]
        \node (gamma1) {$\alpha$};
        \node (vzero) [left=2mm of gamma1] {$(\modify v 0)$};
        \node (gamma) [left=2mm of vzero] {$\gamma$};
        \node[right=2mm of gamma1] (bla2) {$\{\heartsuit,\spadesuit\}$};
        \node[yshift=5mm,right=of bla2] (gamma2) {$\beta$};
        \node[below=5mm of gamma2] (gamma5) {$\{\heartsuit\}$};
        \node[right=of gamma2] (bla3) {$\{\spadesuit\}$};
        \node[right=2mm of gamma5] (bla4) {$\{\heartsuit[3/\sw]\}$};
        \node[right=9mm of bla3] (bla5) {$\spadesuit[2/\sw]$};
        \node[right=3mm of bla4] (bla6) {$(\modify v 1)$};
        \node[right=of bla5] (bla7) {$\{\spadesuit[2/\sw]\}$};
        \node[right=2mm of bla6] (bla8) {$\{\heartsuit[3/\sw]\}$};
        \node[right=of bla7] (bla9) {$\{\spadesuit[4/\sw]\}$};
        \node[right=of bla8] (bla10) {$\{\heartsuit[4/\sw]\}$};
    \end{scope}
    \path (gamma) edge[->] (vzero);
    \path (vzero) edge[->] (gamma1);
    \path (gamma1) edge[->] (bla2);
    \path (bla2) edge[->] (gamma2);
    \path (bla2) edge[->] (gamma5);
    \path (gamma2) edge[->] (bla3);
    \path (gamma5) edge[->] (bla4);
    \path (bla3) edge[->] (bla5);
    \path (bla4) edge[->] (bla6);
    \path (bla5) edge[->] (bla7);
    \path (bla6) edge[->] (bla8);
    \path (bla6) edge[->] (gamma2.310);
    \path (bla7) edge[->] (bla9);
    \path (bla8) edge[->] (bla10);
    \path (bla6) edge[->] (bla7);
    \path (bla5) edge[->] (bla8);
\end{tikzpicture}

\noindent with $\alpha(v)=0$, $\beta(v)=1$, and $\gamma$ unrestricted. This pomset is guarded because it contains an arrow from $\modify v 1$ to $\beta$, justifying the change in valuation from $\alpha$ to $\beta$. As we show in
\ifarxiv%
\cref{sec:analysis},
\else%
the full version of this article~\cite[Appendix~E]{fullversion},
\fi%
all guarded pomsets in the semantics will have this arrow, and satisfy the desired property: $\heartsuit$ packets are observed at switch $3$ before $\spadesuit$ packets are observed at switch $2$.

Now consider the axiomatic claim we made in \cref{sec:overview} (i.e.,~\eqref{eq:bla}),
$
(\heartsuit\parallel\spadesuit) \mathbin{;} q \leqq (\heartsuit\parallel\spadesuit) \mathbin{;} p
$
where $q$ is the program from \cref{programq}. We can easily see that the following holds:
$
\closure[]{\sem{q}}\{\heartsuit,\spadesuit\}\subseteq \closure[]{\sem{p}}\{\heartsuit,\spadesuit\}
$.
Hence, we can use \cref{lemma:packet-specified} and the completeness result for \cnetkat (\cref{thm:complete2}) to obtain~\eqref{eq:bla}.

\subsection{Stateful Load Balancer, Cache, and Firewall}

For a more complex example, consider the network in \cref{fig:lb}, which is adapted from an example from~\cite{mooly}. The overall goal is to (i) prevent packets from a high-priority server $S_h$ going to low priority hosts $l_1,\ldots, l_k$ and (ii) load balance requests to the servers in a round robin fashion. We provide naive specifications for the cache, firewall and load balancer programs in \Cref{fig:lb}. For simplicity, we assume that there is exactly one low-priority host, and exactly one high-priority host, i.e., $n=k=1$, and we leave the specification of the topology implicit.

\begin{remark}
In contrast with the previous example, the program in \Cref{fig:lb} includes reads and writes of a global variable that occur on different physical devices.
In principle, synchronizing variables like $r$ would give rise to additional packets that update local copies of variables---a process that could itself be modelled in \cnetkat.
We leave the implementation of a translation pass that achieves the synchronization of global variables across switches to future work.
\end{remark}

In~\cite{mooly}, the authors point out a problem with the example that arises because the cache has no means to enforce the security policy. One strategy for resolving this problem is to swap the placement of the firewall and the cache. Another is to distribute access control rules onto the cache as well as the firewall. However, there is also a second, more subtle issue: the load balancer uses the global variable $r$ to decide to which server to forward requests. In the presence of multiple packets, another packet may arrive before the change to the global variable occurs allowing two (or more!) packets to be sent to the same server.

\begin{figure}[t]
\[
\begin{array}{ll}
\includegraphics[scale=.3]{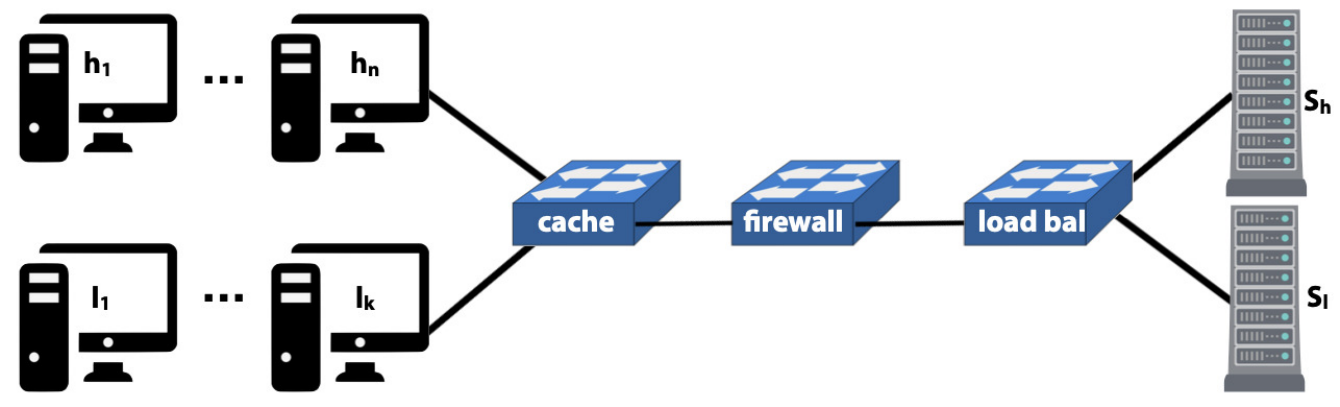}
\end{array}
\]
\[
\begin{array}{ll}
C \defeq & ((\match v 1 )\mathbin{;} (\modify {\mathsf{dst}} h_1)\mathbin{;}\dup + (\match v 0 )\mathbin{;} (\modify {\mathsf{dst}} l_1)\mathbin{;}\dup) \\
&
\parallel (\mathsf{src} = l_1 \mathbin{;} (\modify {\mathsf{dst}} {\mathsf{firewall}})) \parallel
(\mathsf{src} = h_1 \mathbin{;} (\modify {\mathsf{dst}} {\mathsf{firewall}}))\\[1ex]
F \defeq & (\mathsf{src}=s_h \mathbin{;} (\modify v 0)\mathbin{;}(\modify{\mathsf{dst}} {\mathsf{cache}})) \parallel
(\mathsf{src}=s_l \mathbin{;} (\modify v 1)\mathbin{;}(\modify{\mathsf{dst}} {\mathsf{cache}})) \\
&
\parallel
(\mathsf{src}=l_1 \mathbin{;} (\modify r 0)\mathbin{;}(\modify{\mathsf{dst}} {\mathsf{loadb}}))
\parallel
(\mathsf{src}=h_1 \mathbin{;} (\modify r 1)\mathbin{;}(\modify{\mathsf{dst}} {\mathsf{loadb}}))\\[1ex]
L \defeq & ((\match r 1 )\mathbin{;} (\modify {\mathsf{dst}} s_h)\mathbin{;}\dup + (\match r 0 )\mathbin{;} (\modify {\mathsf{dst}} s_l)\mathbin{;}\dup) \\
&
\parallel \mathsf{src} = s_h \mathbin{;} (\modify {\mathsf{dst}} {\mathsf{firewall}}) \parallel
\mathsf{src} = s_l \mathbin{;} (\modify {\mathsf{dst}} {\mathsf{firewall}})
\end{array}
\]
\caption{Stateful firewall between high/low priority hosts and servers.}\label{fig:lb}
\end{figure}

The issue with the load balancer can be observed in the following example.
Take as input packets $\spadesuit$ and $\heartsuit$ with $\spadesuit(\mathsf{src})=\heartsuit(\mathsf{src})=l_1$.
After being processed at the cache, both packets arrive at the firewall.
One of the pairs in the semantics of the firewall $F$ is the following, with $\alpha$ unrestricted and $\beta(r)=0$:
$
(\alpha \rightarrow(\modify r 0)\rightarrow \beta)\cdot \{\heartsuit[\mathsf{loadb}/\mathsf{dst}],\spadesuit[\mathsf{loadb}/\mathsf{dst}]\}
$.
After processing by the load balancer, both packets are sent to $s_l$ simultaneously. To illustrate this event, we claim that there is a guarded pomset in the semantics of the load balancer. Observe that in the semantics of $L$ we find the following pomset, with $\alpha$ and $\beta$ from before (the second $\beta$ is the result of the $\match r 0$ in $L$):
$\alpha \rightarrow
(\modify r 0)\rightarrow  \beta\rightarrow \beta\rightarrow \{\heartsuit[s_l/\mathsf{dst}],\spadesuit[s_l/\mathsf{dst}]\}
$.
Using closure under contraction, we obtain a guarded pomset (the two $\beta$-nodes are merged into one) where both packets appear at $s_l$ at the same time.

A final issue stems from the fact that the firewall implementation is flawed as written. Specifically, it uses a global variable to determine whether a packet should be forwarded on to a high priority host. Of course, if another packet arrives before the current one has been forwarded, the value of this variable might change, resulting in both packets being forwarded to a low priority host.

The issue with the firewall can be observed as follows. Take as input two packets $\spadesuit$ and $\heartsuit$ with $\spadesuit(\mathsf{src})=s_h$ and $\heartsuit(\mathsf{src})=s_l$.
After processing by the load balancer, both packets end up in the firewall.
One of the pairs in the semantics of the firewall is the following, with $\alpha(v)=1$ and $\beta$ unrestricted:
$
(\alpha\cdot\modify v 0\parallel \beta\cdot\modify v 1) \cdot \{\heartsuit[\mathsf{cache}/\mathsf{dst}],\spadesuit[\mathsf{cache}/\mathsf{dst}]\}
$.
After processing by the cache, both packets are sent to $h_1$ or $l_1$. To illustrate how the packets travel to e.g. $l_1$, we find the following pomset in the semantics of $C$, with $\alpha, \beta$ from before and $\gamma(v)=0$:

\makebox[.9\linewidth]{\centering
\begin{tikzpicture}[node distance=0.35cm]
    \begin{scope}[every node/.style={anchor=center}]
    \node (g1) {$\alpha$};
    \node[right=of g1] (gamma1) {$(\modify v 0)$};
    \node[below=of gamma1] (gamma2) {$(\modify v 1)$};
        \node[left=of gamma2] (g2) {$\beta$};
    \node[yshift=-3.5mm, right=of gamma1] (gamma5) {$\gamma$};
        \node[right=of gamma5] (bla3) {$\{\heartsuit[l_1/\mathsf{dst}],\spadesuit[l_1/\mathsf{dst}]\}$};
    \end{scope}
    \path  (g1) edge[->] (gamma1);
    \path (g2) edge[->] (gamma2);
\path  (gamma1) edge[->] (gamma5);
\path (gamma2) edge[->] (gamma5);
\path (gamma5) edge[->] (bla3);
\end{tikzpicture}
}

This pomset subsumes a guarded pomset. Hence, by exchange closure, we find guarded pomsets in the behavior of $C$ where the packets both end up at $l_1$.

Overall, these examples show that \cnetkat can model subtle interactions between packets that arise in the presence of concurrency and state. Moreover, the axiomatic semantics can be used to prove (in)equivalences between programs. 

\begin{figure}[t]
\begin{minipage}{.49\textwidth}
\includegraphics[scale=.18]{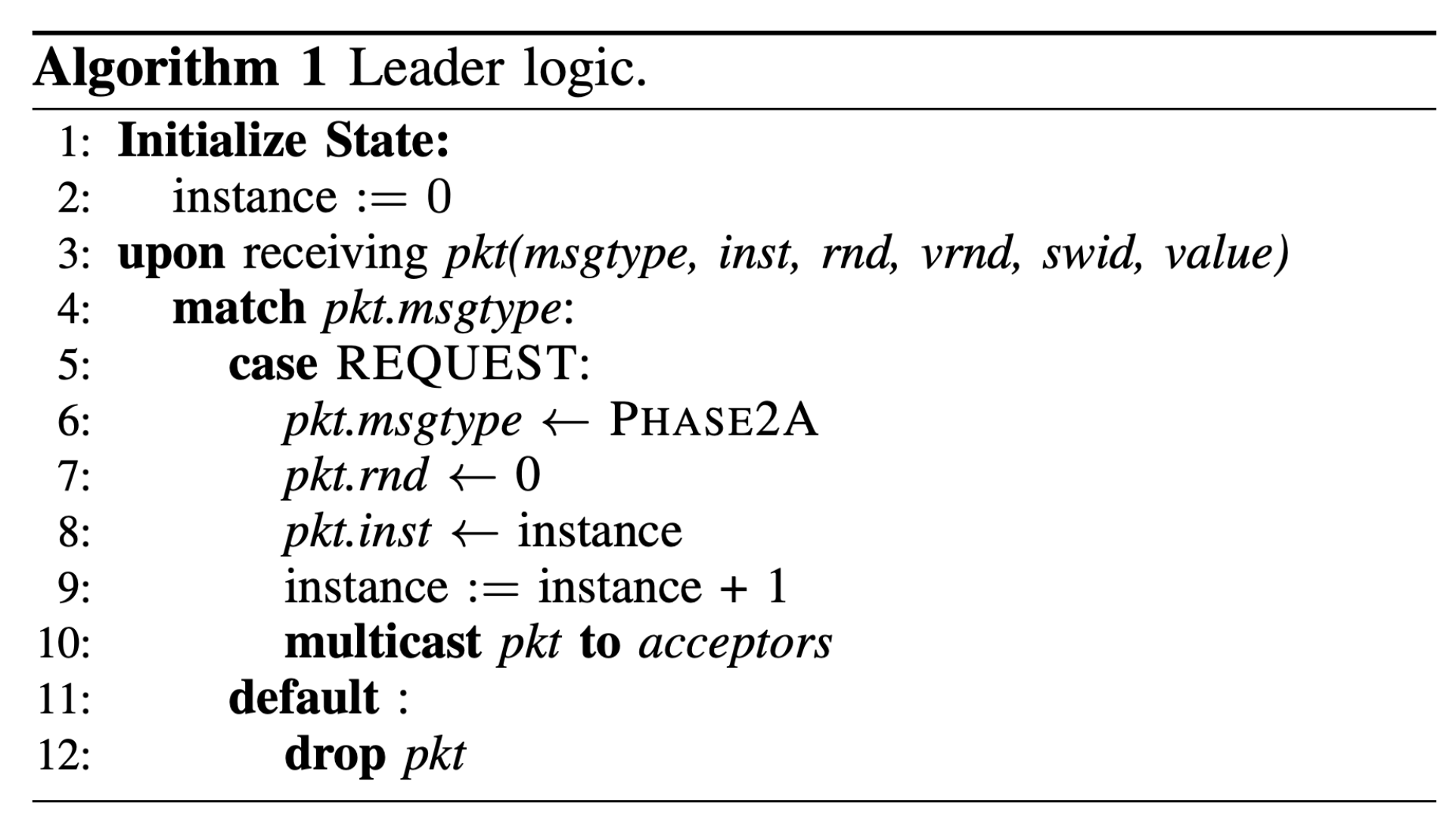}
\end{minipage}
\begin{minipage}{.49\textwidth}
\[\begin{array}{l}
\modify {\textrm{instance}} 0
\mathbin{;} \big( \\
\quad \textit{msgtype}=\textsc{Request}
 \mathbin{;} \\
\quad \modify {\textit{msgtype}} {\textsc{Phase2A}}
  \mathbin{;} \\
\quad \modify {\textit{rnd}} {0}
   \mathbin{;} \\
\quad \modify {\textit{inst}} {\textrm{instance}}
    \mathbin{;} \\
\quad \modify {\textrm{instance}} {\textrm{instance}+1}
    \mathbin{;} \\
\quad    \modify {\mathit{dst}} {1} \Vert \cdots \Vert  \modify {\mathit{dst}} {k}
    \\\big) + \drp
\end{array}
\]
\end{minipage}
\caption{Leader logic from~\cite{p4xos} and \cnetkat term, with $k$ acceptors.}\label{fig:p4xos}
\end{figure}

\section{Related Work}\label{sec:related}

The core of \cnetkat is two extensions of Kleene Algebra: \netkat~\cite{netkat,netkat2}, a networking extension of Kleene algebra with tests, and \POCKA~\cite{pocka}, a concurrent extension of \KA. \netkat describes how single packets move through a network, whereas \cnetkat can handle multiple packets. \POCKA was introduced to describe concurrent interactions of global variables, whereas \cnetkat makes use of this algebra to enable intra-packet communication. \cnetkat captures local and global state interactions which was not in any of the previous work. 

In the family of \KA extensions, \POCKA is closest to Concurrent Kleene algebra with Observations (\CKAO)~\cite{kao,fossacs2020}, which was proposed
to integrate concurrency with conditionals such as $\mathsf{if}$-statements
and $\mathsf{while}$-loops. Contrary to \CKAO, which uses a Boolean algebra
to axiomatize conditionals, \POCKA uses a pseudocomplemented distributive
lattice (PCDL) as the algebra for tests, which are referred to as
observations to mark the difference. The idea to use a PCDL as the algebra
for observations was first proposed in~\cite{jipsen-moshier-2016}.

Our work fits within the \CKA tradition, which gives a true concurrency semantics and is thereby distinct from bisimulation semantics typically considered in process algebras, such as CSP and CCS\@. Another distinction is that \cnetkat uses global state rather than message passing.

Some recently published work has also extended \netkat with constructs for modeling multi-packet behavior~\cite{dynetkat}. Here the goal is to model interactions between the control- and date-plane in dynamic updates. Parallel composition is axiomatized with a left-merge operator and a communication-merge operator, and semantics is in terms of bisimilarity instead of traces. The examples largely focus on the table updates, not on the flow of packets through the network.

The current paper deviates from earlier concurrent variations on \netkat, such as Concurrent NetCore~\cite{netcore} and a stateful variant of \netkat introduced in~\cite{events}.
Both have a different algebraic structure than \netkat. Concurrent NetCore does not have Kleene star, and does not provide a denotational semantics, or axiomatization. Moreover, it does not handle multiple packets, the use of $+$ in the language is multicast rather than non-determinism, and $\parallel$ is concurrent processing of disjoint fields of the same packet. Because of these restrictions, concurrent NetCore is less suitable to specify inter-packet concurrency.

The approach in~\cite{events} models interactions among multiple packets, but is accompanied by semantic correctness guarantees, rather than algebraic formalizations as in~\cnetkat. A recent PhD thesis~\cite{xiang-thesis} contains another version of stateful \netkat, which assumes packet processing can always be serialized into a deterministic, global order. This assumption enables a simpler semantics and a decision procedure, though completeness is left as an open problem. Flow control in~\cite{xiang-thesis} is handled in the style of Guarded Kleene Algebra with Tests~\cite{kozentseng2008,gkat}, which means that programs and specifications must be deterministic.

More broadly, there is a growing community doing research on network verification tools. Early work such as HSA~\cite{hsa}, Anteater~\cite{anteater}, Veriflow~\cite{veriflow}, Atomic Predicates~\cite{yang-lam}, etc.\ focused on stateless SDN data planes, while more recent work such as p4v~\cite{p4v} and VMN~\cite{vmn} supports richer models such as P4 and stateful middleboxes. These tools typically use analyses based on symbolic simulation or they encode verification tasks into first-order formulas that can be checked using  SMT solvers. To the best of our knowledge, \cnetkat is the first algebraic framework to model network-wide, multi-packet interaction with mutable state.

\section{Discussion}\label{sec:discussion}
We proposed \cnetkat, an algebraic framework to reason about programs with both local and global state, in the presence of parallel threads and control-flow statements. We provided a denotational semantics and a complete axiomatization. We also provided examples of how the language can be used to reason about stateful network programs and different sources of concurrency in a network. 

As a result of the algebraic approach, the semantics of a program arises from the semantics of its parts. This clashes with the idea of observational equivalence when concurrency comes into play: some behaviors of a program can only be observed when executed concurrently with another program, and not in isolation. Hence it becomes necessary to include some elements in the semantics that do not immediately correspond to observable behavior. This implies that observational equivalence is not the right notion for axiomatising the semantics. However, using the greatest congruence contained in a notion of observational equivalence is interesting; this guided us in the development of our axiomatisation but it remains to be shown that our axiomatisation is indeed the greatest congruence.

\cnetkat relies on a classic approach to proving program correctness: develop a framework can model both specifications and implementations, and show that equivalence is decidable. Past experience with \netkat suggests that this approach is usable, although \cnetkat lacks a procedure to check semantic equivalence, or at least membership of a given pomset. Devising an efficient procedure for this task is our immediate priority. The procedure will most likely rely on automata models such as fork automata~\cite{lodaya} or Petri automata~\cite{bps,brunet-pous}.
%

Ultimately, we would like to use \cnetkat to reason about stateful and distributed P4 programs. A target case study is provided in~\cite{p4xos}, which implemented Lamport's Paxos algorithm in the forwarding plane. To show correctness, the authors used a translation to Promela, a model checking language, and specify check that learners never decide on separate values for a single instance of consensus. This property is closely related to guarded pomsets. We would like to use \cnetkat to show correctness of the P4 implementation of the protocol directly (translation from the P4 code is almost direct, see \cref{fig:p4xos} for an example).

The reader will notice that the \cnetkat expression in \cref{fig:p4xos} uses an action of the form $\modify f v$, where $f$ is a field (\textit{inst}) and $v$ a global variable (instance). Adding actions of the converse form $\modify v f$ is trivial since the packet logic specifies that $f$ always has exactly one value. However, actions $\modify f v$ require more care: the value of global variables can only be determined at the end since parallel threads might change it while it is being copied. To accommodate this in the semantics, we will have to allow partially defined packet fields and determine the missing field values at the end (when we check for guarded traces).

Another exciting direction for future work is the development of a library of \emph{litmus tests} for networking in the spirit of~\cite{litmus}. Litmus tests are carefully crafted concurrent programs operating on shared memory locations that expose subtle bugs in memory models of hardware. One could imagine using the guarded pomsets semantics to discover minimal witnesses of undesired concurrent behavior.

We would also like to investigate the memory model of \cnetkat; this would give insight into the rules followed by operations on the global state. For a partial answer, we can look at \POCKA. 
The guarded fragment of the \POCKA semantics was shown to be \emph{sequentially consistent} (concurrent memory accesses behave as if they are executed sequentially~\cite{lampie}), as it passed the \emph{store buffering litmus test}~\cite{litmus}. The guarded fragment of the pomsets recording global variable changes is expected to pass this litmus test as well. It is worth investigating whether \cnetkat also supports other weak memory models, such as linearizability.

\paragraph{Acknowledgements}
N.~Foster and T.~Kapp\'{e} were partially supported by DARPA grant HR001120C0107 (Pronto).
T.~Kapp\'{e} also received funding from the European Union’s Horizon 2020 research and innovation programme under the Marie Sk\l{}odowska-Curie grant agreement No. 101027412 (VERLAN).
D.~Kozen was supported by NSF grant CCF-20008083.
A.~Silva was partially funded by ERC grant AutoProbe (101002697), EPSRC project CleVer (EP/S028641/1), and a Royal Society fellowship.

\bibliographystyle{splncs04}
\bibliography{refs}

\vfill

{\small\medskip\noindent{\bf Open Access} This chapter is licensed under the terms of the Creative Commons\break Attribution 4.0 International License (\url{http://creativecommons.org/licenses/by/4.0/}), which permits use, sharing, adaptation, distribution and reproduction in any medium or format, as long as you give appropriate credit to the original author(s) and the source, provide a link to the Creative Commons license and indicate if changes were made.} 

{\small \spaceskip .28em plus .1em minus .1em The images or other third party material in this chapter are included in the\break chapter's Creative Commons license, unless indicated otherwise in a credit line to the\break material.~If material is not included in the chapter's Creative Commons license and\break your intended use is not permitted by statutory regulation or exceeds the permitted\break use, you will need to obtain permission directly from the copyright holder.} 

\medskip\noindent\includegraphics{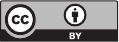} 

\ifarxiv%
\newpage
\appendix

\section{Proofs for Section~\ref{sec:preliminaries} (Pomsets and pomset languages)}%
\label{app:preliminaries}
\printProofs[preliminaries]

\section{Proofs for Section~\ref{sec:syntax} (Syntax and Semantics)}%
\label{app:syntax}
\printProofs[syntax]

\section{Proofs for Section~\ref{sec:relation} (Relation to \netkat/\POCKA)}%
\label{app:relation}
\printProofs[relation]


\section{Proofs for Section~\ref{sec:completeness} (Soundness and Completeness)}%
\label{app:completeness}
\printProofs[completeness]

\section{Analysis of Example}\label{sec:analysis}
We follow the same strategy used in~\cite{pocka}, where they identified a type of pomset called a \emph{guarded} pomset. It was demonstrated that guarded pomsets have seven characteristics of behaviors of (possibly concurrent) programs in isolation, and that if a pomset represents some execution of an isolated program, it must be guarded.
%
%
%
We need the following definitions from~\cite{pocka} for the proofs on guardedness.

First, we define the result of a state
after updating it for one value. Let $\ltr{a}\in\Act$ and $\alpha \in \State$.
We say
that $\alpha[\ltr{a}]$ \emph{exists} if $\ltr{a}=\modify v n$
for some $n \in \Val$ or $\ltr{a} = \modify v {v'}$ and $v' \in \dom(\alpha)$.
If $\alpha[\ltr{a}]$ exists, we define it for all $w\in \Var$ as follows:
\begin{mathpar}
\alpha[\modify v n](w) =
\begin{cases}
n & \text{if } w=v \\
\alpha(w) & \text{otherwise}
\end{cases}
\and
\alpha[\modify v {v'}](w) =
\begin{cases}
\alpha(v') & \text{if } w=v \\
\alpha(w) & \text{otherwise}
\end{cases}
\end{mathpar}
\noindent Second, we define a binary operator $\oplus$ on $\State$ to combine states. For $\alpha,\beta\in \State$:
\[
\alpha\oplus\beta =
\begin{cases}
  \alpha\cup\beta & \text{ if $\alpha(v)=\beta(v)$ for all }v \in \dom(\alpha) \cap \dom(\beta)\\
  \text{undefined} & \text{otherwise}
\end{cases}
\]

\begin{definition}
The set of \emph{guarded pomsets}, denoted $\G$, is the smallest set satisfying:
\begin{mathpar}
\inferrule{%
    \alpha \in \State
}{%
    \alpha \in \G
}
\and
\inferrule{%
    \alpha \in \State \\
    \ltr{a} \in \Act \\
    \alpha[\ltr{a}] \text{ exists}
}{%
    \alpha \cdot \ltr{a} \cdot \alpha[\ltr{a}] \in \G
}
\and
\inferrule{%
    U \cdot \alpha, \alpha \cdot V \in \G \\
    \alpha \in \State
}{%
    U \cdot \alpha \cdot V \in \G
}
\and
\inferrule{%
    \alpha \cdot U \cdot \beta \\
    \gamma \cdot V \cdot \delta \in \G \\
    \alpha\oplus\gamma \text{ defined } \\
    \beta\oplus\delta \text{ defined } \\
    \alpha, \beta, \gamma, \delta \in \State
}{%
    \alpha\oplus\gamma \cdot (U \parallel V) \cdot \beta\oplus\delta \in \G
}
\end{mathpar}
\end{definition}

Guardedness in pomsets can be characterized by the conjunction of seven properties~\cite[Theorem 5.9]{pocka}. We only need two of those properties in the proofs that follow, which we will write out below.

A path for a variable $v$ from a state-node $u$ to another state-node $s$ is a chain such that the changes in the value of $v$ between $u$ and $s$ are explained by the actions between them and recorded in all the states between $u$ and $s$.

\begin{definition}[Path]\label{def:path}
Let $\upom \in \Pom(\Act\cup \State)$ and $u_1,u_2\in S_{\upom}$ such that $u_1\leq_{\upom} u_2$ and $\lambda_{\upom}(u_1),\lambda_{\upom}(u_1)\in \State$.
We say a \emph{path} $p_v$ from $u_1$ to $u_2$ for variable $v\in \Var$ is a sequence of nodes $q_1, a_1, \dots, a_n, q_{n+1} \in S_{\upom}$ that satisfy the following conditions:
\begin{enumerate}[label={(P\arabic*)},leftmargin=1cm]
  \item\label{item:acties}
  For all $1 \leq i \leq n$, we have $\lambda_{\upom}(a_i)\in \Act$ and
  $u_1\leq_{\upom} a_i \leq_{\upom} u_2$ for all $i$.
  Additionally we require that $a_i\leq_{\upom}a_{i+1}$ for $1\leq i <n$.

  \item\label{item:guards}
  For all $1 \leq i \leq n + 1$ it holds that $\lambda_{\upom}(q_i)\in \State$, and for all $1 \leq i \leq n$, the predecessor of $a_i$ is $q_{i}$, and the successor of $a_i$ is $q_{i+1}$.
  Additionally we have that $\lambda_{\upom}(q_1)=\lambda_{\upom}(u_1)$, $v\in\dom(\lambda_{\upom}(u_1))$ and
  $\lambda_{\upom}(q_{n+1})=\lambda_{\upom}(u_2)$.
  Lastly, for $1\leq  i \leq n$ we have:
  \[
  \lambda_{\upom}(q_{i+1})(v) =
  \begin{cases}
  n & \lambda_{\upom}(a_{i})= \modify v n \text{ for some }n\in \Val\\
  \lambda_{\upom}(q_{i})(v') & \lambda_{\upom}(a_{i})= \modify v {v'}\text{ for some }v'\in\dom(\lambda_{\upom}(q_{i}))\\
  \lambda_{\upom}(q_{i})(v) & \text{otherwise }
  \end{cases}
  \]
\end{enumerate}
\end{definition}

Property $(A5)$:
\begin{enumerate}[label={(A\arabic*)},leftmargin=1cm,start=5]
    \item If $u\in S_\lp{u}$ such that $\lambda_\lp{u}(u)=\modify v n$ for some $v\in \Var$ and $n\in \Val$, we require that the successor of $u$ is $s$ s.t. $\lambda_\lp{u}(s)(v)=n$.
\end{enumerate}

Property $(A7)$:
\begin{enumerate}[label={(A\arabic*)},leftmargin=1cm,start=7]
    \item Let $u\in S_\lp{u}$ be a state-node. Then for all $v\in\dom(\lambda_\lp{u}(u))$, there exists a path for $v$ from $s\in S_\lp{u}$ to $u$
    such that either $v\in \dom(\lambda_\lp{u}(s))$ and $s=*_{\min}$ or $s$ is the successor of an assignment-node with label $\modify v k$ with $k\in \Var\cup \Val$.
\end{enumerate}

Lastly, we need the following definition and lemma.
\begin{definition}[Bottleneck]\label{def:bottleneck}
    Let $\upom \in \Pom(\Act\cup \State)$ and $u_0,u_1,u_2\in S_{\upom}$.
We say $u_1$ is a \emph{bottleneck between $u_0$ and $u_2$} if $u_0\leq_{\upom} u_1 \leq_{\upom} u_2$ and
for all $u_3\in S_{\upom}$ s.t. $u_0\leq_{\upom} u_3$ we have $u_1\leq_{\upom} u_3$ or $u_3\leq_{\upom} u_1$.
\end{definition}

We use the following result in the proofs below (\cite[Lemma C.3]{pocka}):
\begin{lemma}\label{lemma:bottleneck}
    Let $\upom\in \Pom(\Act\cup \State)$ and $u_1,u_2\in S_{\upom}$ s.t.\ $u_1\leq_{\upom} u_2$.
    If there exists a path $p_v$ from $u_1$ to $u_2$, and a bottleneck $u_3$ between them, then the bottleneck is on $p_v$.
\end{lemma}

We now return to our running example. In order to identify the isolated behaviors of $p$, we thus have to filter out the pairs where the state pomset
is guarded. Guarded pomsets are defined specifically for pomsets whose nodes are labeled with state observations and
state modifications, and our pomsets also have nodes labeled with elements from $2^{\Pk}$. When deciding whether a behavior is guarded, we simply study the pomset of a behavior, and then in particular the nodes labeled with state observations and state modifications, and see whether they form a guarded pomset. If they do, we call the original behavior guarded.

\begin{definition}[Guarded pairs]
Let $\upom \in \Pom ( \State \cup \Act\cup 2^{\Pk})$ and $a\in 2^{\Pk}$. We call a pair $\upom\cdot a$
\emph{guarded} if the pomset $\vpom$ with $S_{\vpom}=\{s \mid s\in S_{\upom}, \lambda_{\upom}\in\State\cup\Act\}$, $\lambda_{\vpom}(u)=\lambda_{\upom}(u)$ and
$\leq_{\vpom}=\leq_{\upom}\upharpoonright{S_{\vpom}}$ is guarded according to~\cite[Definition 5.1]{pocka}.
\end{definition}

In order to show that all guarded behaviors in the semantics of the running example record the $\heartsuit $ packets at switch $3$ before they record the $\spadesuit$ packets at switch $2$, we first show that all state pomsets of the pairs in the semantics of $p$ have a certain property $P$. We then claim that if a pomset has this property, and is guarded, it must be such that the $\heartsuit $ packets are recorded at switch $3$ before the $\spadesuit$ packets
are recored at switch $2$.

We first look at the semantics of the running example before closure, after which we define property $P$.
We are interested in the behavior of the program when the packets have reached their final destination (switch $4$).
Hence, we add a test $(\match \sw 4)$ to the end, to ensure the packets have arrived at switch $4$:
\[
p \defeq (\modify v 0 )\mathbin{;}{(p_1\parallel p_2\parallel p_3\parallel p_4)}^* \mathbin{;} (\match \sw 4)
\]

If we input packets $\{\heartsuit,\spadesuit\}$ at switch $1$, after one iteration of the Kleene star (before closure of the semantics), we get the following distirbution of packets when we multicast according to $p$:
\[
\begin{array} {cc ccc c c}
 \{\heartsuit,\spadesuit\} && \emptyset  && \emptyset &&\emptyset\\[2ex]
 (p_1) &\parallel& (p_2) &\parallel& ( p_3) &\parallel&  ( p_4)
 \end{array}
\]

In terms of packets, the output looks like $\{\spadesuit[2/\sw],\heartsuit[3/\sw]\}$. In terms of global state pomset, the output may look like the following, with $\beta(v)=1$:

\begin{tikzpicture}[node distance=0.25cm]
    \begin{scope}[every node/.style={anchor=center}]
    \node (gamma1) {$(\modify v 0)$};
    \node[right=of gamma1] (bla2) {$\{\heartsuit,\spadesuit\}$};
    \node[yshift=4mm,right=of bla2] (gamma2) {$\beta$};
    \node[yshift=-4mm,right=of bla2] (gamma5) {$\{\heartsuit\}$};
        \node[right=of gamma2] (bla3) {$\{\spadesuit\}$};
          \node[right=of gamma5] (bla4) {$\{\heartsuit[3/\sw]\}$};
        \node[right=of bla3] (bla5) {$\spadesuit[2/\sw]$};
        \node[right=of bla4] (bla6) {$(\modify v 1)$};
    \end{scope}
    \draw[->] (gamma1) edge (bla2);
      \draw[->] (bla2) edge (gamma2);
      \draw[->] (bla2) edge (gamma5);
      \draw[->] (gamma2) edge (bla3);
      \draw[->] (gamma5) edge (bla4);
        \draw[->] (bla3) edge (bla5);
            \draw[->] (bla4) edge (bla6);
\end{tikzpicture}

The other state pomsets in the semantics of $(\modify v 0 )\mathbin{;}(p_1\parallel p_2\parallel p_3\parallel p_4)$  are pomsets with the same nodes and ordering as the one above but with extra $\State^*$-nodes around state observations and assignments.

In the next iteration of the Kleene star we obtain the output set of packets
$\{\spadesuit[4/\sw],\heartsuit[4/\sw]\}$, and the corresponding global state pomsets may look like this, again with $\beta(v)=1$:
\[
    \begin{tikzpicture}[node distance=0.4cm]
        \begin{scope}[every node/.style={anchor=center}]
            \node (gamma1) {$(\modify v 0)$};
            \node[right=3mm of gamma1] (bla2) {$\{\heartsuit,\spadesuit\}$};
            \node[yshift=4mm,right=of bla2] (gamma2) {$\beta$};
            \node[below=2mm of gamma2] (gamma5) {$\{\heartsuit\}$};
            \node[right=of gamma2] (bla3) {$\{\spadesuit\}$};
            \node[right=3mm of gamma5] (bla4) {$\{\heartsuit[3/\sw]\}$};
            \node[right=12mm of bla3] (bla5) {$\spadesuit[2/\sw]$};
            \node[below=2mm of bla5] (bla6) {$(\modify v 1)$};
            \node[right=of bla5] (bla7) {$\{\spadesuit[2/\sw]\}$};
            \node[below=2mm of bla7] (bla8) {$\{\heartsuit[3/\sw]\}$};
            \node[right=3mm of bla7] (bla9) {$\{\spadesuit[4/\sw]\}$};
            \node[right=3mm of bla8] (bla10) {$\{\heartsuit[4/\sw]\}$};
        \end{scope}
        \path (gamma1) edge[->] (bla2);
        \path (bla2) edge[->] (gamma2);
        \path (bla2) edge[->] (gamma5);
        \path (gamma2) edge[->] (bla3);
        \path (gamma5) edge[->] (bla4);
        \path (bla3) edge[->] (bla5);
        \path (bla4) edge[->] (bla6);
        \path (bla5) edge[->] (bla7);
        \path (bla6) edge[->] (bla8);
        \path (bla7) edge[->] (bla9);
        \path (bla8) edge[->] (bla10);
        \path (bla6) edge[->] (bla7);
        \path (bla5) edge[->] (bla8);
    \end{tikzpicture}
\]
In the iteration after that, the packets remain at switch $4$ (with an output set of $\{\spadesuit[4/\sw],\heartsuit[4/\sw]\}$) and the corresponding global state pomsets get an extra node labeled with $\{\spadesuit[4/\sw],\heartsuit[4/\sw]\}$ sequentially added to the end.
In any further iterations, the output packets stay the same (they never leave switch $4$), and the state pomset gets extended with $\{\heartsuit[4/\sw],\spadesuit[4/\sw]\}$ by the $\dup$ of $p_4$.

We now define property $P$, which contains some characteristics of the global state pomset in each pair in the semantics of $p$ that we can use later to show that all guarded pomsets in the semantics of the running example record the $\heartsuit $ packets at switch $3$ before they record the $\spadesuit$ packets at switch $2$.

 \begin{definition}[Pomset Property $P$]\label{def:propp}
   Let $\sw \in \Field$, $v\in \Var$ and $0,\dots,4\in \Val$.
A pomset $\upom$ has property $P$, denoted $P(\upom)$, if there exist $u_1,\dots,u_5\in S_{\upom}$ s.t.
   \begin{enumerate}
     \setlength\itemsep{0em}
     \item\label{item:existence-in-u} the following conditions hold:
     \begin{mathpar}\lambda_{\upom}(u_1)=(\modify v 0) \and \lambda_{\upom}(u_2)=(\modify v 1) \and \lambda_{\upom}(u_3)=\beta\wedge\beta(v)=1 \and \lambda_{\upom}(u_4)=\{\heartsuit[3/\sw]\} \and
     \lambda_{\upom}(u_5)=\{\spadesuit[2/\sw]\}
     \and
 u_1 \leq_{\upom} u_3 \leq_{\upom} u_5
       \and u_1 \leq_{\upom} u_4 \leq_{\upom} u_2
     \end{mathpar}
     Graphically, we can represent these conditions as the following diagram:

     \hfill
     \begin{tikzpicture}[xscale=1.4,yscale=.4]
       \node(u1)at (0,1) {$u_1: (\modify v 0)$};
       \node(u3)at (1.8,2) {$u_3:\beta(v)=1$};
       \node(u4)at (2.05,0) {$u_4:\{\heartsuit[3/\sw]\}$};
       \node(u5)at (4.2,2) {$u_5:\{\spadesuit[2/\sw]\}$};
       \node(u2)at (4.05,0) {$u_2:(\modify v 1)$};
       \draw[->] (u1) to (u3);
       \draw[->] (u1) to (u4);
       \draw[->] (u3) to (u5);
       \draw[->] (u4) to (u2);
     \end{tikzpicture}
     \hfill$ $

     \item\label{item:relative-existence} For all nodes $z\in S_{\upom}$ we have the following conditions.
     \begin{mathpar}
       \forall z. \lambda_{\upom}(z)=(\modify v k) \Rightarrow z = u_2 \vee z = u_1
       \and
       S_{\upom}.z\leq_{\upom} u_1 \vee u_1\leq_{\upom} z
       \and
       \exists\pk\in \lambda_{\upom} (z)\in 2^{\Pk}. \pk(\sw)=3\wedge\pk(\type)=\heartsuit \Rightarrow u_4\leq_{\upom} z
       \and
          \exists\pk\in \lambda_{\upom} (z)\in 2^{\Pk}. \pk(\sw)=2\wedge\pk(\type)=\spadesuit \Rightarrow u_5\leq_{\upom} z
   \end{mathpar}
  \end{enumerate}
\end{definition}

The property $P$ describes global asserts and modifications and sets of packets found in the running example, and their relative ordering. The condition $\forall z. \lambda_{\upom}(z)=(\modify v k) \Rightarrow z = u_2\vee z=u_1$ entails that there are only two nodes in the pomset labeled with an action that modifies $v$. The condition $\forall z\in S_{\upom}.z\leq_{\upom} u_1 \vee
u_1\leq_{\upom} z$ implies that $u_1$ is always on any sequence of nodes
between the minimal node of $\upom$ and $u_3$. The last two conditions entail that $u_4$ and $u_5$ are the first times in the execution that respectively $\heartsuit$ packets are present at switch $3$ and $\spadesuit$ packets are present at switch $2$.

\begin{lemma}\label{lemma:order}
Let $\upom$ be a pomset with $P(\upom)$. If $\upom$ is guarded, then $u_2 \leq_{\upom} u_3$.
\end{lemma}
\begin{proof}
  We use characteristics $(A5)$ and $(A7)$ of guarded pomsets, which $\upom$ satisfies as we assume it is guarded. We take node $u_3$, which is a node with state label $\beta$ such that $\beta(v)=1$. According to $(A7)$, this then means that there exists a path for $v$ from the minimal node of the pomset, let us denote it with $*_{\min}$, to $u_3$ or there exists a path for $v$ from a node $s$ to $u_3$ and $s$ is the successor of an assignment-node with label $\modify v k$ for some $k\in \Var\cup\Val$.
  In the former case, we use \cref{lemma:bottleneck} to conclude that $u_1$ is
  on the path for $v$ from $*_{\min}$ to $u_3$. Then, via $(A5)$, we obtain that $u_1$ has a successor node $t$ such that $\lambda_{\upom}(t)=\alpha$ with $\alpha(v)=0$. By definition of successor, this means that $u_1\leq_{\upom} t\leq_{\upom} u_3$, and via the properties of
  \cref{def:path} there must exist at least one node $q$ labeled with $\modify v k$ such that $t\leq_{\upom } q \leq_{\upom}
  u_3$ altering the value of $v$, as the path must explain how the value of $v$ changed from $0$ to $1$. Via item~\ref{item:relative-existence} of property $P$, we then know that $q=u_2$ ($q$ cannot be $u_1$ as that would imply that
  $u_1\leq_{\upom} t \leq_{\upom} u_1$ which is a contradiction as $u_1$ and $t$ do not have the same labels), and thus $u_2\leq_{\upom} u_3$.
  In the latter case, $s$ has to be the successor of $u_1$ or $u_2$ via item~\ref{item:relative-existence} of property $P$.
In the former case, we obtain that $u_1 \leq_{\upom} s$, and again via property $(A5)$ we get that $\lambda_{\upom}(s)=\alpha$ and $\alpha(v)=0$. For there to be a path from $s$ to $u_3$ for $v$, we obtain in similar fashion as before that $u_2\leq_{\upom} u_3$. In the latter case, thus if $s$ is the successor of $u_2$, we get that  $u_2\leq_{\upom }s$.
  For there to be a path from $s$
  to $u_3$ we need $s\leq_{\upom } u_3$. This implies then that $u_2\leq_{\upom} u_3$.
\end{proof}

Note that this entails via transitivity that $u_4 \leq_{\upom} u_5$, and thus that the $\heartsuit$ packets arrived at switch $3$ before the $\spadesuit$ packets arrived at switch $2$.

We return to our running example and study its semantics closed under contraction and subsumption.
We show that property $P$ is preserved:

\begin{lemma}\label{lemma:prop}
  Let $a\in 2^{\Pk}_\nempty$. For all $\upom\cdot b\in \closure[]{\sem{p}}(a)$, it holds that $P(\upom)$.
\end{lemma}
\begin{proof}
  It is clear that for all $\upom$ such that $\upom\cdot b\in \sem{p}(a)$, we have $P(\upom)$. For $\upom\cdot b\in \closure[]{\sem{p}}(a)$, we know that $\vpom\cdot b \in \sem{p}(a)$ and $\upom\in
  \closure[\hexch\cup\hcontr]{\{\vpom\}}$. Via \cref{closure-order} and the definition of closure under $\hexch$ and $\hcontr$, we can conclude that there
  exists a pomset $\wpom$ such that $\upom\preceq \wpom\sqsubseteq \vpom$. We also know that $P(\vpom)$. We now show that then also $P(\wpom)$ and $P(\upom)$. From the definition of $\sqsubseteq$ we get that there exists a bijective pomset morhphism $h$ from $\vpom$ to $\wpom$. Thus, $h$ is a
  bijective function from $S_{\vpom}$ to $S_{\wpom}$ such that $\lambda_{\wpom}\circ h =\lambda_{\vpom}$ and if $u\leq_{\vpom} u'$ then $h(u)\leq_{\wpom}h(u')$.
  Now we need to verify the properties of \cref{def:propp}.
  \begin{enumerate}
    \item The existence of nodes with certain labels in $\wpom$ follows immediately from $\lambda_{\wpom}\circ h =\lambda_{\vpom}$. Their relative ordering is also immediately satisfied.
    \item Take a $z\in S_{\wpom}$ such that $\lambda_{\wpom}(z)=(\modify v k)$. Then, because $h$ is surjective, there exists $y\in S_{\vpom}$ such that $h(y)=z$.
    Thus $\lambda_{\wpom} \circ h (y) = \lambda_{\vpom} (y)=\modify v k
    $. As $\vpom$ has property $P$, this means that $y=u_2$ or $y=u_1$, and then, because $h$ is a function, we get $h(y)=h(u_2)=z$ or $h(y)=h(u_1)=z$.
    \item Take $z\in S_{\wpom}$. We need to show that $z\leq_{\wpom}h(u_1)$ or $h(u_1)\leq_{\wpom} z$. Then, because $h$ is surjective, there exists $y\in S_{\vpom}$ such that $h(y)=z$. As $\vpom$ has property $P$, this means that
    $y\leq_{\vpom} u_1$ or $u_1\leq_{\vpom} y$. Then immediately $h(y)\leq_{\wpom}h(u_1)$ or $h(u_1)\leq_{\wpom} h(y)$ holds, and together with $h(y)=z$ this gives the required result.
    \item Take a $z\in S_{\wpom}$ such that $\pk\in \lambda_{\wpom}(z)\in 2^{\Pk}$ and $\pk(\sw)=3$ and $\pk(\type)=\heartsuit$. Then, because $h$ is surjective, there exists $y\in S_{\vpom}$ such that $h(y)=z$ and
     $\lambda_{\wpom} \circ h (y) = \lambda_{\vpom} (y)    $. As $\vpom$ has property $P$, this means that $u_4\leq_{\vpom}y$, and thus that $h(u_4)\leq_{\wpom} h(y)=z$.
    \item The last condition is verified in a manner symmetrical to the case above.
  \end{enumerate}
  This demonstrates that $\wpom$ has property $P$. In a similar manner, we can verify that $\upom$ with $\upom\preceq \wpom$ also has property $P$.
\end{proof}

\begin{corollary}
  Let $a\in 2^{\Pk}_\nempty$. For all $\upom\cdot b \in \closure[]{\sem{s}}(a)$, if
  $\upom\cdot b$ is a guarded behavior, then $\heartsuit$ packets are observed at switch $3$ before $\spadesuit$
  packets are observed at switch~$2$.
\end{corollary}
\begin{proof}
Via \cref{lemma:prop}, we know that $P(\upom)$. Then, via \cref{lemma:order}, we obtain that $u_2\leq_{\upom} u_3$, and thus that $u_4\leq_{\upom}u_5$. As $\lambda_{\upom}(u_4)=\{\heartsuit[3/\sw]\}$ and $
\lambda_{\upom}(u_5)=\{\spadesuit[2/\sw]\}$, and both of these nodes are the first occurrences of $\heartsuit$ packets at switch $3$ and $\spadesuit$ packets at switch $2$, this proves the claim.
\end{proof}


\fi%

\end{document}